\documentclass[journal]{resources/IEEEtran_new}
\IEEEoverridecommandlockouts
\usepackage{graphicx}
\graphicspath{{images/}}

\usepackage{amsthm, amsmath, amsfonts, amssymb}
\usepackage{enumerate}
\usepackage{graphicx}
\usepackage{mathtools}
\usepackage{thmtools}
\usepackage{thm-restate}
\usepackage{cleveref}
\usepackage{subfigure}
\usepackage{resources/custom_commands}
\usepackage{resources/coloredboxes}
\usepackage{float}
\usepackage{enumerate}
\usepackage{marginnote}

\usepackage{autonum}

\usepackage{scalerel,stackengine}
\stackMath
\newcommand\reallywidecheck[1]{%
\savestack{\tmpbox}{\stretchto{%
  \scaleto{%
    \scalerel*[\widthof{\ensuremath{#1}}]{\kern-.6pt\bigwedge\kern-.6pt}%
    {\rule[-\textheight/2]{1ex}{\textheight}}%WIDTH-LIMITED BIG WEDGE
  }{\textheight}% 
}{0.5ex}}%
\stackon[1pt]{#1}{\scalebox{-1}{\tmpbox}}%
}

\usepackage{mathabx}

\newtheorem*{theorem*}{Theorem}

\declaretheoremstyle[
spaceabove=\topsep, spacebelow=\topsep,
headfont=\normalfont\bfseries,
notefont=\bfseries, notebraces={}{},
bodyfont=\normalfont\itshape,
postheadspace=0.5em,
name={\ignorespaces},
numbered=no,
headpunct=.]
{mystyle}
\declaretheorem[style=mystyle]{namedthm*}

\newcommand{\dec}{\mathrm{Dec}}

\DeclareMathOperator{\kl}{{\scriptscriptstyle KL}}
\DeclareMathOperator{\ID}{{\scriptscriptstyle ID}}
\DeclareMathOperator{\SQ}{{\scriptscriptstyle SQ}}

\newcommand{\idloss}{L_{\ID}}
\newcommand{\sqloss}{L_{\SQ}}

\newcommand{\Btdis}{\mathrm{Beta}}

\newcommand{\eqlinebreakshort}{\ensuremath{\nonumber \\ & \quad \quad}}
\newcommand{\eqstartshort}{\ensuremath{&}}
\newcommand{\eqstartnonumshort}{\ensuremath{& \nonumber}}
\newcommand{\eqbreakshort}{\ensuremath{ \\}}

\newcommand{\arcsinh}{\ensuremath{\mathrm{ArcSinh}}}

\def\BibTeX{{\rm B\kern-.05em{\sc i\kern-.025em b}\kern-.08em
    T\kern-.1667em\lower.7ex\hbox{E}\kern-.125emX}}

\usepackage{blindtext}

\newcommand{\dpfconst}{24}

\newcommand{\shift}[2]{\ensuremath{T_{#2}(#1)}}

\definecolor{olivedrab}{rgb}{0.42, 0.56, 0.14}

\definecolor{palatinatepurple}{rgb}{0.41, 0.16, 0.38}

\definecolor{princetonorange}{rgb}{1.0, 0.56, 0.0}

\usepackage{subfiles} % Best loaded last in the preamble

\newif\iflong

\longtrue
\newcommand{\az}{\ensuremath{K}}

\newcommand{\newvar}{\ensuremath{w}}
\newcommand{\newVar}{\ensuremath{W}}

\newcommand{\normvar}{\ensuremath{z}}
\newcommand{\bnormvar}{\ensuremath{\boldsymbol{z}}}
\newcommand{\normVar}{\ensuremath{Z}}
\newcommand{\bnormVar}{\ensuremath{\boldsymbol{Z}}}
\newcommand{\normset}{\ensuremath{\mathcal{Z}}}
\newcommand{\rawvar}{\ensuremath{y}}
\newcommand{\brawvar}{\ensuremath{\boldsymbol{y}}}

\newcommand{\comp}{\ensuremath{f}}
\newcommand{\compder}{\ensuremath{f'}}
\newcommand{\compset}{\ensuremath{\cF}}

\newcommand{\locloss}{\ensuremath{g}}

\newcommand{\rawloss}{\ensuremath{\widetilde{\cL}}}
\newcommand{\singleloss}{\ensuremath{\widetilde{L}}}

\title{Efficient Representation of Large-Alphabet Probability Distributions% via Arcsinh-Compander 
{}\thanks{This work was supported in part by the NSF grant CCF-2131115 and sponsored by the United States Air Force Research Laboratory and the United States Air Force Artificial Intelligence Accelerator and was accomplished under Cooperative Agreement Number FA8750-19-2-1000. The views and conclusions contained in this document are those of the authors and should not be interpreted as representing the official policies, either expressed or implied, of the United States Air Force or the U.S. Government. The U.S. Government is authorized to reproduce and distribute reprints for Government purposes notwithstanding any copyright notation herein.
\indent This paper has supplementary downloadable material available at http://ieeexplore.ieee.org, provided by the authors. The material includes the appendices. Contact adlera@mit.edu, jstang@mit.edu, and yp@mit.edu for further questions about this work.
}}

\author{\IEEEauthorblockN{Aviv Adler, Jennifer Tang, Yury Polyanskiy} \\
MIT EECS Department, Cambridge, MA, USA \\
adlera@mit.edu, jstang@mit.edu, yp@mit.edu
}

\date{\today}

\begin{document}

\maketitle
%%%%%%%%%% show page number, remove for submission
%\thispagestyle{plain}
%\pagestyle{plain}

\begin{abstract}
A number of engineering and scientific problems require representing and manipulating probability
distributions over large alphabets, which we may think of as long vectors of reals summing to $1$.
In some cases it is required to represent such a vector with only $b$ bits per entry. A natural
choice is to partition the interval $[0,1]$ into $2^b$ uniform bins and quantize entries to each
bin independently. We show that a minor modification of this procedure -- applying an entrywise non-linear
function (compander) $f(x)$ prior to quantization -- yields an extremely effective quantization method. For example, for $b=8
(16)$ and $10^5$-sized alphabets, the quality of representation improves from a loss (under KL
divergence) of $0.5 (0.1)$ bits/entry to $10^{-4} (10^{-9})$ bits/entry. Compared to floating point representations, our compander method improves the loss from $10^{-1}(10^{-6})$ to $10^{-4}(10^{-9})$ bits/entry. These numbers hold for
both real-world data (word frequencies in books and DNA $k$-mer counts) and for synthetic randomly
generated distributions. Theoretically, we analyze a minimax optimality criterion and show
that the closed-form compander $f(x) ~\propto~ \arcsinh(\sqrt{c_\az (\az \log \az) x})$ is (asymptotically as $b\to\infty$) optimal for quantizing probability distributions over a $\az$-letter alphabet. Non-asymptotically, such a compander (substituting $1/2$ for $c_\az$ for simplicity) has KL-quantization loss bounded by $\leq 8\cdot 2^{-2b} \log^2 \az$. 
Interestingly, a similar minimax criterion for the quadratic loss on the hypercube shows optimality of the standard uniform quantizer. This suggests that the $\arcsinh$ quantizer is as fundamental for KL-distortion as the uniform quantizer for quadratic distortion.

\end{abstract}

%\section{Introduction}

%\annotate{YURY WRITES HERE.}

%\subsubsection{Organization} 

%We give background and related works in \Cref{sec::previous_works}. Due to space constraints, we give the sketch for proving \Cref{thm::asymptotic-normalized-expdiv} in \Cref{sec::asymt_single}. The sketch for \Cref{thm::minimax_compander} and \Cref {cor::worstcase_prior} are given in \Cref{sec::minimax}. We omit the proof for \Cref{thm::worstcase_power_minimax}.

\vspace{-0.5pc}

\section{Compander Basics and Definitions}

%\nbyp{General comment: think of footnotes as warts. If you have a few, it's not noticeable, but if there are many, that's making the reader shiver. You are not generally using footnotes properly and most of them should just be moved to the main text. A proper usage of a footnote is to discuss a tangent that is almost entirely useless for the reader who is interested in your main Theorem. Let's agree to reduce the number of footnotes by a factor of 2 as a simple goal.} -- \textcolor{blue}{Sure.}

Consider the problem of \emph{quantizing} the probability simplex $\triangle_{\az-1} = \{\bx \in \bbR^\az : \bx \geq \bzero, \sum_i x_i = 1 \}$ of alphabet size $\az$,\footnote{While the alphabet has $\az$ letters, $\triangle_{\az-1}$ is $(\az-1)$-dimensional due to the constraint that the entries sum to $1$.} i.e. of finding a finite subset $\normset \subseteq \triangle_{\az-1}$ to represent the entire simplex.  Each $\bx \in \triangle_{\az-1}$ is associated with some $\bnormvar = \bnormvar(\bx) \in \normset$, and the objective is to find a set $\normset$ and an assignment such that the difference between the values $\bx \in \triangle_{\az-1}$ and their representations $\bnormvar \in \normset$ are minimized; while this can be made arbitrarily small by making $\normset$ arbitrarily large, the goal is to do this efficiently for any given fixed size $|\normset| = M$. Since $\bx, \bnormvar \in \triangle_{\az-1}$, they both represent probability distributions over a size-$\az$ alphabet. Hence, a natural way to measure the quality of the quantization is to use the KL (Kullback-Leibler) divergence $D_{\kl}(\bx \| \bnormvar)$, which corresponds to the excess code length for lossless compression and is commonly used as a way to compare probability distributions. (Note that we want to minimize the KL divergence.)

While one can consider how to best represent the vector $\bx$ as a whole, in this paper we consider only \emph{scalar quantization} methods in which each element $x_j$ of $\bx$ is handled separately, since we showed in \cite{adler_ratedistortion_2021} that for Dirichlet priors on the simplex, methods using scalar quantization perform nearly as well as optimal vector quantization. Scalar quantization is also typically simpler and faster to use, and can be parallelized easily. 
%Given $x_j$, the quantizer \emph{encodes} it to an integer $n \in [N] = \{1, 2, \dots, N\}$ ($N$ is the number of quantization levels, or \emph{granularity}) which can be easily stored; when needed, the quantizer \emph{decodes} it to a value $y^{(n)}$ to represent it.
Our scalar quantizer is based on \emph{companders} (portmanteau of `compressor' and `expander'), a simple, powerful and flexible technique first explored by Bennett in 1948 \cite{bennett1948} in which the value $x_j$ is passed through a nonlinear function $f$ before being uniformly quantized. We discuss the background in greater depth in \Cref{sec::previous_works}.
%in particular to compare certain previously-known results to \Cref{thm::asymptotic-normalized-expdiv}. %\jennifer{Maybe add that background is in section blah}

In what follows, $\log$ is always base-$e$ unless otherwise specified. We denote $[N] := \{1,\dots, N\}$.

%Since optimal vector quantizations generally carry a large overhead and use difficult computations, a simpler choice is \emph{scalar quantization}, where each element of the vector is quantized separately, which is both much simpler and is easy to parallelize.

\subsubsection{Encoding}

Companders require two things: a monotonically increasing\footnote{We require increasing functions as a convention, so larger $x_i$ map to larger values in $[N]$. Note that $\comp$ does \emph{not} need to be \emph{strictly} increasing; if $f$ is flat over interval $I \subseteq [0,1]$ then all $x_i \in I$ will always be encoded by the same value. This is useful if no $x_i$ in $I$ ever occurs, i.e. $I$ has zero probability mass under the prior.} function $\comp:[0,1] \to [0, 1]$ (we denote the set of such functions as $\compset$) and an integer $N$ representing the number of quantization levels, or \emph{granularity}. To simplify the problem and algorithm, we use the same $\comp$ for each element of the vector $\bx = (x_1, \dots, x_\az) \in \triangle_{\az-1}$ (see \Cref{rmk::symmetric-distribution}). To quantize $x \in [0, 1]$, the compander computes $\comp(x)$ and applies a uniform quantizer with $N$ levels, i.e. encoding $x$ to $n_N(x) \in [N]$ if $\comp(x) \in (\frac{n-1}{N}, \frac{n}{N}]$; this is equivalent to $n_N(x) = \lceil \comp(x) N \rceil$. 

%\footnote{\jennifer{}The purpose of requiring $f$ to be increasing is a that we can use $f$ to define which interval of values will map to each quantized value. It is not necessary for $f$ to be strictly increasing. If an interval $A$ under $f$ maps to the same value, our compander method will simply map everything in $A$ to the same value. We would use functions which are not increasing (say all decreasing functions) but doing this is more complicated and does not gain anything for our method.}

This encoding partitions $[0,1]$ into \emph{bins} $I^{(n)}$:
\begin{align}\label{eq::bins}
    x \in I^{(n)} = \comp^{-1} \Big(\Big(\frac{n-1}{N}, \frac{n}{N}\Big] \Big) \iff n_N(x) = n
\end{align}
where $\comp^{-1}$ denotes the preimage under $f$.

As an example, consider the function $f(x) = x^s$.
Varying $s$ gives a natural class of functions from $[0,1]$ to $[0,1]$, which we call the class of \emph{power companders}.
If we select $s = 1/2$ and $N = 4$, then the $4$ bins created by this encoding are 
\begin{align}
I^{(1)} &= (0, 1/16], I^{(2)} = (1/16, 1/4], \\
I^{(3)} &= (1/4, 9/16], I^{(4)} = (9/16, 1]\,. 
\end{align}

%\begin{align}
    %x \in I^{(n)} = \comp^{-1} \big(((n-1)/N, n/N] \big) \iff n_N(x) = n
%\end{align}
%where $\comp^{-1}$ denotes the preimage in $\comp$; note that $I^{(n)}$ is a sub-interval of $[0,1]$.
%We denote these bins as $I^{(n)}$, and the bin containing $x$ at granularity $N$ as $I^{(n_N(x))}$.

%Let $\hat y(x) = \dec(\enc(x))$, the decoding of the encoding of $x$.

%For any $x \in [0, 1]$ and granularity $N$, let $\enc(x)\in [N]$ be the encoded value of $x$. 

%For compander quantization, $\enc(x) = \lceil \comp(x) N \rceil$, creating a set of quantization bins on $[0, 1]$. 

%We let $\compder$ be the deriviative of $\comp$, i.e $\compder = \comp'$.

\subsubsection{Decoding} \label{sec::decoding}

%\ypnb{Guys, this is incredibly confusing even for someone who has read this many times. What you denote as $\hat y(x)$ is undefined. Because you say it sometimes can be $\bar y$ and sometimes $\tilde y$. You need to do two things. Define 4 different quantities: $\bar y$ (raw midpoint), $\tilde y$ (raw centroid), $\bar \newvar$ (normalized midpoint) and $\tilde \newvar$ (normalized centroid). You can choose other letters, but there should never be a situation like we have here that a quantity that is used on every page ($\hat y$ and $y$) have no fixed meaning. Another possibility is to denote $\bnormvar^{(c)}$, $\bnormvar^{(m)}$ and $\widehat{\bnormvar^{(c)}}$, $\widehat{\bnormvar^{(m)}}$. Then when you want to avoid specifying which one exactly you mean, you can just use $\bnormvar$ and $\widehat{\bnormvar}$.} -- \textcolor{blue}{AVIV: You must be kidding, we use this notation throughout the entire ISIT paper. We can't change it now. Also I think we've been through this before and tried various alternatives and none worked better.}  \nbyp{I don't understand how any of the Aviv's arguments change the fact that notation is undefined. This is not how math papers are written. The fact that this was used in ISIT is irrelevant, nobody is going to read ISIT paper if a journal version exists.}

To decode $n \in [N]$, we pick some $\rawvar_{(n)} \in I^{(n)}$ to represent all $x \in
I^{(n)}$; for a given $x$ (at granularity $N$), its representation is denoted $\rawvar(x) =
\rawvar_{(n_N(x))}$. This is generally either the \emph{midpoint} of the bin or, if $x$ is drawn
randomly from a known prior\footnote{Priors on $\triangle_{\az-1}$ induce priors over $[0,1]$ for each entry.} $p$, the \emph{centroid} (the mean within bin $I^{(n)}$). %The midpoint of $I^{(n)}$ can be computed quickly using the inverse of $f$. 
The midpoint and centroid of $I^{(n)}$ are defined, respectively, as 
\begin{align}
 \bar{y}_{(n)} &= {1\over2} \left(\comp^{-1}\left(\frac{n-1}{N}\right) + \comp^{-1}\left(\frac{n}{N}\right)\right)\\
    \widetilde{y}_{(n)} &= \bbE_{X \sim p} [X \, | \, X \in I^{(n)}] \,.
\end{align}
We will discuss this in greater detail in \Cref{sec::x-from-prior}. %\jennifer{Is this link correct, should it not be remark 3?}

%To decode $n \in [N]$, we pick a point $\widehat{y}^{(n)}$ in $I^{(n)}$ to represent all $x \in I^{(n)}$; for a given $x$ (at granularity $N$), its representative is denoted $\widehat{y}(x) = \widehat{y}^{(n_N(x))}$. This is usually the \emph{midpoint} (denoted $\bar y(x)$) or, if $x$ is drawn randomly according to a prior,\footnote{If the probability vector comes from a prior on $\triangle_{\az-1}$, this induces a prior over $[0,1]$ for each entry.} the \emph{centroid} (the mean within bin $I^{(n_N(x))}$, denoted as $\widetilde y(x)$). 

%Choosing a simple function $\comp$ makes it easy to store the quantization scheme.

Handling each element of $\bx$ separately means the decoded values may not sum to $1$, so we normalize the vector after decoding. Thus, if $\bx$ is the input,
\begin{align}\label{eq::norm_step}
    \normvar_i(\bx) = \frac{\rawvar(x_i)}{\sum_{j = 1}^\az \rawvar(x_j)}
\end{align}
and the vector $\bnormvar = \bnormvar(\bx) = (\normvar_1(\bx), \dots, \normvar_\az(\bx)) \in \triangle_{\az-1}$ is the output of the compander. This notation reflects the fact that each entry of the normalized reconstruction depends on all of $\bx$ due to the normalization step. We refer to $\brawvar = \brawvar(\bx) = (\rawvar(x_1), \dots, \rawvar(x_\az))$ as the \emph{raw} reconstruction of $\bx$, and $\bnormvar$ as the \emph{normalized} reconstruction. %(We generally use the $\widehat{\cdot}$ accent to mark values dependent on the raw reconstruction.)
If the raw reconstruction uses centroid decoding, we likewise denote it using $\widetilde{\by} = \widetilde{\by}(\bx) = (\widetilde{y}(x_1), \dots, \widetilde{y}(x_\az))$. For brevity we may sometimes drop the $\bx$ input in the notation, e.g. $\bz := \bz(\bx)$; if $\bX$ is random we will sometimes denote its quantization as $\bZ := \bz(\bX)$.

Thus, any $\bx \in \triangle_{\az-1}$ requires $\az \lceil \log_2 N \rceil$  bits to store; to encode and decode, only $\comp$ and $N$ need to be stored (as well as the prior if using centroid decoding). Another major advantage is that a single $\comp$ can work well over many or all choices of $N$, making the design more flexible.

\subsubsection{KL divergence loss}

The loss incurred by representing $\bx$ as $\bnormvar := \bnormvar(\bx)$ is the KL divergence
%\footnote{In this paper $\log$ denotes logarithm to the natural base. \nbyp{Right?}} -- AVIV: Yes, but we already state that elsewhere.
\begin{align}\label{eq::kl_loss_norm}
    D_{\kl}(\bx\| \bnormvar) = \sum_{i=1}^{\az} x_i \log \frac{x_i}{\normvar_i} \,.
\end{align}
Although this loss function has some unusual properties (for instance $D_{\kl}(\bx \| \bnormvar) \neq D_{\kl}(\bnormvar \| \bx)$ and it does not obey the triangle inequality) it measures the amount of `mis-representation' created by representing the probability vector $\bx$ by another probability vector $\bnormvar$, and is hence is a natural quantity to minimize. In particular, it represents the excess code length created by trying to encode the output of $\bx$ using a code built for $\bnormvar$, as well as having connections to hypothesis testing (a natural setting in which the `difference' between probability distributions is studied).
%\aviv{Added more justification, is it ok now?}
%Note that the vectors $\bx$ and $\bnormvar(\bx)$ are probability distributions over a size-$\az$ alphabet, i.e. they sum up to one. The quantity $D_{\kl}(\bx\| \bnormvar(\bx))$ does not stand for the divergence between distributions on $\bx$ and $\bnormvar(\bx)$.

\subsubsection{Distributions from a prior} \label{sec::x-from-prior}

Much of our work concerns the case where $\bx \in \triangle_{\az-1}$ is drawn from some prior $P_\bx$ (to be commonly denoted as simply $P$). Using a single $\comp$ for each entry means we can WLOG assume that $P$ is symmetric over the alphabet, i.e. for any permutation $\sigma$, if $\bX \sim P$ then $\sigma(\bX) \sim P$ as well. 
This is because for any prior $P$ over $\triangle_{\az-1}$, there is a symmetric prior $P'$ such that
    \begin{align}
        \bbE_{\bX \sim P} [D_{\kl}(\bX \| \bnormvar(\bX))] \hspace{-0.2pc} = \hspace{-0.2pc} \bbE_{\bX' \sim P'} [D_{\kl}(\bX' \| \bnormvar(\bX'))]
    \end{align}
for all $f$, where $\bnormvar(\bX)$ is the result of quantizing (to any number of levels) with $f$ as the compander. To get $\bX' \sim P'$, generate $\bX \sim P$ and a uniformly random permutation $\sigma$, and let $\bX' = \sigma(\bX)$.

We denote the set of symmetric priors as $\cP^\triangle_\az$. Note that a key property of symmetric priors is that their marginal distributions are the same across all entries, and hence we can speak of $P \in \cP^\triangle_\az$ having a single marginal $p$.

\begin{remark} \label{rmk::symmetric-distribution}
In principle, given a nonsymmetric prior $P_\bx$ over $\triangle_{\az-1}$ with marginals $p_1, \dots, p_\az$, we could quantize each letter's value with a different compander $f_1, \dots, f_\az$, giving more accuracy than using a single $f$ (at the cost of higher complexity). However, the symmetrization of $P_\bx$ over the letters (by permuting the indices randomly after generating $\bX \sim P_\bx$) yields a prior in $\cP^\triangle_\az$ on which any single $f$ will have the same (overall) performance and cannot be improved on by using varying $f_i$. Thus, considering symmetric $P_\bx$ suffices to derive our minimax compander.
\end{remark}

While the random probability vector comes from a prior $P \in \cP^\triangle_\az$, our analysis will rely on decomposing the loss so we can deal with one letter at a time. Hence, we work with the marginals $p$ of $P$ (which are identical since $P$ is symmetric), which we refer to as \emph{single-letter distributions} and are probability distributions over $[0,1]$.

We let $\cP$ denote the class of probability distributions over $[0,1]$ that are absolutely continuous with respect to the Lebesgue measure. We denote elements of $\cP$ by their probability density functions (PDF), e.g. $p \in \cP$; the cumulative distribution function (CDF) associated with $p$ is denoted $F_p$ and satisfies $F'_p(x) = p(x)$ and $F_p(x) = \int_0^x p(t) \, dt$ (since $F_p$ is monotonic, its derivative exists almost everywhere). Note that while $p \in \cP$ does not have to be continuous, its CDF $F_p$ must be absolutely continuous. Following common terminology~\cite{grimmett2001}, we refer to such probability distributions as \emph{continuous}.

%$p \in \cP$ has a cumulative distribution function (CDF) $F_p$ and a probability density function (PDF) $p$ satisfying $p(x) = F'_p(x)$ and $F_p(x) = \int_0^x p(t) \, dt$ (since $F_p$ is monotonic, its derivative exists almost everywhere).  (the PDF $p$ does not have to be continuous, but the CDF $F_p$ has to be absolutely continuous).

%i.e. the distribution is completely characterized by its probability density function $F'_p(x)$, which we also denote as $p(x)$.

Let $\cP_{1/\az} = \{p \in \cP : \bbE_{X\sim p}[X] = 1/\az\}$. Note that $P \in \cP^\triangle_\az$ implies its marginals $p$ are in $\cP_{1/\az}$. %Finally, let $\cP_\az^\triangle$ be the set of symmetric priors over $\triangle_{\az-1}$ with marginals $p \in \cP$; note that this also implies $p \in \cP_{1/\az}$ because vectors on the simplex sum to $1$.

%We let $\cP$ denote the class of absolutely continuous probability distributions on $[0,1]$; we represent members of $\cP$ by their probability density functions (PDFs), denoted $p$. Let $\cP_{1/\az} \subset \cP$ be the set of $p$ where $\bbE_{X\sim p}[X] = 1/\az$. Finally, let $\cP_\az^\triangle$ be the set of $p \in \cP$ which correspond to a prior $P$ over $\triangle_{\az-1}$. Note that $P \in \cP_\az^\triangle$ implies its marginals are in $\cP_{1/\az}$. 

%\begin{itemize}
%    \item $P$ is a prior distribution on $\bx$.
%    \item $p = {1\over \az} \sum_i P_{x_i}$ is the average prior marginal.
%    \item Throughout we assume that $\bx\sim P$ and $p$ and $P$ are related as above.
%\end{itemize}

%$\cP_\az^{\triangle} \subset \cP_{1/\az}$ to the set of marginal densities of symmetric distributions on the probability simplex $\triangle_{\az-1}$.

\subsubsection{Expected loss and preliminary results}

For $P \in \cP^\triangle_\az$, $\comp \in \compset$ and granularity $N$, we define the \emph{expected loss}:
\begin{equation}\label{eq::def_loss}
    \cL_\az(P, \comp, N) = \bbE_{\bX \sim P}[D_{\kl}(\bX \| \bnormvar(\bX))]\,.
\end{equation}
This is the value we want to minimize over $\comp$.

\begin{remark}
While $\bX$ and $\bnormvar(\bX)$ are random, they are also probability vectors. The KL divergence $D_{\kl}(\bX \| \bnormvar(\bX))$ is the divergence between $\bX$ and $\bnormvar(\bX)$ themselves, not the prior distributions over $\triangle_{\az-1}$ they are drawn from.
\end{remark}

Note that $\cL_\az(P,\comp,N)$ can almost be decomposed into a sum of $\az$ separate expected values, except the normalization step \eqref{eq::norm_step} depends on the random vector $\bX$ as a whole. Hence, we define the \emph{raw loss}:
\begin{align} \label{eq::raw-loss}
    \rawloss_\az(P, \comp, N) \hspace{-0.2pc} = \hspace{-0.2pc} \bbE_{\bX \sim P}\Big[\sum_{i=1}^\az X_i \log(X_i/\widetilde{y}(X_i))\Big]\,.
\end{align}
We also define for $p \in \cP$, the \emph{single-letter loss} as
\begin{align} \label{eq::raw-ssl}
    \singleloss(p, \comp, N) =  \bbE_{X \sim p} \big[ X \log ( X/\widetilde{y}(X)) \big]\,.
\end{align}
The raw loss is useful because it bounds the (normalized) expected loss and is decomposable into single-letter losses. Note that both raw and single-letter loss are defined with centroid decoding.
\begin{proposition}\label{lem::im-a-barby-girl}
For $P \in \cP^\triangle_\az$ with marginals $p$,  %($\widehat{y}(X_i) = \widetilde{y}(X_i) ~\, \forall i$),
    \begin{align}
        \cL_\az(P, \comp, N) \leq \rawloss_\az(P, \comp, N) = \az \, \singleloss(p,\comp,N)\,.
    \end{align}
\end{proposition}

\iflong%%%%%%%%%
\begin{proof}
Separating out the normalization term gives
\begin{align}
    \cL \eqstartnonumshort (P, \comp, N) =  \bbE_{\bX \sim P} [D_{\kl}(\bX || \bnormvar(\bX))]
    %& = \bbE_{\bX \sim P}  \left[\widehat{D}_{\kl}(\bX || \widetilde{\bY}) + \log \left( \sum_{k=1}^\az \widetilde{Y}_k \right)\right]\\
     \\ &=  \rawloss_\az(P, \comp, N) + \bbE_{\bX \sim P}  \left[ \log \left( \sum_{i=1}^\az \widetilde{y}(X_i) \right)\right] \,.
\end{align}
%it remains to show that 
%\begin{align}
%    \bbE_{\bX \sim P}  \left[ \log \left( \sum_{k=1}^\az \widetilde{Y}_k \right)\right] \leq 0\,.
%\end{align}
Since $\bbE[\widetilde{y}(X_i)] = \bbE[X_i]$ for all $i$, $\sum_{i = 1}^\az \bbE[\widetilde{y}(X_i)] =\sum_{i = 1}^\az \bbE[{X}_i] = 1 $.
Because $\log$ is concave, by Jensen's Inequality
\begin{align}
    \bbE_{\bX \sim P} \bigg[\log \Big( \sum_{i=1}^\az \widetilde{y}(X_i) \Big)\bigg] &\leq \log \Big( \bbE \Big[\sum_{i=1}^\az \widetilde{y}(X_i)\Big] \Big)
    \\&= \log(1) = 0
\end{align}
and we are done.\footnote{An upper bound similar to \Cref{lem::im-a-barby-girl} can be found in \cite[Lemma 1]{benyishai2021}.}
\end{proof}
\fi %%%%%%%%%%%%%%
To derive our results about worst-case priors (for instance, \Cref{thm::minimax_compander}), we will also be interested in $\singleloss(p,\comp,N)$ even when $p$ is not known to be a marginal of some $P \in \cP^\triangle_\az$.

\begin{remark} \label{rmk::centroid-needed}
Though one can define raw and single-letter loss without centroid decoding (replacing $\widetilde{y}$ in \eqref{eq::raw-loss} or \eqref{eq::raw-ssl} with another decoding method $\widehat{y}$), this removes much of their usefulness. This is because the resulting expected loss can be dominated by the difference between $\bbE[X]$ and $\bbE[\widehat{y}(X)]$, potentially even making it negative; specifically, the Taylor expansion of $X \log(X/\widehat{y}(X))$ has $X - \widehat{y}(X)$ in its first term, which can have negative expectation.

While this can make the expected `raw loss' negative under general decoding, it cannot be exploited to make the (normalized) expected loss negative 
%\footnote{As expected, since negative KL loss between probability distributions is not possible.} 
because the normalization step $\normvar_i(\bX) = \widehat{y}(X_i)/\sum_j \widehat{y}(X_j)$ cancels out the problematic term. Centroid decoding avoids this problem by ensuring $\bbE[X] = \bbE[\widetilde{y}(X)]$, removing the issue.
\end{remark}

As we will show, when $N$ is large these values are roughly proportional to $N^{-2}$ (for well-chosen $\comp$) and so we define the \emph{asymptotic single-letter loss}:
\begin{align} \label{eq::raw-assl}
    \singleloss(p,\comp) = \lim_{N \to \infty} N^2 \singleloss(p,\comp,N)\,.
\end{align}
We similarly define $\rawloss_\az(P,\comp)$ and $\cL_\az(P,\comp)$. 
While the limit in \eqref{eq::raw-assl} does not necessarily exist for every $p, \comp$, we will show that one can ensure it exists by choosing an appropriate $\comp$ (which works against any $p \in \cP$), and cannot gain much by not doing so.

%We show in \Cref{thm::asymptotic-normalized-expdiv} that $\singleloss(p,\comp) = L^\dagger(p,\comp)$ under certain general conditions (entirely within the control of the quantizer), and that $L^\dagger(p,\comp)$ is always a lower bound on the asymptotic performance of $\comp$ against $p$.

\section{Results}

\label{sec::main-theorems}

We demonstrate, theoretically and experimentally, the efficacy of companding for quantizing probability distributions with KL divergence loss.

%Theoretically, we consider both \emph{asymptotic average loss} (given a compander $f$ and a prior $P$
%
%In particular, we formulate the \emph{minimax compander} $\comp^*_\az$ for alphabet size $\az$, and prove it minimizes $\sup_{p \, \in \, \cP_{1/\az}} \widetilde{L}(p,\comp^*_\az)$. 

%and show a compander which is minimax-optimal (for raw loss) in the limit as $N \to \infty$; we also show strong non-asymptotic worst-case guarantees for this compander with midpoint decoding which match (up to a constant) the asymptotic performance and validate these results experimentally.
%While our average loss results assume centroid decoding, our worst-case loss results and experiments use midpoint decoding (as it is both computationally much simpler and often the prior is not known or even doesn't exist).
%Though our theoretical results are asymptotic as $N \to \infty$ and focus on raw loss (which uses centroid decoding), the experimental loss of the various companders with midpoint decoding (normalized) closely tracks the raw loss predicted theoretically, even for quantization levels as low as $N = 256$ ($8$ bits per value). 
%
%This project has a mix of theoretical and experimental results. %We want to show good performance on real world data, though the way we derive these algorithms are from our theoretical results.

\subsection{Theoretical Results} \label{sec::theoretical-results}

While we will occasionally give intuition for how the results here are derived, our primary concern in this section is to fully state the results and to build a clear framework for discussing them. 

%We prove \Cref{thm::asymptotic-normalized-expdiv} in \Cref{sec::asymt_single} (though proofs of some lemmas and propositions leading up to it are given in \Cref{sec::appendix_proofs_asymptotic}); in \Cref{sec::minimax}, we show \Cref{thm::optimal_compander_loss,thm::minimax_compander} and \Cref{cor::worstcase_prior} (leaving \Cref{thm::approximate-minimax-compander} for \Cref{sec::proof_L_appx_minimax_compander}).

Our main results concern the formulation and evaluation of a \emph{minimax compander} $\comp^*_\az$ for alphabet size $\az$, which satisfies
\begin{align} \label{eq::minimax-condition}
    \comp^*_\az = \underset{\comp \, \in \, \compset}{\argmin} \underset{p \, \in \, \cP_{1/\az}}{\sup} \widetilde{L}(p,\comp) \,.
\end{align}
We require $p \in \cP_{1/\az}$ because if $P \in \cP^\triangle_\az$ and is symmetric, its marginals are in $\cP_{1/\az}$. 
%Thus $\comp^*_\az$ guarantees the best possible raw loss (asymptotic as $N \to \infty$) against an unknown or adversarial prior.

The natural counterpart of the minimax compander $\comp^*_\az$ is the \emph{maximin density} $p^*_\az \in \cP_{1/\az}$, satisfying
\begin{align} \label{eq::maximin-condition}
    p^*_\az = \underset{p \, \in \, \cP_{1/\az}}{\argmax} \underset{\comp \, \in \, \compset}{\inf} \widetilde{L}(p,\comp) \,.
\end{align}
We call \eqref{eq::minimax-condition} and \eqref{eq::maximin-condition}, respectively, the \emph{minimax condition} and the \emph{maximin condition}.

In the same way that the minimax compander gives the best performance guarantee against an unknown single-letter prior $p \in \cP_{1/\az}$ (asymptotic as $N \to \infty$), the maximin density is the most difficult prior to quantize effectively as $N \to \infty$. Since they are highly related, we will define them together:

\begin{proposition} \label{prop::maximin-density}
    For alphabet size $\az > 4$, there is a unique $c_{\az} \in [\frac{1}{4}, \frac{3}{4}]$ such that if $a_{\az} = (4/(c_{\az} \az \log \az + 1))^{1/3}$ and $b_{\az} = 4/a_{\az}^2 - a_{\az}$, then the following density is in $\cP_{1/\az}$:
    \begin{align}
        &p^*_{\az}(x) = (a_{\az} x^{1/3} + b_{\az} x^{4/3})^{-3/2} \label{eq::maximin-density}\,.
    \end{align}
    Furthermore, $\lim_{\az \to \infty} c_{\az} = 1/2$.
\end{proposition}
Note that this is both a result and a definition: we show that $a_\az, b_\az, c_\az$ exist which make the definition of $p^*_\az$ possible. 
%
%We will show that the $p^*_\az$ defined above also satisfies \eqref{eq::maximin-condition}.
%
With the constant $c_\az$, we define the minimax compander:
\begin{definition} \label{def::minimax-compander}
    Given the constant $c_\az$ as shown to exist in \Cref{prop::maximin-density}, the \emph{minimax compander} is the function $f^*_\az : [0,1] \to [0,1]$ where
    \begin{align}\label{eq::minimax-compander}
        \comp^*_\az(x) = \frac{\arcsinh(\sqrt{c_\az  (\az \log \az) \, x})}{\arcsinh(\sqrt{c_\az  \az \log \az})}\,.
    \end{align}
    The \emph{approximate minimax compander} $f^{**}_\az$ is
    \begin{align} \label{eq::appx-minimax-compander}
        \comp^{**}_\az(x) = \frac{\arcsinh(\sqrt{(1/2) (\az \log \az) \, x})}{\arcsinh(\sqrt{(1/2) \az \log \az})}\,.
    \end{align}
\end{definition}
%As with the maximin density, we will show that the minimax compander $f^*_\az$ defined above satisfies the minimax condition \eqref{eq::minimax-condition}. While this makes it theoretically optimal in a sense, the constant $c_\az$ is potentially tricky to compute and use; however, since $\lim_{\az \to \infty} c_\az = 1/2$ (as shown in \Cref{prop::maximin-density}), we can replace $c_\az$ with $1/2$ to get the approximate minimax compander, which we will show has very similar performance.

\begin{remark} \label{rmk::minimax-is-closed-form}
While $\comp^*_\az$ and $\comp^{**}_\az$ might seem complex,
$
    \arcsinh(\sqrt{\newvar}) = \log(\sqrt{\newvar} + \sqrt{\newvar+1})
$
so they are relatively simple functions to work with.
\end{remark}

We will show that $f^*_\az, p^*_\az$ as defined above satisfy their respective conditions \eqref{eq::minimax-condition} and \eqref{eq::maximin-condition}: 
\begin{theorem}\label{thm::minimax_compander}
The minimax compander $\comp^*_\az$ and maximin single-letter density $p^*_\az$ satisfy
\begin{align}
    &\sup_{p \in \cP_{1/\az}} \singleloss(p,\comp^*_\az) = \inf_{\comp \in \compset} \sup_{p \in \cP_{1/\az}} \singleloss(p,\comp)   \label{eq::minmax}
    \\ = & \sup_{p \in \cP_{1/\az}} \inf_{\comp \in \compset} \singleloss(p,\comp) = \inf_{\comp \in \compset} \singleloss(p^*_\az, \comp) \label{eq::maxmin}
\end{align}
which is equal to $\singleloss(p^*_\az, \comp^*_\az)$ and satisfies
\begin{align} \label{eq::raw_loss_saddle}
    \singleloss(p^*_\az, \comp^*_\az) = \frac{1}{24} (1 + o(1)) \az^{-1}\log^2 \az.
\end{align}
\end{theorem}

Since any symmetric $P \in \cP^\triangle_\az$ has marginals $p \in \cP_{1/\az}$, this (with \Cref{lem::im-a-barby-girl}) implies an important corollary for the normalized KL-divergence loss incurred by using the minimax compander:
\begin{corollary}\label{cor::worstcase_prior}
For any prior $P \in \cP^{\triangle}_\az$,
\begin{align}
    \cL_\az(P,\comp^*_\az) \leq \rawloss_\az(P,\comp^*_\az) = \frac{1}{24} (1 + o(1))\log^2 \az \,.
\end{align}
\end{corollary}
However, the set of symmetric $P \in \cP^\triangle_\az$ does not correspond exactly with $p \in \cP_{1/\az}$: while any symmetric $P \in \cP^\triangle_\az$ has marginals $p \in \cP_{1/\az}$, it is not true that any given $p \in \cP_{1/\az}$ has a corresponding symmetric prior $P \in \cP^\triangle_\az$. Thus, it is natural to ask: can the minimax compander's performance be improved by somehow taking these `shape' constraints into account? The answer is `not by more than a factor of $\approx 2$':
\begin{proposition}\label{prop::bound_worstcase_prior_exist}
There is a prior $P^* \in \cP^{\triangle}_\az$ such that for any $P \in \cP^\triangle_\az$
\begin{align}\label{eq::bound_worstcase_prior}
    \inf_{\comp \in \compset} \rawloss_\az(P^*, \comp) \geq \frac{\az - 1}{2\az} \rawloss_\az(P, \comp^*_\az) \,.
\end{align}
\end{proposition}

While the minimax compander satisfies the minimax condition \eqref{eq::minimax-condition}, it requires working with the constant $c_\az$, which, while bounded, is tricky to compute or use exactly. Hence, in practice we advocate using the \emph{approximate minimax compander} \eqref{eq::appx-minimax-compander}, which yields very similar asymptotic performance without needing to know $c_\az$:

\begin{proposition} \label{thm::approximate-minimax-compander}
Suppose that $\az$ is sufficiently large so that $c_\az \in [\frac{1}{2 (1 + \varepsilon)}, \frac{1 + \varepsilon}{2}]$. Then for any $p \in \cP$,
\begin{align}
    \singleloss(p,\comp^{**}_\az) \leq (1+ \varepsilon) \singleloss(p,\comp^*_\az)\,.
\end{align}
\end{proposition}

Before we show how we get \Cref{thm::minimax_compander}, we make the following points:

\begin{remark}\label{rmk::loss_with_uniform}
%\nbyp{You need to explain what density you are talking about.} 
If we use the uniform quantizer instead of minimax there exists a $P \in \cP^\triangle_\az$ where
\begin{align}\label{eq::uniform_achieve}
    \bbE_{\bX \sim P}[D_{\kl}(\bX \| \bnormVar)] = \Theta\left(\az^2 N^{-2} \log N \right)\,.
\end{align}
This is done by using marginal density $p$ uniform on $[0,2/\az]$. To get a prior $P \in \cP^\triangle_\az$ with these marginals, if $\az$ is even, we can pair up indices so that $x_{2j-1} = 2/\az - x_{2j}$ for all $j = 1, \dots, \az/2$ (for odd $\az$, set $x_\az = 1/\az$) and then symmetrize by permuting the indices. See \Cref{sec::uniform} for more details.

The dependence on $N$ is worse than $N^{-2}$ resulting in $\widetilde{L}(p,f) = \infty$. This shows theoretical suboptimality of the uniform
quantizer. Note also that the quadratic dependence on $\az$ is significantly worse than the $\log^2 \az$ dependence achieved by the minimax compander.
%specifically, by \Cref{thm::asymptotic-normalized-expdiv}, any $\comp \in \cF^\dagger$ is guaranteed to have loss $\propto \, N^{-2}$ for all priors.

Incidentally, other single-letter priors such as $p(x) = (1-\alpha)x^{-\alpha}$ where $\alpha = \frac{\az-2}{\az-1}$ can achieve worse dependence on $N$ (specifically, $N^{-(2-\alpha)}$ for this prior). However, the example above achieves a bad dependence on both $N$ and $\az$ simultaneously, showing that in all regimes of $\az, N$ the uniform quantizer is vulnerable to bad priors.

%\nbyp{Guys, what is this? You can either say ``we do not know if~\eqref{eq::uniform_achieve} is the worst case behavior'' or replace the density $P$ with something else that you know is worse. What is the point of saying here's some behavior, but we know about even worse ones? That's incredibly strange} -- \textcolor{blue}{AVIV: The worse one is annoying to analyze but we know what it is. This is just to show that bad companders don't do $N^{-2}$.}\nbyp{What do you mean it's annoying to analyze? You are not analyzing anything here, you are making an unverifiable claim.} \textcolor{blue}{AVIV: By ``worse but annoying to analyze'' I meant that we can get worse dependence on $N$, i.e. $N^{-3/2+\varepsilon}$ as opposed to $\log(N) N^{-2}$ but the dependence on $K$ is unknown. The example we give has bad dependence on both $N$ and $K$. Otherwise people might say `ok dependence on $N$ is bad but for reasonable values and with large $K$ blah blah blah' whereas for this example there's no argument possible.}
\end{remark}

\begin{remark}
Instead of the KL divergence loss on the simplex, we can do a similar analysis to find the minimax compander for $L_2^2$ loss on the unit hypercube. The solution 
is given by the identity function $\comp(x)=x$ corresponding to the standard (non-companded) uniform quantization. (See \Cref{sec::other_losses}.)
\end{remark}

%\subsection{Intermediate Theoretical Results}

To show \Cref{thm::minimax_compander} we formulate and show a number of intermediate results which are also of significant interest for a theoretical understanding of companding under KL divergence, in particular studying the asymptotic behavior of $\widetilde{L}(p,f,N)$ as $N \to \infty$. We define:
\begin{definition}
For $p \in \cP$ and $\comp \in \compset$, let
\begin{align}
    L^\dagger(p,\comp) &= \frac{1}{24} \int_0^1 p(x) \compder(x)^{-2} x^{-1} \, dx
    \\ &= \bbE_{X \sim p}\Big[\frac{1}{24}\compder(X)^{-2} X^{-1}\Big] \label{eq::raw_loss}
    %\\ &= \int_{[0,1]} \frac{1}{24} \compder(x)^{-2} x^{-1} \, dp
    \,.
\end{align}
\end{definition}

For full rigor, we also need to define a set of `well-behaved' companders:
\begin{definition}
    Let $\compset^\dagger \subseteq \compset$ be the set of $\comp$ such that for each $f$ there exist constants $c > 0$ and $\alpha \in (0,1/2]$ for which $\comp(x) - c x^{\alpha}$ is still monotonically increasing.
\end{definition}

Then the following describes the asymptotic single-letter loss of compander $f$ on prior $p$ (with centroid decoding):
\begin{theorem}  \label{thm::asymptotic-normalized-expdiv}
    For any $p \in \cP$ and $\comp \in \compset$,
    \begin{align}
        \liminf_{N \to \infty} N^2 \singleloss(p,\comp,N) \geq L^\dagger(p,\comp) \,. \label{eq::fatou-bound}
    \end{align}
    Furthermore, if $\comp \in \compset^\dagger$ then an exact result holds:
    \begin{align}
        \singleloss(p,\comp) &= L^\dagger(p,\comp) < \infty
        %\\L(p, \comp) &= \frac{1}{24}  \int_{0}^1 p(x) \compder(x)^{-2} (x^{-1} - 1) dx 
        \label{eq::norm_loss}
    \,.
    \end{align}
\end{theorem}

The intuition behind the formula for $L^\dagger(p,f)$ is that as $N \to \infty$, the density $p$ becomes roughly uniform within each bin $I^{(n)}$. Additionally, the bin containing a given $x \in [0,1]$ will have width $r_{(n)} \approx N^{-1} \compder(x)^{-1}$. Then, letting $\unif_{I^{(n)}}$ be the uniform distribution over $I^{(n)}$ and $\bar{y}_{(n)} \approx x$ be the midpoint of $I^{(n)}$ (which is also the centroid under the uniform distribution),
%\footnote{$\bar{y}_{(n)} \approx x$ because we are assuming $N$ is large and hence the interval $I^{(n)}$ is small, so any two values in $I^{(n)}$ are very similar.}
we apply the approximation
\begin{align}
    \bbE_{X \sim \unif_{I^{(n)}}}[X \log(X/\bar{y}_{(n)})] &\approx \frac{1}{24} r_{(n)}^2 \bar{y}_{(n)}^{-1} 
    \\ &\approx \frac{1}{24} N^{-2} \compder(x)^{-2} x^{-1} \,.
\end{align}
Averaging over $X \sim p$ and multiplying by $N^2$ then gives \eqref{eq::raw_loss}. One wrinkle is that we need to use the Dominated Convergence Theorem to get the exact result \eqref{eq::norm_loss}, but we cannot necessarily apply it for all $\comp \in \compset$; instead, we can apply it for all $\comp \in \compset^\dagger$, and outside of $\compset^\dagger$ we get \eqref{eq::fatou-bound} using Fatou's Lemma.

While limiting ourselves to $\comp \in \compset^\dagger$ might seem like a serious restriction, it does not lose anything essential because $\compset^\dagger$ is `dense' within $\compset$ in the following way:
\begin{proposition} \label{prop::approximate-compander}
For any $\comp \in \compset$ and $\delta \in (0,1]$,
\begin{align}
    \comp_\delta (x) = (1-\delta) \comp(x) + \delta x^{1/2} \label{eq::approximate-compander}
\end{align}
satisfies $\comp_\delta \in \compset^\dagger$ and
\begin{align}
    \lim_{\delta \to 0} \singleloss(p,\comp_\delta) = \lim_{\delta \to 0} L^\dagger(p,\comp_\delta) = L^\dagger(p,\comp) \label{eq::approximate-optimal-compander}\,.
\end{align}
\end{proposition}

\begin{remark}
    It is important to note that strictly speaking the limit represented by $\widetilde{L}(p,\comp)$ may not always exist if $\comp \not \in \cF^\dagger$. However: (i) one can always guarantee that it exists by selecting $\comp \in \compset^\dagger$; (ii) by \eqref{eq::fatou-bound}, it is impossible to use $f$ outside $\compset^\dagger$ to get asymptotic performance better than $L^\dagger(p,\comp)$; and (iii) by \Cref{prop::approximate-compander}, given $f$ outside $\compset^\dagger$, one can get a compander in $\compset^\dagger$ with arbitrarily close (or better) performance to $\comp$ by using $\comp_\delta(x) = (1-\delta)\comp(x) + \delta x^{1/2}$ for $\delta$ close to $0$. This suggests that considering only $\comp \in \compset^\dagger$ is sufficient since there is no real way to benefit by using $\comp \not \in \compset^\dagger$.
    
    Additionally, both $\comp^*_\az$ and $\comp^{**}_\az$ are in $\compset^\dagger$. Thus, in \Cref{thm::minimax_compander}, although the limit might not exist for certain $\comp \in \compset, p \in \cP_{1/\az}$, the minimax compander still performs better since it has less loss than even the $\liminf$ of the loss of other companders.
\end{remark}

Given \Cref{thm::asymptotic-normalized-expdiv}, it's natural to ask: for a given $p \in \cP$, what compander $f$ minimizes $L^\dagger(p,f)$? This yields the following by calculus of variations:
\begin{theorem} \label{thm::optimal_compander_loss}
The best loss against source $p \in \cP$ is
\begin{align}
    \hspace{-0.75pc}  \inf_{\comp \in \compset} \singleloss(p,\comp)  &= \min_{\comp \in \compset} L^\dagger(p,\comp) 
    \\ &=  \frac{1}{24} \Big(\int_0^1 (p(x)x^{-1})^{1/3} dx\Big)^3
    \label{eq::raw_overall_dist}
\end{align}
where the \emph{optimal compander against $p$} is
\begin{align}
    &\comp_p(x)  =  \underset{\comp \in \compset}{\argmin}  L^\dagger(p,\comp)  =  \frac{\int_0^x (p(t)t^{-1})^{1/3} \, dt}{\int_0^1 (p(t)t^{-1})^{1/3} \, dt} \label{eq::best_f_raw} %\argmin_{\comp \in \compset}
\end{align}
(satisfying  $\compder_p(x) \, \propto \, (p(x) x^{-1})^{1/3}$). 
\end{theorem}

Note that $f_p$ may not be in $\compset^\dagger$ (for instance, if $p$ assigns zero probability mass to an interval $I \subseteq [0,1]$, then $f_p$ will be constant over $I$). However, this can be corrected by taking a convex combination with $x^{1/2}$ as described in \Cref{prop::approximate-compander}.

The expression \eqref{eq::raw_overall_dist} represents in a sense how hard $p \in \cP$ is to quantize with a compander, and the maximin density $p^*_\az$ is the density in $\cP_{1/\az}$ which maximizes it;\footnote{The maximizing density over all $p \in \cP$ happens to be $p(x) = \frac{1}{2} x^{-1/2}$; however, $\bbE_{X \sim p}[X] = 1/3$ so it cannot be the marginal of any symmetric $P \in \cP^\triangle_\az$ when $\az > 3$.} in turn, the minimax compander $f^*_\az$ is the optimal compander against $p^*_\az$, i.e.
\begin{align}
    f^*_\az = f_{p^*_\az} \,.
\end{align}

So far we considered quantization of a random probability vector with a known prior. We next consider the case where the quantization guarantee is given pointwise, i.e. we cover $\triangle_{\az-1}$ with a
finite number of KL divergence balls of fixed radius. Note that since the prior is unknown, only the midpoint decoder can be used. 
%the results suggest that the minimax and approximate minimax companders are ideal for
%quantizing probability vectors from an unknown prior, for good practical guarantees a
%non-asymptotic bound would be preferable. In addition, the results above deal with centroid
%decoding, which is theoretically nice but difficult at best to use and, in the case of an unknown
%prior, actually not well-defined.\footnote{In \Cref{thm::minimax_compander} and the minimax
%condition \eqref{eq::minimax-condition}, the prior $p$ is unknown when the compander function is
%chosen but then the decoder is allowed to use it to compute the centroids of the bins.} Thus, we
%show the following non-asymptotic bound using midpoint decoding:
\begin{theorem}[Divergence covering]
\label{thm::worstcase_power_minimax}
On alphabet size $\az > 4$ and $N \geq 8 \log(2\sqrt{\az \log \az} + 1)$ intervals, the minimax and approximate minimax companders with midpoint decoding achieve \emph{worst-case loss} over $\triangle_{\az-1}$ of
\begin{align}
    \max_{\bx \in \triangle_{\az-1}}D_{\kl}(\bx\|\bnormvar) \leq (1 + \mathrm{err}(\az)) N^{-2} \log^2 \az
\end{align}
where $\mathrm{err}(\az)$ is an error term satisfying 
\begin{align}
    \mathrm{err}(\az) \leq 18 \frac{\log \log \az}{\log \az} \leq 7 \text{ when } \az > 4 \,.
\end{align}
\end{theorem}

%Note that this bound is also \emph{worst-case}, i.e. instead of dealing with a prior $P$ over
%$\triangle_{\az-1}$ and taking the average loss, it takes the loss of the worst possible $\bx \in
%\triangle_{\az-1}$. It also bounds the normalized loss rather than the raw loss, which is fine as
%this is what we really want to minimize. 
Note that the non-asymptotic worst-case bound matches (up to a constant factor) the 
known-prior asymptotic result~\eqref{eq::raw_loss_saddle}.
We remark that condition on $N$ is mild: for example, if $N = 256$ (i.e. we are representing the
probability vector with $8$ bits per entry), then $N > 8 \log(2\sqrt{\az \log \az}+1)$ for all $\az \leq 2.6 \times 10^{25}$. 

\begin{remark}
When $b$ is the number of bits used to quantize each value in the probability vector, using the
approximate minimax compander yields a worst-case loss  on the order of $2^{-2b} \log ^2 \az$.
In~\cite{phdthesis} we prove bounds on the optimal loss under arbitrary (vector) quantization of
probability vectors and show that this loss is sandwiched between $2^{-2 b\frac{\az}{\az - 1}}$ (\cite[Proposition 2]{phdthesis}) and  $2^{-2 b\frac{\az}{\az - 1}} \log \az$ (\cite[Theorem 2]{phdthesis}). 
Thus, the entrywise companders in this work are quite competitive. %\nbyp{Please reference an exact theorem/lemma from the thesis. Also isn't it also in our ISIT paper from last year? If yes, reference it here as well.} \textcolor{blue}{DONE. Did not mention in previous ISIT paper.}
\end{remark}

%\begin{remark} \nbyp{I think this remark is trivial and should be removed.} -- \textcolor{blue}{AVIV: I think it addresses a natural question a reader might have (`what about nonasymptotic bounds in the average case?')}
%We do not consider average-case non-asymptotic bounds, i.e. those that work for known finite $N$, because they are functionally indistinguishable from worst-case bounds as the prior can be arbitrarily concentrated around a single point $\bx \in \triangle_{\az-1}$.
%\end{remark}
%\jennifer{Removed remark}

%In addition to our main and intermediate results studying the minimax compander and various related ideas, 
We also consider the natural family of \emph{power companders} $f(x)=x^s$, both in terms of average asymptotic raw loss and worst-case non-asymptotic normalized loss. By definition, $f(x) \in \compset^\dagger$ and hence $\widetilde{L}(p,f)$ is well-defined and \Cref{thm::asymptotic-normalized-expdiv} applies. 
\begin{theorem}\label{thm::power_compander_results}
The power compander $f(x) = x^s$ with exponent $s \in (0,1/2]$ has asymptotic loss
\begin{align}
    \underset{p \in \cP_{1/\az}} \sup \widetilde{L}(p,f) = \frac{1}{24} s^{-2} K^{2s-1}\label{eq::power_loss_s}\,.
\end{align}
For $\az > 7$, \eqref{eq::power_loss_s} is minimized by setting $s = \frac{1}{\log \az}$ (when $\az \leq 7$, $\frac{1}{\log \az} > 1/2$) and $f(x) = x^s$ achieves
\begin{align}
    \underset{p \in \cP_{1/\az}} \sup \widetilde{L}(p,f) &= \frac{e^2}{24} \frac{1}{\az} \log^2 \az
    \\ \text{and }~~ \underset{P \in \cP^\triangle_\az} \sup \widetilde{\cL}(P,f) &= \frac{e^2}{24} \log^2 \az\,.
\end{align}

Additionally, when $s = \frac{1}{\log \az}$, it achieves the following worst-case bound with midpoint decoding for $\az > 7$ and $N > \frac{e}{2} \log \az$:
\begin{align}
    \max_{\bx \in \triangle_{\az-1}} \hspace{-0.4pc} D_{\kl}(\bx\|\bnormvar) \hspace{-0.2pc} &\leq \hspace{-0.2pc}  (1 + \mathrm{err}(\az,N)) \frac{e^2}{2} N^{-2} \log^2 \az
    \\ \text{where } \mathrm{err}&(\az,N) = \frac{e}{2} \frac{\log \az}{N - \frac{e}{2}\log \az} \,. \label{eq::power_worst_case_bound}
\end{align}
\end{theorem}
Note in particular that when $N \geq e \log \az$, we have $\mathrm{err}(\az,N) \leq 1$, giving a bound of $\max_{\bx \in \triangle_{\az-1}}D_{\kl}(\bx\|\bnormvar) \leq  e^2 N^{-2} \log^2 \az$.

We can think of $s = \frac{1}{\log \az}$ as a `minimax' among the class of power companders. This result shows $f(x) = x^{\frac{1}{\log \az}}$ has performance within a constant factor of the minimax compander, and hence might be a good alternative.

\iflong
\else
Due to space constraints, we omit the proofs of \Cref{rmk::best-constant,thm::worstcase_power_minimax}. We sketch the other proofs in \Cref{sec::asymt_single,sec::minimax}.
\fi

\subsection{Experimental Results}

\label{sec::experimental_results}

We compare the performance of five quantizers, with granularities $N = 2^8$ and $N = 2^{16}$, on three types of datasets of various alphabet sizes: 
\begin{itemize}%\jennifer{footnote to text}
\item Random synthetic distributions drawn from the uniform prior over the simplex: {We draw and take the average over 1000 random samples for our results.} 
\item Frequency of words in books: {These frequencies are computed from text available on the Natural Language Toolkit (NLTK) libraries for Python. For each text, we get tokens (single words or punctuation) from each text and simply count the occurrence of each token} 
\item Frequency of $k$-mers in DNA: {For a given sequence of DNA, the set of $k$-mers are the set of length $k$ substrings which appear in the sequence. We use the human genome as the source for our DNA sequences. Parts of the sequence marked as repeats are removed.} 
\end{itemize}
Our quantizers are:

\begin{itemize}
    \item \textbf{Approximate Minimax Compander:} As given by equation \eqref{eq::appx-minimax-compander}. Using the approximate minimax compander is much simpler than the minimax compander since the constant $c_\az$ does not need to be computed.%, and by \Cref{thm::approximate-minimax-compander} it has almost identical performance for large $\az$. NOT TRUE FOR K = 10^5
    \item \textbf{Truncation:} Uniform quantization (equivalent to $\comp(x) = x$), which truncates the least significant bits. This is the natural way of quantizing values in $[0,1]$. 
    \item \textbf{Float and bfloat16:} For 8-bit encodings ($N = 2^8$), we use a floating point implementation which allocates 4 bits to the exponent and 4 bits to the mantissa. For 16-bit encodings ($N = 2^{16}$), we use bfloat16, a standard which is commonly used in machine learning~\cite{kalamkar2019study}.
    \item \textbf{Exponential Density Interval (EDI):} This is the quantization method we used in an achievability proof in  \cite{adler_ratedistortion_2021}. It is designed for the uniform prior over the simplex.
    \item \textbf{Power Compander:} Recall that the compander is $ \comp(x) = x^{s}$. We optimize $s$ and find that $s = \frac{1}{\log_e \az}$ asymptotically minimizes KL divergence, and also gives close to the best performance among power companders empirically. To see the effects of different powers $s$ on the performance of the power compander, see \Cref{fig::books_power}.
\end{itemize}

Because a well-defined prior does not always exist for these datasets (and for simplicity) we use midpoint decoding for all the companders. When a probability value of exactly $0$ appears, we do not use companding and instead quantize the value to $0$, i.e. the value $0$ has its own bin.

\begin{figure}
    \centering
    \includegraphics[scale = .4, trim = {25 0 0 0}
    ]{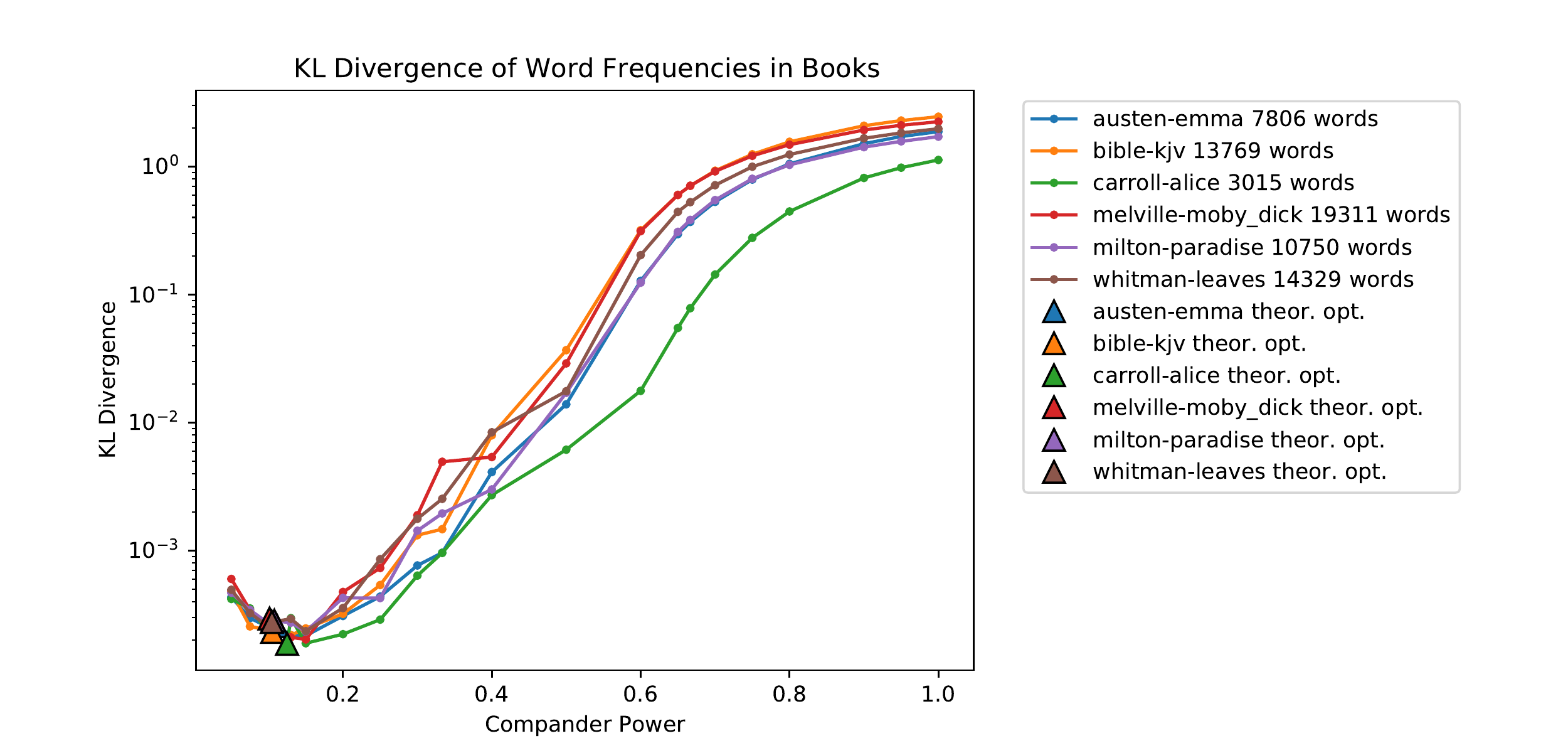}
    \caption{Power compander $\comp(x) = x^s$ performance with different powers $s$ used to quantize frequency of words in books. The number $\az$ of distinct words in each book is shown in the legend. The theoretical optimal power $s = \frac{1}{\log \az}$ is plotted.}
    \label{fig::books_power}
\end{figure}

Our main experimental results are given in \Cref{fig::first_compare_compander}, showing the KL divergence between the empirical distribution $\bx$ and its quantized version $\bnormvar$ versus alphabet size $\az$. The approximate minimax compander performs well against all sources. %Every distribution, regardless of the source, behaves similarly under the minimax compander.
%
%gives the smallest KL divergence loss on the worst distribution of those tested
\begin{figure}
    \centering
    \includegraphics[scale = .4]{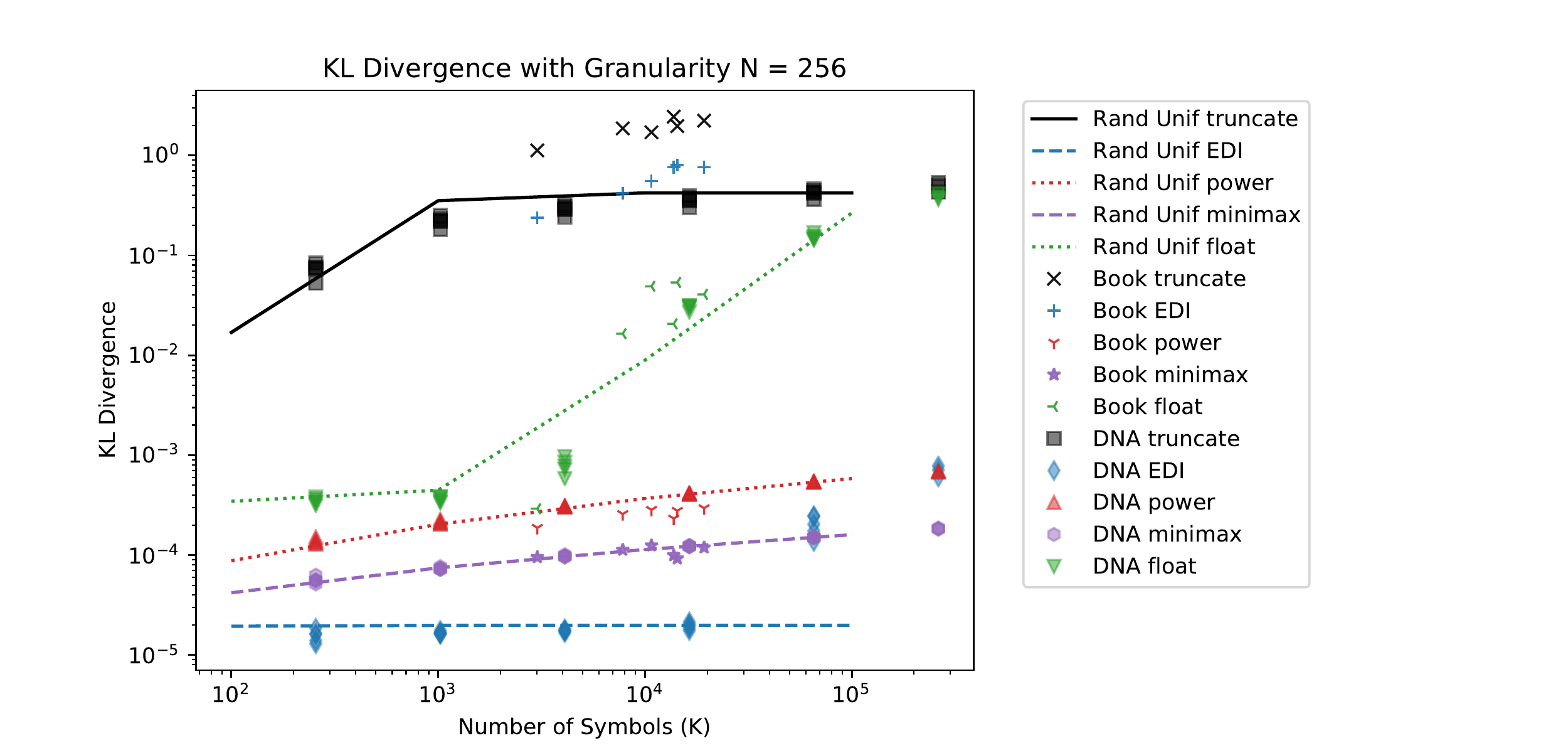}
    \includegraphics[scale = .4]{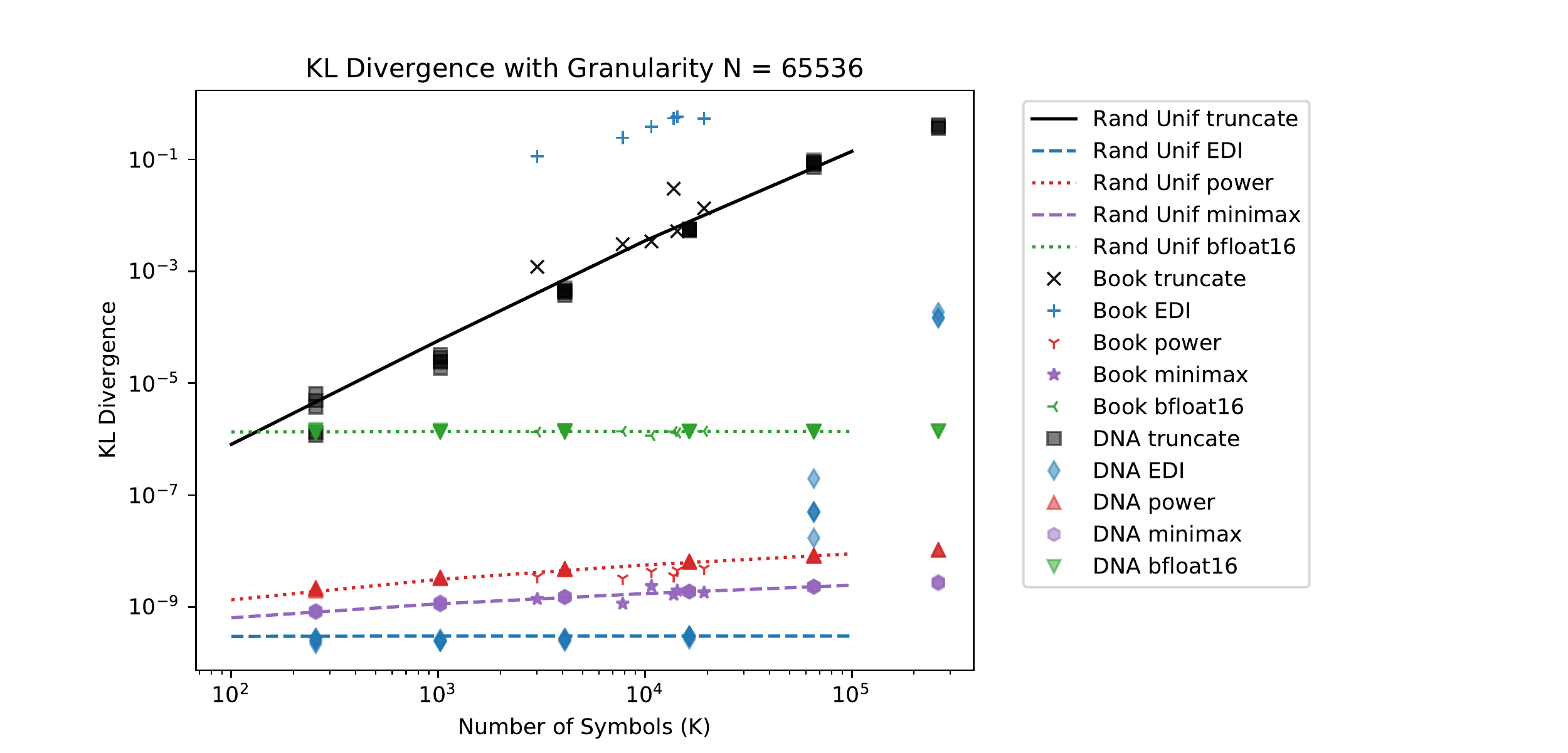}
    \caption{Plot comparing the performance of the truncation compander, the EDI compander, floating points, the power compander, and the approximate minimax compander \eqref{eq::appx-minimax-compander} on probability distributions of various sizes.  %The granularity $N$ is the number of level each scalar value is quantized to. %Hence for the whole probability vector, we need $N^\az$ quantization points. 
    }
    \label{fig::first_compare_compander}
    \vspace{-1pc}
\end{figure}
For truncation, the KL divergence increases with $\az$ and is generally fairly large. %When $N = 256$, the average KL divergence loss increases until it gets clipped. 
The EDI quantizer works well for the synthetic uniform prior (as it should), but for real-world datasets like word frequency in books, it performs badly (sometimes even worse than truncation).
The loss of the power compander is similar to the minimax compander (only worse by a constant factor), as predicted by \Cref{thm::power_compander_results}.

The experiments show that the approximate minimax compander achieves low loss on the entire ensemble of data (even for relatively small granularity, such as $N = 256$) and outperforms both truncation and floating-point implementations on the same number of bits. Additionally, its closed-form expression (and entrywise application) makes it simple to implement and computationally inexpensive, so it can be easily added to existing systems to lower storage requirements at little or no cost to fidelity. %Therefore, we advocate for use of the minimax compander for storage systems.

%Currently, most systems which store probability values simply use standard floating-point values; while very precise, this requires 32 or 64 bits per letter. Our theoretical and experimental results show that using the approximate minimax compander can achieve low KL divergence loss across a wide variety of datasets, even with much fewer bits per value (8 or 16 in our experiments); additionally, its closed-form expression (and entrywise application) makes it simple to implement and computationally inexpensive. Thus it can be easily added to existing systems to lower storage requirements at little or no cost to fidelity.

%%%%%% NEEDED FOR JOURNAL ONLY!!!!

%Due to space contraints, \Cref{sec::appendix_proofs_asymptotic} and \Cref{sec::appendix_other_companders} are included as supplementary materials.

%%%%%%%%%%%%%%%%%%%%%%

\subsection{Paper Organization}
We provide background and discuss previous work on companders in \Cref{sec::previous_works}. %\Cref{sec::asymt_single} is dedicated to exploring \eqref{eq::raw_loss}. 
We prove \Cref{thm::asymptotic-normalized-expdiv} in \Cref{sec::asymt_single} (though proofs of some lemmas and propositions leading up to it are given in Appendix~\ref{sec::appendix_proofs_asymptotic}). \Cref{prop::approximate-compander} is proved in Appendix~\ref{sec::approximate-compander-f-dagger}.
In \Cref{sec::minimax}, we optimize over  \eqref{eq::raw_loss} to get the maximin single-letter distribution (showing part of \Cref{prop::maximin-density} with other parts left to Appendix~\ref{sec::minimax_const_analysis}) and the minimax compander, thus showing \Cref{thm::optimal_compander_loss,thm::minimax_compander}, \Cref{cor::worstcase_prior} and \Cref{prop::bound_worstcase_prior_exist} (leaving \Cref{thm::approximate-minimax-compander} for Appendix~\ref{sec::proof_L_appx_minimax_compander}).
We prove \Cref{thm::worstcase_power_minimax} and the worst-case part of \Cref{thm::power_compander_results} in Appendix~\ref{sec::worst-case_analysis}. Other parts of \Cref{thm::power_compander_results} are discussed in Appendix~\ref{sec::power_compander_analysis}. In \Cref{sec::other_losses} we discuss companders for losses other than KL divergence. Finally, in \Cref{sec::info_distillation_main} we discuss a connection of our problem to the problem of information distillation with proofs given in \Cref{sec::info_distillation_detail}. %(The appendices are included in the supplementary material.)

\section{Background}

\label{sec::previous_works}

Companders (also spelled ``compandors'') were introduced by Bennett in 1948 \cite{bennett1948} as a way to quantize speech signals, where it is advantageous to give finer quantization levels to weaker signals and coarser levels to larger signals. Bennett gives a first order approximation that the mean-square error in this system is given by
\begin{align}\label{eq::bennett_compander}
    \frac{1}{12 N^2} \int_{a}^b \frac{p(x)}{(f'(x))^2} dx
\end{align}
where $N$ is the number quantization levels, $a$ and $b$ are the minimum and maximum values of the input signal, $p$ is the probability density of the input signal, 
and $f'$ is the slope of the compressor function placed before the uniform quantization.
This formula is similar to our \eqref{eq::raw_loss} except that we have an extra $x^{-1}$ since we are working with KL divergence. 
Others have expanded on this line of work. In \cite{panter_dite}, the authors studied the same problem and determined the optimal compressor under mean-square error, a result which parallels our result \eqref{eq::raw_overall_dist}. However, results like those in \cite{bennett1948, panter_dite} are stated either as first order approximations or make simplifying assumptions. For example, in \cite{panter_dite}, the authors state that they assume the values $\widehat{y}_{(n)}$ are close together enough that probability density within any given bin can be treated as a constant.
%``Suppose that the levels are close that $P(y)$ may be considered as nearly constant over the region of integration and equal to $P(y_n)$''. \jennifer{}
%While previous work assumes that the probability over each bin can be approximated by a uniform distribution, 
In contrast, we rigorously show that this fundamental logic holds under very general conditions ($f \in \compset^\dagger$).

Generalizations of Bennett's formula are also studied when instead of mean-square error, the loss is the expected $r$th moment loss $\bbE\norm{\cdot}^r$. This is computed for vectors of length $\az$ in \cite{zador1982} and \cite{gersho1979}. %However, this analysis is not extended to finding minimax companders for these settings.\reversemarginpar\aviv{}

The typical examples of companders used in engineering and signals processing are the $\mu$-law and $A$-law companders \cite{lewis_mu-law}. For the $\mu$-law compander, \cite{panter_dite} and \cite{smith1957} argue that for mean-squared error, for a large enough constant $\mu$ the distortion becomes independent of the signal.

Quantizing probability distributions is a well-studied topic, though typically the loss function is a norm and not KL divergence \cite{graf2007}. Quantizing for KL divergence is considered in our earlier work \cite{adler_ratedistortion_2021}, focusing on average KL loss for Dirichlet priors. 

A similar problem to quantizing under KL divergence is \emph{information $k$-means}. This is the problem of clustering $n$ points $a_i$ to $k$ centers $\hat{a}_j$ to minimize the KL divergences between the points and their associated centers. Theoretical aspects of this are explored in \cite{slonim1999} and \cite{tishby2000}. 
Information $k$-means has been implemented for several different applications \cite{Pereira93, jiang2013, cao2013}. There are also other works that study clustering with a slightly different but related metric \cite{dhillon2003,nielsen2013, veldhuis2002}; however, the focus of these works is to analyze data rather than reduce storage. 

\begin{remark}
A variant of the classic problem of prediction with log-loss is an equivalent formulation to quantizing the simplex with KL loss: let $\bx \in \triangle_{\az-1}$ and $A \sim \bx$ (in the alphabet $[\az]$); we want to predict $A$ by positing a distribution $\bnormvar \in \triangle_{\az-1}$, and our loss is $-\log \normvar_A$. In the standard version, the problem is to pick the best $\bnormvar$ given limited information about $\bx$; however, if we \emph{know} $\bx$ but are required to express $\bnormvar$ using only $\log_2 M$ bits, it is equivalent to quantizing the simplex with KL divergence loss.
\end{remark}

%\newpage

\section{Asymptotic Single-Letter Loss} \label{sec::asymt_single}

In this section we give the proof of \Cref{thm::asymptotic-normalized-expdiv} (though the proofs of some lemmas must be sketched). We use the following notation: 
%\begin{itemize}
%\item 

Given an interval $I$ we define $\bar{y}_I$ to be its midpoint and $r_I$ to be its width, so that by definition
\begin{align}
I = [\bar{y}_I - r_I/2, \bar{y}_I + r_I/2]\,.
\end{align}
Note that if $I \subseteq [0,1]$ then $r_I \leq 2 \bar{y}_I$.
%\item 

%MOVED TO APPENDIX
%Given probability distribution $p$ and interval $I$, $p|_I$ denotes $p$ restricted to $I$, i.e. $X \sim p|_I$ is the same as $X \sim p$ conditioned on $X \in I$.  We also define the probability mass of $I$ under $p$ as $\pi_{p,I} = \bbP_{X \sim p}[ X \in I]$. If $\pi_{p,I} = 0$, we let $p|_I$ be uniform on $I$ by default.

%\item 

Given probability distribution $p$ and interval $I$, we denote the following: $p|_I$ is $p$ restricted to $I$; $\pi_{p,I} := \bbP_{X \sim p}[X \in I]$ is the probability mass of $I$; and the \emph{centroid of $I$ under $p$} is
\begin{align}
\widetilde{y}_{p,I} := \bbE_{X \sim p|_I} [X] = \bbE_{X \sim p}[X \, | \, X \in I]\,.
\end{align}
If they are undefined because $\bbP_{X \sim p}[X \in I] = 0$ then by convention $p|_I$ is uniform on $I$ and $\widetilde{y}_{p, I} = \bar{y}_I$.

When $I = I^{(n)}$ is a bin of the compander, we can replace it with $(n)$ in the notation, i.e. $\bar{y}_{(n)} = \bar{y}_{I^{(n)}}$ (so the midpoint of the bin containing $x$ at granularity $N$ is denoted $\bar{y}_{(n_N(x))}$ and the width of the bin is $r_{(n_N(x))}$). When $I$ and/or $p$ are fixed, we sometimes drop them from the notation, i.e. $\widetilde{y}_I$ or even just $\widetilde{y}$ to denote the centroid of $I$ under $p$. 

\subsection{The Local Loss Function}

One key to the proof is the following perspective: instead of considering $X \sim p$ directly, we (equivalently) first select bin $I^{(n)}$ with probability $\pi_{p, (n)}$, and then select $X \sim p|_{(n)}$. The expected loss can then be considered within bin $I^{(n)}$. This makes it useful to define:
\begin{definition}
Given probability measure $p$ and interval $I$, the \emph{single-interval loss of $I$ under $p$} is
\begin{align}
\ell_{p,I} = \bbE_{X \sim p|_{I}}[X \log(X/\widetilde{y}_{p,I})]\,.
\end{align}
\end{definition}
As before, if $p$ and/or $I$ is fixed and clear, we can drop it from the notation (and if $I = I^{(n)}$ is a bin, we can denote the local loss as $\ell_{p, (n)}$).
This can be interpreted as follows: if we quantize all $x \in I$ to the centroid $\widetilde{y}_I$, then $\ell_{p,I}$ is the expected loss of $X \sim p$ conditioned on $X \in I$. Thus the values of $\ell_{p,(n)}$ can be used as an alternate means of computing the single-letter loss:
\begin{align}
	\singleloss(p,f,N) &= \bbE_{X \sim p} [X \log(X/\widetilde{y}(X))]
	\\ &=  \sum_{n=1}^N \pi_{p, (n)} \bbE_{X \sim p|_{(n)}}[X \log(X/\widetilde{y}_{p,(n)})]
	\\ &=  \sum_{n=1}^N \pi_{p, (n)} \ell_{p, (n)}
	= \int_{[0,1]}  \ell_{p, (n_N(x))} \, dp\,.
\end{align}
Thus the normalized single-letter loss (whose limit is the asymptotic single-letter loss \eqref{eq::raw-assl}) is
\begin{align}
N^2 \, \singleloss(p,f,N) = \int_{[0,1]} N^2 \, \ell_{p, (n_N(x))} \, dp\,.
\end{align}
For single-letter density $p$ and compander $f$, we define the \emph{local loss function at granularity $N$}:
\begin{align} \label{eq::loc-loss-fn}
    g_N(x) = N^2 \, \ell_{p, (n_N(x))}\,.
\end{align}
We also define the \emph{asymptotic local loss function}:
\begin{align}
    g(x) = \frac{1}{24} f'(x)^{-2} x^{-1} \,.
\end{align}
%The function $g_N$ basically takes each $x$ and returns the expected loss for $X \sim p$ which fall in the same bin as $x$ and normalizes by multiplying by $N^2$. 
\Cref{thm::asymptotic-normalized-expdiv} is therefore equivalent to:
\begin{align}
\liminf_{N \to \infty} \int g_N \, dp &\geq \int g \, dp ~\forall~ p \in \cP, f \in \cF \label{eq::fatou-restated} 
\\ \text{and} \, \lim_{N \to \infty} \int g_N \, dp &= \int g \, dp ~\forall~ p \in \cP, f \in \cF^\dagger \label{eq::raw-restated}.
\end{align}
To prove \eqref{eq::fatou-restated} and \eqref{eq::raw-restated}, we show:

\begin{proposition} \label{prop::locloss-convergence}
For all $p \in \cP$, $\comp \in \compset$, if $X \sim p$ then
\begin{align}
    \lim_{N \to \infty} \locloss_N(X) = \locloss(X) ~~~\text{almost surely.}
\end{align}
\end{proposition}

\begin{proposition} \label{prop::locloss-dominating}
Let $f \in \cF^\dagger$ be a compander and $c > 0$ and $\alpha \in (0,1]$ such that $f(x) - c x^\alpha$ is monotonically increasing. Letting $g_N$ be the local loss functions as in \eqref{eq::loc-loss-fn} and
\begin{align}
     h(x) = (2^{2/\alpha} + \alpha^2 2^{1/\alpha - 2}) (c \alpha)^{-2} x^{1-2\alpha} +  c^{-1/\alpha} 2^{1/\alpha - 2}
\end{align}
then $g_N(x) \leq h(x)$ for all $x, N$. Additionally, if $\alpha \leq 1/2$ then $\int_{[0,1]} h \, dp < \infty$.
\end{proposition}

The lower bound \eqref{eq::fatou-restated} then follows immediately from \Cref{prop::locloss-convergence} and Fatou's Lemma; and when $f \in \cF^\dagger$, by \Cref{prop::locloss-dominating} there is some $h$ which is integrable over $p$ and dominates all $g_N$, thus showing \eqref{eq::raw-restated} by the Dominated Convergence Theorem.

To prove \Cref{prop::locloss-convergence}, we use the following:
\begin{itemize}
\item For any $x$ at which $f$ is differentiable, when $N$ is large, the width of the interval $x$ falls in is
\begin{align}
r_{(n_N(x))} \approx N^{-1} f'(x)^{-1} \, .
\end{align}
\item For any $x$ at which $F_p$ is differentiable, $p|_I$ will be approximately uniform over any sufficiently small $I$ containing $x$.
\item For a sufficently small interval $I$ containing $x$ and such that $p|_I$ is approximately uniform,
\begin{align}\label{eq::local_loss_approx}
\ell_{p,I} \approx \frac{1}{24} r_I^2 x^{-1}\,.
\end{align}
\end{itemize}
Putting these together, we get that if $F_p$ and $f$ are both differentiable at $x$ then when $N$ is large,
\begin{align}
	g_N(x) &= N^2 \, \ell_{p,(n_N(x))} \\
	&\approx N^2 \frac{1}{24} r_{(n_N(x))}^2 x^{-1} \approx \frac{1}{24} f'(x)^{-2} x^{-1} = g(x)
\end{align}
as we wanted.  We formally state each of these steps in \Cref{sec::locloss-convergence-preliminaries} and combine them to prove \Cref{prop::locloss-convergence} in \Cref{sec::locloss-convergence-pf}.

The proof of \Cref{prop::locloss-dominating} is given in \Cref{sec::locloss-dominating-proof}, along with its own set of definitions and lemmas needed to show it.

%\newpage

\section{Minimax Compander}

\label{sec::minimax}

\Cref{thm::asymptotic-normalized-expdiv} showed that for $f \in \cF^\dagger$, the asymptotic single-letter loss is equivalent to
\begin{align}
    \singleloss(p, f) = \frac{1}{24} \int_0^1 p(x) f'(x)^{-2} x^{-1} dx\,.
\end{align}
Using this, we can analyze what is the `best' compander $f$ we can choose and what is the `worst' single-letter density $p$ in order to show \Cref{thm::optimal_compander_loss,thm::minimax_compander} and their related results.

\iflong
\subsection{Optimizing the Compander}

%In this section we use $h(x)$ auxiliary functions; note that these have nothing to do with the asymptotic local loss function or dominating function from the previous section.
\fi

We show \Cref{thm::optimal_compander_loss}, 
%which we show along with and \Cref{prop::approximate-compander}. 
which follows from \Cref{thm::asymptotic-normalized-expdiv} by finding $\comp \in \compset$ which minimizes $L^\dagger(p,\comp)$. This is achieved by optimizing over $\compder$; we will also use some concepts from \Cref{prop::approximate-compander} to connect it back to $\inf_{\comp \in \cF} \widetilde{L}(p,\comp)$ when the resulting $\comp$ is not in $\cF^\dagger$. Since $\comp : [0,1] \to [0,1]$ is monotonic, we use constraints $\compder(x) \geq 0$ and $\int_0^1 \compder(x) \, dx = 1$.
\iflong
We solve the following:
\begin{align}
    \text{minimize } &L^\dagger(p,\comp) = \frac{1}{24}\int_0^1 p(x) \compder(x)^{-2} x^{-1} \, dx \\
    \text{subject to } &\int_0^1 \compder(x) \, dx = 1 \\&\text{ and } \compder(x) \geq 0 \text{ for all } x \in [0,1] \,.
\end{align}
The function $L^\dagger(p,\comp)$ is convex in $\compder$, and thus first order conditions show optimality. Let $\lambda(x)$ satisfy $\int_0^1 \lambda(x) dx = 0$. If $\compder(x) \, \propto \, (p(x)x^{-1})^{1/3}$, we derive:
\begin{align}
    &\frac{d}{dt}  \frac{1}{24}  \int_0^1 p(x) \big(\compder(x) + t \, \lambda(x) \big)^{-2} x^{-1} \, dx \label{eq::perturbed_func_min_L} \\
    &= \frac{1}{24} \int_0^1 p(x) x^{-1} \frac{d}{dt} \big(\compder(x) + t \, \lambda(x) \big)^{-2}  \, dx \\
    &= -\frac{1}{12} \int_0^1 p(x) x^{-1} \big(\compder(x) + t \, \lambda(x) \big)^{-3} \lambda(x)  \, dx \\
     &= -\frac{1}{12} \int_0^1 p(x) x^{-1} \compder(x)^{-3} \lambda(x)  \, dx  ~~(\text{at } t = 0) \,
     \\ & \propto \, -\frac{1}{12} \int_0^1  \lambda(x)  \, dx = 0  \label{eq::f_proportional_to}\,.
\end{align}
%If $\comp$ is such that 
%\begin{align}\label{eq::f_proportional_to}
%\compder(x) \, \propto \, (p(x)x^{-1})^{1/3}
%\end{align}
Thus, such $f$ satisfies the first-order optimality condition under the constraint $\int f'(x) \, dx = 1$. This gives
\else
Using calculus of variations, we get
\fi
$
    \compder_p(x) \, \propto \, (p(x)x^{-1})^{1/3}
$
and $\comp(0) = 0$ and $\comp(1) = 1$, from which \eqref{eq::raw_overall_dist} and \eqref{eq::best_f_raw} follow. 
If $\comp_p \in \compset^\dagger$, then $\comp_p = \argmin_f \singleloss(p,\comp)$, and for any other $\comp \in \compset$,
\begin{align}
    \hspace{-0.5pc} \singleloss(p,\comp_p) &= L^\dagger(p,\comp_p)
     \leq L^\dagger(p,\comp)
     \\&\leq \liminf_{N \to \infty} N^2 \singleloss(p,\comp,N) 
     \,.
\end{align}

If $\comp_p \not \in \compset^\dagger$, for any $\delta > 0$ define $\comp_{p,\delta} = (1 - \delta) \comp_p + \delta x^{1/2}$ (as in \eqref{eq::approximate-compander}). Then $\comp_{p,\delta} - \delta x^{1/2} = (1 - \delta) \comp_p$ is monotonically increasing so $\comp_{p,\delta} \in \compset^\dagger$, so \Cref{thm::asymptotic-normalized-expdiv} applies to $\comp_{p,\delta}$; additionally, $\comp_{p,\delta} - (1-\delta) \comp_p = \delta x^{1/2}$ is monotonically increasing as well so $\comp'_{p,\delta} \geq (1-\delta) \comp'_p$. Hence, plugging into the $L^\dagger$ formula gives:
\begin{align}
 &\singleloss(p,\comp_{p,\delta}) = L^\dagger(p,\comp_{p,\delta}) \leq L^\dagger(p,\comp_p)(1-\delta)^{-2}  \, .
\end{align}
Taking $\delta \to 0$ (and since $\compset^\dagger \subseteq \compset$) shows that 
\begin{align}
    L^\dagger(p,\comp_p) = \inf_{\comp \in \compset^\dagger} \singleloss(p,\comp)\,, %\leq L^\dagger(p,\comp_p)
\end{align} 
finishing the proof of \Cref{thm::optimal_compander_loss}. 
%and \Cref{prop::approximate-compander}.

\begin{remark}
Since we know the corresponding single-letter source $p$ for a Dirichlet prior, using this $p$ with \Cref{thm::optimal_compander_loss} gives us the optimal compander for Dirichlet priors on any alphabet size. This gives us a better quantization method than EDI which was discussed in \Cref{{sec::experimental_results}}. This optimal compander for Dirichlet priors is called the \emph{beta compander} and its details are given in Appendix~\ref{sec::beta_companding}.
\end{remark}

\iflong
\subsection{The Minimax Companders and Approximations}
\fi

To prove \Cref{thm::minimax_compander} and \Cref{cor::worstcase_prior}, we first consider what density $p$ maximizes equation \eqref{eq::raw_overall_dist}:
\begin{align}
\frac{1}{24} \left(\int_0^1 (p(x)x^{-1})^{1/3} dx\right)^3
\end{align}
i.e. is most difficult to quantize with a compander.
Using calculus of variations to maximize
\begin{align}
    \int_0^1 (p(x)x^{-1})^{1/3} \,dx \label{eq::loss-to-maximize}
\end{align}
(which of course maximizes \eqref{eq::raw_overall_dist}) subject to $p(x) \geq 0$ and $\int_0^1 p(x) \, dx = 1$, we find that maximizer is $p(x) = \frac{1}{2} x^{-1/2}$. However, while interesting, this is only for a single letter; and because $\bbE[X] = 1/3$ under this distribution, it is clearly impossible to construct a prior over $\triangle_{\az-1}$ (whose output vector \emph{must} sum to $1$) with this marginal (unless $\az = 3$). 

%Instead, since \eqref{eq::loss-to-maximize} is concave in $p$, the expected raw loss produced by any prior $P \in \cP^\triangle_\az$ (the sum of the expected raw losses over each letter) is never decreased when symmetrizing $P$ (say, by randomly permuting the letter labels). Thus, it is sufficient to consider symmetric $P$, whose marginals are all $p$; since $P$ returns probability vectors, $\bbE_{X \sim p}[X] = 1/\az$. 

\iflong
Hence, we add an expected value constraint to the problem of maximizing \eqref{eq::loss-to-maximize}, giving:
\begin{align}
    \text{maximize } &\int_0^1 \big(p(x)x^{-1}\big)^{1/3} \, dx  \\
    \text{subject to } &\int_0^1 p(x) \, dx = 1\label{eq::p_constraint_sum}; \\ 
    &\int_0^1 p(x)x \,dx = \frac{1}{\az} \label{eq::p_constraint_mean};\\
    &\text{and } p(x) \geq 0 \text{ for all } x\,.
\end{align}
We can solve this again using variational methods (we are maximizing a concave function so we only need to satisfy first-order optimality conditions). A function $p(x) > 0$ is optimal if, for any $\lambda(x)$ where
\begin{align}
    \int_0^1 \lambda(x) \, dx = 0 \text{ and } \int_0^1 \lambda(x)x \, dx &= 0 \label{eq::variational_constraints}
\end{align}
the following holds:
\begin{align}
    \frac{d}{dt} \int_0^1 x^{-1/3} \big(p(x) + t \, \lambda(x)\big)^{1/3} \, dx = 0\,.
\end{align}
%(we need $p(x) > 0$ to distance the function from boundary of the nonnegativity constraint). 
We have by the same logic as before:
\begin{align}
    \frac{d}{dt} \eqstartshort \int_0^1  x^{-1/3} \big(p(x) + t \, \lambda(x)\big)^{1/3} \, dx \eqbreakshort
    &= \frac{1}{3} \int_0^1 x^{-1/3} \big(p(x) + t \, \lambda(x)\big)^{-2/3} \lambda(x) \, dx \\
    &= \frac{1}{3} \int_0^1 x^{-1/3} p(x)^{-2/3} \lambda(x) \, dx ~~ (\text{at } t = 0)\,. \label{eq::variational_condition}
\end{align}
%The condition \eqref{eq::variational_condition} must hold for all $\lambda(x)$ satisfying constraints \eqref{eq::variational_constraints}. 
Thus, if we can arrange things so that there are constants $a_\az, b_\az$ such that 
\begin{align}
    x^{-1/3} p(x)^{-2/3} = a_\az  + b_\az  x
\end{align}
this ensures \eqref{eq::variational_condition} equals zero. In that case,
\begin{align}
    x^{-1/3} p(x)^{-2/3} &= a_\az  + b_\az  x \\
    \iff \hspace{1pc} p(x)^{-2/3} &= a_\az  x^{1/3} + b_\az  x ^{4/3} \\
    \iff \hspace{2.6pc}  p(x) &= \big(a_\az  x^{1/3} + b_\az  x^{4/3}\big)^{-3/2}\label{eq::maximin-density_repeat}\,.
\end{align}
This is the maximin density $p^*_\az$ from \Cref{prop::maximin-density} \eqref{eq::maximin-density}, where $a_\az , b_\az $ are set to meet the constraints \eqref{eq::p_constraint_sum} and \eqref{eq::p_constraint_mean}. 
%
%In terms of $a$,
%\begin{align}
%    b = 4 a^{-2} - a\,.
%\end{align}
%
Exact formulas for $a_\az, b_\az$ are difficult to find; we give more details on after the next step.

We want to determine the optimal compander for the maximin density \eqref{eq::maximin-density_repeat}. We know from
\eqref{eq::f_proportional_to} that we need to first compute
\begin{align} 
    \phi(x) &= \int_0 ^{x} \newvar^{-1/3} \left(a_\az  \newvar^{1/3} + b_\az  \newvar^{4/3} \right)^{-1/2} \,d\newvar 
     \\ & = \frac{2 \arcsinh \left(\sqrt{\frac{b_\az  x}{a_\az }} \right)}{\sqrt{b_\az }}\,. \label{eq::minimax-compander-proportional}
\end{align}

The best compander $f(x)$ is proportional to \eqref{eq::minimax-compander-proportional} and is exactly given by $f(x) = \phi(x)/\phi(1)$. The resulting compander, which we call the \emph{minimax compander}, is
\begin{align} \label{eq::arcsinh-general}
    f(x) &= \frac{\arcsinh \left(\sqrt{\frac{b_\az  x}{a_\az }} \right)}{\arcsinh \left(\sqrt{\frac{b_\az  }{a_\az }} \right)} \,.
\end{align}
Given the form of $f(x)$, it is natural to
determine an expression for the ratio $b_\az / a_\az$.  We can parameterize both $a_\az$ and $b_\az$ by $b_\az / a_\az$ and then examine how $b_\az / a_\az$ behaves as a function of $\az$.
The constraints on $a_\az$ and $b_\az$ give that
\begin{align}
    a_\az &= 4^{1/3}(b_\az / a_\az + 1)^{-1/3}
    \\b_\az  &= 4 a_\az^{-2} - a_\az\,.
\end{align}
The ratio $b_\az / a_\az$ grows approximately as $\az \log \az$. Hence, we choose to parameterize
\begin{align}
    b_\az / a_\az = c_\az \az \log \az\,.
\end{align}
To satisfy the constraints, we get $.25 < c_\az < .75$ so long as $\az > 24$ (see \Cref{sec::minimax_const_analysis} for details), and \Cref{lem::c_k_half} in \Cref{sec::lim_c_k} shows that $c_\az \to 1/2$ as $\az \to \infty$.
Combining these gives \Cref{prop::maximin-density}.

We can then express $a_\az$, $b_\az$ in terms of $c_\az$:
\begin{align}
    a_\az &= 4^{1/3}(c_\az  \az \log \az + 1)^{-1/3}
    \\b_\az  &= 4 a_\az^{-2} - a_\az
    \\ &= 4^{1/3} (c_\az  \az \log \az + 1)^{2/3} \eqlinebreakshort - 4^{1/3}(c_\az  \az \log \az + 1)^{-1/3}
    \label{eq::value_c_2_worst+p}
    \\& = 4^{1/3}(c_\az \az \log \az)^{2/3} (1+o(1))\,.
\end{align}
When $\az$ is large, the second term in \eqref{eq::value_c_2_worst+p} is negligible compared to the first. Thus, plugging into \eqref{eq::arcsinh-general} we get the {minimax compander} and {approximate minimax compander}, respectively:
\begin{align}
    \comp^*_\az(x) &= \frac{\arcsinh \left(\sqrt{(c_\az \az \log \az) x} \right)}{\arcsinh \left(\sqrt{c_\az \az \log \az} \right)} 
    \\ \approx \comp^{**}_\az(x) &= \frac{\arcsinh (\sqrt{((1/2) \az \log \az) x} )}{\arcsinh (\sqrt{(1/2) \az \log \az} )} \,.
\end{align}
The minimax compander minimizes the maximum (raw) loss against all densities in $\cP_{1/\az}$, while the approximate minimax compander performs very similarly but is more applicable since it can be used without computing $c_\az$.

To compute the loss of the minimax compander, we can use \eqref{eq::raw_overall_dist} to get 
\begin{align}
    L^{\dagger}(p_\az^*, \comp^*_\az) 
    & = \frac{1}{\dpfconst} \left(\frac{2 \arcsinh\left(\sqrt{c_\az  \az \log \az} \right)} {\sqrt{b_\az }} \right)^3\,.
\end{align}
Substituting we get
\begin{align}
      L^{\dagger}\eqstartshort(p_\az^*, \comp^*_\az)  
      \\ &= \frac{1}{\dpfconst} \frac{8 \left(\log \left(\sqrt{c_\az \az \log \az} +  \sqrt{c_\az  \az \log \az + 1} \right)\right)^3}{2 c_\az  \az \log \az (1+o(1))}\\
     & = \frac{1}{\dpfconst} \frac{(\log 4 (c_\az  \az \log \az))^3}{2 c_\az  \az \log \az} (1+o(1)) \\
     %& = \frac{1}{48} \left(\frac{\log^2 \az}{c_\az  \az} + \frac{O(\log \az)}{2c_\az  \az} \right)\\
     & =  \frac{1}{24} \frac{\log^2 \az}{\az} (1 + o(1)) \label{eq::L_dagger_value_minimax}\,.
\end{align}

\else

Hence, we add the constraint $p \in \cP_{1/\az}$, maximizing \eqref{eq::loss-to-maximize} subject to $p(x) \geq 0$, $\int_0^1 p(x) \, dx = 1$, and $\int_0^1 p(x) x \, dx = 1/\az$. Solving this modified problem yields the \emph{maximin density} $p^*_\az$ \eqref{eq::maximin-density} from \Cref{thm::minimax_compander}; $\comp^*_\az$ \eqref{eq::minimax-compander} is simply the optimal compander (given by \eqref{eq::best_f_raw}) against $p^*_\az$. 
\fi

In fact, not only is $\comp^*_\az$ optimal against the maximin density $p^*_\az$, but (as alluded to in the name `minimax compander') it minimizes the maximum asymptotic loss over all $p \in \cP_{1/\az}$. 
More formally we show that $(\comp^*_\az, p^*_\az)$ is a saddle point of $L^\dagger$.

The function $L^\dagger(p,\comp)$ is concave (actually linear) in $p$ and convex in $\compder$, and we can show that the pair $(\comp^*_\az, p^*_\az)$ form a saddle point, thus proving \eqref{eq::minmax}-\eqref{eq::maxmin} from \Cref{thm::minimax_compander}.
%\begin{align}
%    \comp^*_\az = \argmin_{\comp \in \compset} \sup_{p \in \cP_{1/\az}} L^\dagger(p,\comp) \,.
%\end{align}

\iflong
We can compute that
\begin{align}
    (\comp_\az^*)'(x)  \, &\propto \, (p_\az^*(x) x^{-1})^{1/3} \\
    & = x^{-1/3} (a_\az  x^{1/3} + b_\az  x^{4/3})^{-1/2} \\
    %& = (x^{2/3})^{-1/2}(a_\az  x^{1/3} + b_\az  x^{4/3})^{-1/2}\\
    & = \frac{1}{\sqrt{a_\az  x + b_\az  x^2}}\,.
\end{align}
Assume we set $a_\az $ and $b_\az $ to the appropriate values for $\az$.
For any $p \in \cP_{1/\az}$,
\begin{align}
     {L^{\dagger}}(p, \comp_\az^*) & = \int_0^1 p(x) x^{-1} ((\comp_\az^*)'(x))^{-2}  dx\\
    & = \int_0^1 p(x) x^{-1} (a_\az  x + b_\az  x^2) dx \\
    %& = \int_0^1 p(x) (a_\az  + b_\az  x) dx \\
    & = a_\az  + b_\az  \frac{1}{\az}
\end{align}
i.e. $L^\dagger(p,\comp^*_\az)$ does not depend on $p$. Since $\comp^*_\az$ is the optimal compander against the maximin compander $p^*_\az$ we can therefore conclude:
\begin{align}
\sup_{p \in \cP_{1/\az}} {L^{\dagger}}(p, \comp_\az^*)  &= {L^{\dagger}}(p_\az^*, \comp_\az^*)
\\ = \inf_{\comp \in \compset} {L^{\dagger}}(p_\az^*, \comp) &= \sup_{p \in \cP_{1/\az} }\inf_{\comp \in \compset} {L^{\dagger}}(p, \comp)\,. 
\end{align}
Since it is always true that
\begin{align}\label{eq::minimax_ineq}
      \sup_{p \in \cP_{1/\az}} \inf_{\comp \in \compset}  {L^{\dagger}}(p, \comp) \leq   \inf_{\comp \in \compset} \sup_{p \in \cP_{1/\az}}  {L^{\dagger}}(p, \comp) \,,
\end{align}
this shows that $(\comp_\az^*, p^*_\az)$ is a saddle point.
\fi

Furthermore, $\comp_\az^* \in \compset^\dagger$ (specifically it behaves as a multiple of $x^{1/2}$ near $0$), so $\singleloss(p,\comp_\az^*) = L^\dagger(p,\comp_\az^*)$ for all $p$, thus showing that $\comp_\az^*$ performs well against any $p \in \cP_{1/\az}$. Using \eqref{eq::raw_loss} with the expressions for $p^*_\az$ and $\comp_\az^*$ and \eqref{eq::L_dagger_value_minimax} gives \eqref{eq::raw_loss_saddle}. This completes the proof of \Cref{thm::minimax_compander}.

\begin{remark}
While the power compander $\comp(x) = x^{1/\log \az}$ is not minimax optimal, it has similar properties to the minimax compander and differs in loss by at most a constant factor. We analyze the power compander in \Cref{sec::power_compander_analysis}.
\end{remark}

\iflong
\subsection{Existence of Priors with Given Marginals}
\fi

%\begin{remark}
%It is possible for a (non-symmetric) prior $P$ to have $p(x) = \frac{1}{2}x^{-1/2}$ as a marginal, and (as we show) the expected raw loss for this letter will be high. However, because $\bbE_{X \sim p}[X] = 1/3$, this means that to compensate, the \emph{other} letters will need to have smaller expected values and be more concentrated and easier to quantize. Hence $P$ will not be harder to quantize overall than the symmetrized version where every marginal has expected value $1/\az$.
%\end{remark}

While $p^*_\az$ is the most difficult density in $\cP_{1/\az}$ to quantize, it is unclear whether a prior $P^*$ on $\triangle_{\az-1}$ exists with marginals $p^*_\az$ -- even though $\az$ copies of $p^*_\az$ will correctly sum to $1$ in expectation, it may not be possible to correlate them to guarantee they sum to $1$. However, it is possible to construct a prior $P^*$ whose marginals are as hard to quantize, up to a constant factor, as $p^*_\az$, by use of clever correlation between the letters. We start with a lemma:

\begin{lemma}\label{lem::cramming2}
Let $p \in \cP_{1/\az}$. Then there exists a joint distribution of $(X_1, \dots, X_\az)$ such that (i) $X_i \sim p$ for all $i \in [\az]$ and (ii) $\sum_{i \in [\az]} X_i \leq 2$, guaranteed.
\end{lemma}

\iflong
\begin{proof}
Let $F$ be the cumulative distribution function of $p$. Define the quantile function $F^{-1}$ as
\begin{align}
    F^{-1}(u) = \inf \{x : F(x) \geq u\}.
\end{align}

We break $[0,1]$ into $\az$ uniform sub-intervals $I_i = ((i-1)/\az, i/\az]$ (let $I_1 = [0, 1/\az]$). We then generate $X_1, X_2, \dots, X_\az$ jointly by the following procedure:
\begin{enumerate}
    \item Choose a permutation $\sigma : [\az] \to [\az]$ uniformly at random (from $\az!$ possibilities).
    \item Let $U_k \sim \unif_{I_{\sigma(k)}}$ independently for all $k$.
    \item Let $X_k = F^{-1}(U_k)$.
\end{enumerate}

Now we consider $\sum_k X_k$. Let $b_i = F^{-1}(i/k)$ for $i = 0, 1, \dots, \az$. Note that if $\sigma(k) = i$ then $U_k \in ((i-1)/\az, i/\az]$ and hence $X_k = F^{-1}(U_k) \in [b_{i-1}, b_i]$. Therefore $X_{\sigma^{-1}(i)} \in [b_{i-1}, b_i]$ and thus for any permutation $\sigma$,
\begin{align}
    \sum_{i=1}^\az b_{i-1} &\leq \sum_{i=1}^\az X_{\sigma^{-1}(i)} \leq \sum_{i=1}^\az b_i \\
    &= \Big(\sum_{i=1}^\az b_{i-1}\Big) + b_\az - b_0\\
    &\leq \Big(\sum_{i=1}^\az b_{i-1}\Big) + 1  \leq 2
\end{align}
as $\sum_i b_{i-1} \leq \sum_i \bbE[X_{\sigma^{-1}(i)}]
     = \az \bbE_{X \sim p}[X] = 1$.
\end{proof}
\fi
\Cref{lem::cramming2} shows a joint distribution of $\newVar_1, \dots, \newVar_{\az-1}$ such that $\newVar_i \sim p^*_\az$ for all $i$ and $\sum_{i=1}^{\az-1} \newVar_i \leq 2$ (guaranteed) exists.
%(a distribution where $\bbE_{X \sim p}[X] \leq 1/(\az-1)$). 
\iflong
Then, if $X_i = \newVar_i/2$ for all $i \in [\az-1]$, we have $\sum_{i=1}^{\az-1} X_i \leq 1$. Then setting $X_\az = 1 - \sum_{i = 1}^{\az - 1} X_i \geq 0$ ensures that $(X_1, \dots, X_\az)$ is a probability vector.
\else
We scale by $1/2$ and add $X_\az$ so that all $\az$ variables sum to one.
\fi
Denoting this prior $P^*_{\text{hard}}$ and letting $p^{**}_\az(x) = 2p^*_\az(2x)$ (so $\newVar_i \sim p^*_\az \implies X_i \sim p^{**}_\az$) we get that
\begin{align}
    &\inf_{\comp \in \compset} \rawloss_\az(P^*_{\text{hard}}, \comp) \geq (\az - 1) \inf_{\comp \in \compset} \singleloss(p^{**}_\az, \comp) \label{eq::scale2_pstart}
    \\ &= (\az - 1) \frac{1}{2} {L^{\dagger}}(p^*_\az, \comp_\az^*) \label{eq::after_scale2_pstart}
    \geq  \frac{1}{2} \frac{\az - 1}{\az}  \sup_{P \in \cP^\triangle_\az} \rawloss_\az(P, \comp_\az^*) \,.
\end{align}
The last inequality holds because $p^*_\az$ is the maximin density (under expectation constraints). To make $P^*_{\text{hard}}$ symmetric, we permute the letter indices randomly without affecting the raw loss; thus we get \Cref{cor::worstcase_prior}.
%\iflong
To get \eqref{eq::after_scale2_pstart} from \eqref{eq::scale2_pstart}, we have
\begin{align}
    \inf_{\comp \in \compset} &\singleloss(2p^*_\az(2x), \comp)    = \frac{1}{24} \left(\int_0^1 (2p^*_\az(2x) x^{-1})^{1/3} dx\right)^3\\
    & = \frac{1}{24} \left(\int_0^1 (2p^*_\az(u) 2 u^{-1})^{1/3} \frac{1}{2}du\right)^3\\
    %& = \frac{1}{24} \left(\frac{2^{2/3}}{2}\int_0^1 (p^*_\az(u)  u^{-1})^{1/3} du\right)^3\\
    & = \frac{1}{2} {L^{\dagger}}(p^*_\az, \comp_*)\,.
\end{align}
%\fi

%As noted, it is sufficient to consider symmetric priors $P$ (symmetrizing a nonsymmetric prior can only make it harder to quantize); .

\iflong
This shows \Cref{prop::bound_worstcase_prior_exist}. In \Cref{fig:cramming_plot}, we validate the distribution $P^*_{\text{hard}}$ by showing the performance of each compander when quantizing  random distributions drawn from $P^*_{\text{hard}}$. For the minimax compander, the KL divergence loss on the worst-case prior looks to be within a constant of that for the other datasets. 

\begin{figure}
    \centering
    \includegraphics[scale = .4]{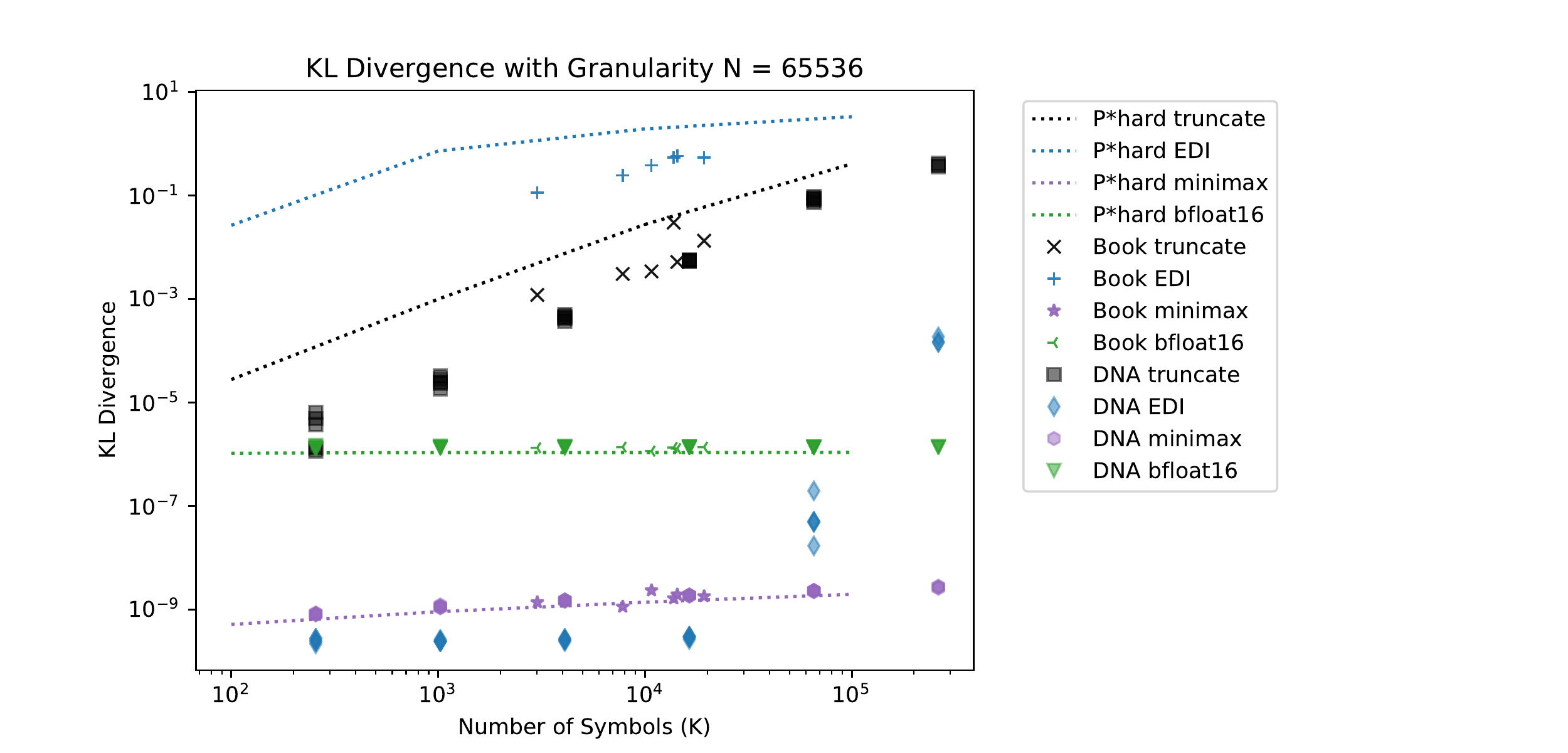}
    \caption{Each compander (or quantization method) is used on random distributions drawn from the prior $P^*_{\text{hard}}$. Comparison is given to when each compander is used on the books and DNA datasets.}
    \label{fig:cramming_plot}
\end{figure}
\fi

\section{Companding Other Metrics and Spaces}

\label{sec::other_losses}

While our primary focus has been KL divergence over the simplex, for context we compare our results to what the same compander analysis would give for other loss functions like squared Euclidean distance ($L_2^2$) and absolute distance ($L_1$ or $TV$ distance). For a vector $\bx$ and its representation $\bnormvar$ let
\begin{align}
    L_2^2(\bx, \bnormvar) &= \sum_i (x_i - \normvar_i)^2\\ 
    L_1(\bx, \bnormvar) &= \sum_i |x_i - \normvar_i|\,.  
\end{align}

For squared Euclidean distance, asymptotic loss was already given by \eqref{eq::bennett_compander} in \cite{bennett1948}, and scales as $N^{-2}$. It turns out that the maximin single-letter distribution over a bounded interval is the uniform distribution. Thus, the minimax compander for $L_2^2$ is simply the identity function, i.e. uniform quantization is the minimax for quantizing a hypercube in high-dimensional space under $L_2^2$ loss. (For unbounded spaces, $L_2^2$ loss does not scale with $N^{-2}$.)

If we add the expected value constraint to the $L_2^2$ compander optimization problem, we can derive the best square distance compander for the probability simplex. For alphabet size $\az$, we get that the minimax compander for $L_2^2$ is given by
\begin{align}
    f_{L_2^2, \az}(x) = \frac{\sqrt{1 + \az(\az-2)x} - 1}{\az-2}
\end{align}
and the total $L_2^2$ loss for probability vector $\bx$ and its quantization $\bnormvar$ has the relation
\begin{align}
    \lim_{N \to \infty} N^2 L_2^2 (\bx, \bnormvar) \leq \frac{1}{3}\,.
\end{align}
%\aviv{have to add some kind of limiting thingy on it}

For $L_1$, unlike KL divergence and $L_2^2$, the loss scales as $1/N$. Like $L_2^2$, the minimax single-letter compander for $L_1$ loss in the hypercube $[0,1]^\az$ is the identity function, i.e. uniform quantization. In general, the derivative of the optimal compander for single-letter density $p(x)$ has the form
\begin{align}
    \comp_{L_1, \az} ' (x) \, \propto \, \sqrt{p(x)}\,.
\end{align}

On the probability simplex for alphabet size $\az$, the worst case prior $p(x)$ has the form 
\begin{align}
    p(x) = (\alpha_\az x + \beta_\az )^{-2}
\end{align}
where $\alpha_\az, \beta_\az$ are constants scaling to allow $\int_{[0,1]} \, dp = 1$ (i.e. $p$ is a valid probability density) and $\int_{[0,1]} x \, dp = 1/\az$ (i.e. $\bbE_{X \sim p}[X] = 1/\az$ so $\az$ copies of it are expected to sum to $1$).

Thus, the minimax compander on the simplex for $L_1$ loss (and letting $\gamma_\az = \alpha_\az/\beta_\az$) satisfies
\begin{align}
    \comp_{L_1, \az} ' (x) \, &\propto \, (\alpha_\az x + \beta_\az)^{-1}
    \\ \implies \comp_{L_1, \az} (x) \, &\propto \, \log((\alpha_\az/\beta_\az) x + 1)
    \\ \implies \comp_{L_1, \az} (x) &= \frac{\log(\gamma_\az x + 1)}{\log(\gamma_\az + 1)}
\end{align}
since $\comp_{L_1, \az} (x)$ has to be scaled to go from $0$ to $1$.

The asymptotic $L_1$ loss for probability vector $\bx$ and its quantization $\bnormvar$ is bounded by
\begin{align}
    \lim_{N \to \infty} N L_1 (\bx, \bnormvar) %\leq \frac{(\log \az)^3}{4}
    =O(\log \az)
    \,.
\end{align}

%\annotate{Should we include a chart for a summary?}

\renewcommand{\arraystretch}{1}
\begin{figure*}[h!]
    \centering
    \begin{tabular}{|c|c|l|c|}
        \hline
         Loss & Space & Optimal Compander & Asymptotic Upper Bound   \\
         \hline
         KL & Simplex & $ \comp^*_\az(x) = \frac{\arcsinh(\sqrt{c_\az  (\az \log \az) \, x})}{\arcsinh(\sqrt{c_\az  \az \log \az})}$  & $N^{-2} \log^2 \az $ \\
         $L_2^2$ & Simplex & $ f_{L_2^2, \az}(x) = \frac{\sqrt{1 + \az(\az-2)x} - 1}{\az-2}$ & $N^{-2}$ \\
         $L_2^2$ & Hypercube & $ f_{L_2^2}(x) = x$ (uniform quantizer) & $N^{-2} \az $ \\
          $L_1~ (TV)$  & Simplex & $f_{L_1, \az}(x) = \frac{\log(\gamma_\az x + 1)}{\log(\gamma_\az + 1)}$ & 
          %$N^{-1} \log^3 \az$ 
          $N^{-1} \log \az$ 
          \\
           $L_1~ (TV) $  & Hypercube & $f_{L_1}(x) = x$ (uniform quantizer)  & $N^{-1} \az$ \\
          \hline
    \end{tabular}
    \caption{Summary of results for various losses and spaces. Asymptotic Upper Bound is an upper bound on how we expect the loss of the optimal compander to scale with $N$ and $\az$ (constant terms are neglected). }
    \label{tab::other_losses_spaces}
\end{figure*}

\section{Connection to Information Distillation}

\label{sec::info_distillation_main}

It turns out that the general problem of quantizing the simplex under the \textit{average} KL
divergence loss, as defined in~\eqref{eq::def_loss}, is equivalent to recently introduced problem
of \emph{information distillation}. %\cite{bhatt_distilling2021}. 
Information distillation has a number of applications, including in constructing polar codes~\cite{bhatt_distilling2021,tal2015}. 
In this section we establish this equivalence and also demonstrate how the compander-based
solutions to the KL-quantization can lead to rather simple and efficient information distillers. 

\vspace{-0.5pc}

\subsection{Information Distillation}

In the information distillation problem we have two random variables $A \in \cA$ and $B \in \cB$,
where $|\cA| = \az$ (and $\cB$ can be finite or infinite) under joint distribution $P_{A,B}$
%\nbyp{I think we should insert comma in $P_{A,B}$ everywhere} \textcolor{blue}{AVIV: Done.} 
with marginals $P_A, P_B$. 
%Without loss of generality we will let $\cA = [\az]$ (relabeling). 
The goal is, given some finite $M < |\cB|$, to find an \emph{information distiller} (which we will also refer to as a \emph{distiller}), which is a (deterministic) function $h : \cB \to [M]$, which minimizes the information loss
\begin{align} \label{eq::inf-distillation}
    I(A;B) - I(A;h(B))%\text{ where } \widetilde{B} = h(B)
\end{align}
associated with quantizing $B \to h(B)$. The interpretation here is that $B$ is a
(high-dimensional) noisy observation of some important random variable $A$ and we want to record
observation $B$, but only have $\log_2 M$ bits to do so. Optimal $h$ minimizes the
additive loss entailed by this quantization of $B$. 
%
%For simplicity we denote $\widetilde{B} = h(B)$; 
%
%As $M < |\cB|$, $h$ can be understood as a quantizer. The setup typically represents a scenario where random variable $A$ is sent through a noisy channel resulting in output $B$, which is then quantized to $\widetilde B = h(B)$.

To quantify the amount of loss incurred by this quantization, we use 
%
%This naturally leads to two related questions, similar to our analysis of the probability simplex quantization problem: (i) given $P_{A,B}$ and $M$, find a quantizer $h$ onto $M$ outputs that minimizes \eqref{eq::inf-distillation}, i.e. the best quantizer for a given channel; (ii) given $\az$ and $M$, find 
the \emph{degrading cost}~\cite{tal2015,bhatt_distilling2021}
\begin{align}
    \mathrm{DC}(\az, M) = \sup_{P_{A,B}} \inf_{h} I(A; B) - I(A;h(B))\,.
\end{align}
%which measures the worst possible information loss for a channel with alphabet size $\az$ and a (optimal) quantizer $h$ onto $M$ values~\cite{tal2015,bhatt_distilling2021}.
%
%Equivalent to maximizing $I(A;\widetilde B)$ is
%\begin{align}
%    \inf_{h} I(A; B) - I(A; \widetilde B)
%\end{align}
%which measures the minimum amount of information that must be lost in quantization~\cite{tal2015,bhatt_distilling2021}. If we do not know $P_{A,B}$ but know $|\cA| = \az$ and $M$, we consider the \emph{degrading cost}
%\begin{align}
%\mathrm{DC}(\az, M) = \sup_{P_{A,B}} \inf_{h} I(A; B) - I(A; \widetilde B)\,.
%\end{align}
Note that in supremizing over $P_{A,B}$ there is no restriction on $\cB$, only on $|\cA|$ and the size of the range of $h$. It has been shown in \cite{tal2015} that there is a $P_{A,B}$ such that 
\begin{align}
    \inf_{h} I(A; B) - I(A;h(B)) = \Omega(M^{-2/(\az - 1)})
\end{align}
giving a lower bound to $\mathrm{DC}(\az, M) $.
For an upper bound, \cite{Kartowsky2017} showed that if $2\az < M < |\cB|$, then 
\begin{align}
    \mathrm{DC}(\az, M) = O(M^{-2/(\az - 1)})\,.
\end{align}
Specifically, $\mathrm{DC}(\az, M) \leq \nu(\az) M^{-2/(\az - 1)}$ where $\nu(\az) \approx 16 \pi e \az^2$ for large $\az$.
While~\cite{bhatt_distilling2021} focused on multiplicative loss, their work also implied an
improved bound on the additive loss as well; namely, for all $\az \geq 2$ and $M^{1/(\az - 1)}
\geq 4$, we have
\begin{align}
    \mathrm{DC}(\az, M) \leq 1268 (\az - 1) M^{-2/(\az - 1)}\,.\label{eq::distilling_result}
\end{align}

\subsection{Info Distillation Upper Bounds Via Companders}

Using our KL divergence quantization bounds, we will show an upper bound to $\mathrm{DC}(\az, M)$ which improves on \eqref{eq::distilling_result} for $K$ which are not too small and for $M$ which are not exceptionally large. 
First, we establish the relation between the two problems: 

\iffalse
Note that any $b \in \cB$ induces a probability vector $\bx(b) \in \triangle_{\az-1}$ where
\begin{align}
    x_a(b) = P_{A|B}(a|b) = \bbP[A = a \, | \, B = b] \,.
\end{align}
We call an information distiller $h$ \emph{separable} if for any $b,b' \in \cB$,
\begin{align}
    \bx(b) = \bx(b') \implies h(b) = h(b')
\end{align}
i.e. if $b$ and $b'$ induce the same conditional probability vector for $A$, they are assigned the same quantization label.

\begin{proposition}\label{prop:infodist} 
For every $P_{A,B}$ define a random variable 
$\bX \in \triangle_{\az-1}$ by setting $X_a = P[A=a|B]$. Then, for every separable information distiller
$h: \cB \to [M]$ there is a %(possibly randomized) 
vector quantizer $\bnormvar(\bx): \triangle_{\az-1} \to
\triangle_{\az-1}$ with range of cardinality $M$ such that
\begin{equation}\label{eq::id_kl}
	I(A;B) - I(A;h(B)) = \bbE[D_{\kl}(\bX \| \bnormvar(\bX))]\,.
\end{equation}
Conversely, for any vector quantizer with centroid decoding $\bnormvar$ there exists a distiller $h$ %$h'$ 
such that \eqref{eq::id_kl} holds. 
%\begin{equation}\label{eq::id_kl2}
%	I(A;B) - I(A;h'(B)) = \bbE[D_{\kl}(\bX \| \bnormvar(\bX))]\,.
%\end{equation}
\end{proposition}
\fi

\begin{proposition} \label{prop:infodist}
For every $P_{A,B}$ define a random variable 
$\bX \in \triangle_{\az-1}$ by setting $X_a = P[A=a \, | \, B]$. Then, for every information distiller
$h: \cB \to [M]$ there is a vector quantizer $\bnormvar: \triangle_{\az-1} \to
\triangle_{\az-1}$ with range of cardinality $M$ such that
\begin{align}\label{eq::id_kl}
\hspace{-0.5pc} I(A;B) - I(A;h(B)) \geq \bbE[D_{\kl}(\bX \| \bnormvar(\bX))]\,.
\end{align}
Conversely, for any vector quantizer $\bnormvar$ there exists a distiller $h$ such that
\begin{align}\label{eq::id_kl2}
	I(A;B) - I(A;h(B)) \leq \bbE[D_{\kl}(\bX \| \bnormvar(\bX))]\,.
\end{align}
\end{proposition}
The inequalities in \Cref{prop:infodist} can be replaced by equalities if the distiller $h$ and the quantizer $\bnormvar$ avoid certain trivial inefficiencies. 
If they do so, there is a clean `equivalent' quantizer $\bnormvar$ for any distiller $h$, and vice versa, which preserves the expected loss.
%letting such $h, \bz$ be called \emph{efficient}, in a larger sense what \Cref{prop:infodist} shows is the existence of an quivalence relation between efficient distillers $h$ and efficient quantizers $\bnormvar$. 
This equivalence and \Cref{prop:infodist} are shown in \Cref{sec::info_distillation_detail}.

%\nbyp{Please make sure that Appendix contains words Proof of Prop.~\ref{prop:infodist} and the proof. Remove the other unnecessary stuff with $\inf_h$, $\sup_{P_{A,B}}$ etc.} -- \textcolor{blue}{AVIV: This doesn't hold because $\bnormvar$ might not use centroid decoding... which incidentally answers your later question about doing all the sups and infs and whatnot, we need that because the inf over $\bnormvar$ is what forces it to use centroid decoding.}
%
Thus, we can use KL quantizers to bound the degrading cost above (see \Cref{sec::info_distillation_detail} for details): 
%We show that the compander-based quantizers have competitive performance in terms of DC. 
%
%For a known channel $P_{A,B}$, this means we can derive a quantizer $h$ by considering companders on the equivalent prior $P$ over $\triangle_{\az-1}$; for maximum efficacy a different compander may be used for each of the $\az$ entries.
%
%\nbyp{I don't understand why you wrote all of of this. You are going to apply worst-case divergence-covering theorem, why bother with all the priors? I think this following paragraph can and should be removed.} -- \textcolor{blue}{AVIV: See above, \Cref{prop:infodist} is not actually true as stated hence why we need infs and sups... or we have to specify in \Cref{prop:infodist} that $\bnormvar$ does centroid decoding and show an inequality in cases where it doesn't etc etc etc}
%This means:
\begin{align}
     \mathrm{DC}(\az, M) &= \sup_{P_{A,B}} \inf_{h} I(A; B) - I(A;h(B))
     \\ &= \sup_{P} \inf_{\bnormvar} \bbE_{\bX \sim P} [D_{\kl}(\bX \| \bnormVar)] \label{eq::infg_supP_KL_symPrior}
     \\&\leq \inf_{\bnormvar} \sup_{P}  \bbE_{\bX \sim P} [D_{\kl}(\bX \| \bnormVar)]\,. \label{eq::infg_supP_KL}
\end{align}
We then use the approximate minimax compander results to give an upper bound to \eqref{eq::infg_supP_KL}. %which in turn gives an upper bound to $\mathrm{DC}(\az, M)$. 
This yields:

\begin{proposition}\label{prop::compander_worst_case_for_distilling}
For any $\az \geq 5$ and $M^{1/K} > \lceil 8 \log (2 \sqrt{\az \log \az} + 1) \rceil$ 
\begin{align} \label{eq::compander_worst_case_for_distilling}
    \mathrm{DC}(\az, M) \hspace{-0.2pc} \leq \hspace{-0.2pc} \bigg( \hspace{-0.2pc} 1 \hspace{-0.2pc} + \hspace{-0.2pc} 18 \frac{\log \log \az}{\log \az} \hspace{-0.1pc} \bigg) M^{-\frac{2}{\az}} \log ^2 \az\,.
\end{align}
\end{proposition}
\begin{proof}
Consider the right-hand side of~\eqref{eq::id_kl}. 
The compander-based quantizer from \Cref{thm::worstcase_power_minimax} gives a guaranteed bound on $D(\bX\|\bnormvar(\bX))$ (and $M = N^\az$ substituted), which also holds in expectation.
%Note that granularity $N$ results in $M = N^\az$ quantization points, leading to the result.
\end{proof}
%\nbyp{I have a question: Jennifer, if allegedly you have better results in your thesis for vector-quantization (ie divergence covering), should we also show a bound here via that? If you choose to do that, then you need to add a sentence "Although the vector-quantization bound is better, we note that compander based quantization is much more explicit".} \textcolor{blue}{Not sure we should bother.}

\begin{remark} Similarly, an upper bound on the divergence covering problem~\cite[Thm
2]{phdthesis} implies
	\begin{equation}\label{eq::distilling_result2}
		\mathrm{DC}(\az, M) \le 800 (\log K) M^{-2/(\az-1)} \,.
\end{equation}	
(This appears to be the best known upper bound on
	$\mathrm{DC}$.) The lower bound on the divergence covering, though, does not
	imply lower bounds on $\mathrm{DC}$, since divergence covering seeks one
	collection of $M$ points that are good for quantizing 
        %$\bX \sim P$ for 
        any $P$, whereas
	$\mathrm{DC}$ permits the collection to depend on $P$. For distortion
	measures that satisfy the triangle inequality, though, we have a provable relationship
	between the metric entropy and rate-distortion for the least-favorable prior,
	see~\cite[Section 27.7]{ITbook}.
\end{remark}

\section{Acknowledgements}

We would like to thank Anthony Philippakis for his guidance on the DNA $k$-mer experiments.

\bibliographystyle{resources/IEEEbib}
\bibliography{ref_compander}

\clearpage

\appendices

\crefalias{section}{appendix}

\section*{Appendix Organization}

\paragraph*{Appendix~\ref{sec::appendix_proofs_asymptotic}} We fill in the details of the proof of \Cref{thm::asymptotic-normalized-expdiv}. %In \Cref{sec::proof_asymptotic-interval-sizes,sec::proof_uniform-at-small-scales,sec::proof_divergence-of-close-distributions,sec::proof_loss-under-uniform-dist,sec::additional_lemmas} we cover the results needed to prove \Cref{prop::locloss-convergence}, while in \Cref{sec::proof_interval-loss-ub,sec::proof_dominating-companders} we cover the results needed to prove \Cref{prop::locloss-dominating}.

\paragraph*{Appendix~\ref{sec::approximate-compander-f-dagger}} We prove \Cref{prop::approximate-compander}.

\paragraph*{Appendix~\ref{sec::appendix_other_companders}} We develop and analyze other types of companders, specifically \emph{beta companders}, which are optimized to quantize vectors from Dirichlet priors (\Cref{sec::beta_companding}), and \emph{power companders}, which have the form $f(x) = x^s$ and have properties similar to the minimax compander (\Cref{sec::power_compander_analysis}). Supplemental experimental results are also provided.

\paragraph*{Appendix~\ref{sec::appendix_minimax}} We analyze the minimax compander and approximate minimax compander more deeply, showing that $c_\az  \in [1/4, 3/4]$ (\Cref{sec::minimax_const_analysis}) and $\lim_{\az \to \infty} c_\az = 1/2$ (\Cref{sec::lim_c_k}), and show that when $c_\az \approx 1/2$, the approximate minimax compander performs similarly to the minimax compander against all priors $p \in \cP$ (\Cref{sec::proof_L_appx_minimax_compander}). Supplemental experimental results are also provided.

\paragraph*{Appendix~\ref{sec::worst-case_analysis}} We prove \Cref{thm::worstcase_power_minimax}, showing bounds on the worst-case loss (adversarially selected $\bx$, rather than from a prior) for the power, minimax, and approximate minimax companders.

%\paragraph*{Appendix~\ref{sec::other_losses}} We discuss companders for other losses.

\paragraph*{Appendix~\ref{sec::info_distillation_detail}} We discuss the connection to information distillation in detail.

\section{Asymptotic Single-Letter Loss Proofs}

\label{sec::appendix_proofs_asymptotic}

In this appendix, we give all the proofs necessary for \Cref{thm::asymptotic-normalized-expdiv}, whose proof outline was discussed in \Cref{sec::asymt_single}. We begin with notation in \Cref{sec::locloss_notation}.
In \Cref{sec::locloss-convergence-preliminaries}, we give some preliminaries for showing \Cref{prop::locloss-convergence} (which shows that the local loss functions $g_N$ converge to the asymptotic local loss function $g$ a.s. when the input $X$ is distributed according to $p \in \cP$).
In \Cref{sec::locloss-convergence-pf}, we give the proof of \Cref{prop::locloss-convergence}. In \Cref{sec::locloss-dominating-proof}, we give the proof of \Cref{prop::locloss-dominating} (which shows the existence of an integrable $h$ dominating $g_N$ when the compander $f$ is from the `well-behaved' set $\cF^\dagger$).

In order to focus on the main ideas, some of the more minor details needed for \Cref{prop::locloss-convergence} and \Cref{prop::locloss-dominating} are omitted and left for later sections. We fill in the details on the lemmas and propositions used in the proof of \Cref{prop::locloss-convergence}, including proofs for all results from \Cref{sec::locloss-convergence-preliminaries} (specifically \Cref{lem::asymptotic-interval-sizes,lem::increase-interval-loss-positive} and \Cref{prop::uniform-at-small-scales,prop::divergence-of-close-distributions,prop::loss-under-uniform-dist}) in \Cref{sec::proof_asymptotic-interval-sizes,sec::proof_uniform-at-small-scales,sec::proof_divergence-of-close-distributions,sec::proof_loss-under-uniform-dist,sec::additional_lemmas}.

We then fill in the details of the lemmas for the proof of \Cref{prop::locloss-dominating}, specifically \Cref{lem::interval-loss-ub,lem::dominating-companders}.

\subsection{Notation}
\label{sec::locloss_notation}

%\begin{itemize}
%\item 

%\item 
Given probability distribution $p$ and interval $I$, $p|_I$ denotes $p$ restricted to $I$, i.e. $X \sim p|_I$ is the same as $X \sim p$ conditioned on $X \in I$.  We also define the probability mass of $I$ under $p$ as $\pi_{p,I} = \bbP_{X \sim p}[ X \in I]$. If $\pi_{p,I} = 0$, we let $p|_I$ be uniform on $I$ by default.

%\item 
Given two probability distributions $p, q$ (over the same domain), their \emph{Kolmogorov-Smirnov distance} (KS distance) is
\begin{align}
\hspace{-0.2pc} d_{KS}(p,q) \hspace{-0.2pc} = \hspace{-0.2pc} \norm{F_p - F_q}_{\infty} \hspace{-0.2pc} = \hspace{-0.2pc} \sup_x |F_p(x) - F_q(x)| \label{eq::infty-norm}
\end{align}
(recall that $F_p, F_q$ are the CDFs of $p,q$).

%\item 
We use standard order-of-growth notation (which are also used in \Cref{sec::main-theorems}). We review these definitions here for clarity, especially as we will use some of the rarer concepts (in particular, small-$\omega$). For a parameter $t$ and functions $a(t), b(t)$, we say:
\begin{align}
    a(t) = O(b(t)) &\iff \limsup_{t \to \infty} |a(t)/b(t)| < \infty
    \\ a(t) = \Omega(b(t)) &\iff \liminf_{t \to \infty} |a(t)/b(t)| > 0
    \\ a(t) = \Theta(b(t)) &\iff a(t) = O(b(t)), \, a(t) = \Omega(b(t))\,.
\end{align}
We use small-$o$ notation to denote the strict versions of these:
\begin{align}
    a(t) = o(b(t)) &\iff \lim_{t \to \infty} |a(t)/b(t)| = 0
    \\  a(t) = \omega(b(t)) &\iff \lim_{t \to \infty} |a(t)/b(t)| = \infty\,.
\end{align}
Sometimes we will want to indicate order-of-growth as $t \to 0$ instead of $t \to \infty$; this will be explicitly mentioned in that case.
%\end{itemize}

\subsection{Preliminaries for \Cref{prop::locloss-convergence}} \label{sec::locloss-convergence-preliminaries}

%We first state a lemma clarifying the properties of $\ell_{p,I}$.

We first generalize the idea of \emph{bins}. %\eqref{eq::bins}. 
The bin around $x \in [0,1]$ at granularity $N$ is the interval $I = I^{(n)}$ containing $x$ such that $f(I) = [(n-1)/N, n/N]$ for some $n \in [N]$. This notion relies on integers because $f(I) = [(n-1)/N, n/N]$ for integers $n,N$. We remove the dependence on integers while keeping the basic structure (an interval $I$ about $x$ whose image $f(I)$ is a given size):
\begin{definition} \label{def::pseudobin}
For any $x \in [0,1]$, $\theta \in [0,1]$, and $\varepsilon > 0$, we define the \emph{pseudo-bin} $I^{(x, \theta, \varepsilon)}$ as the interval satisfying:
\begin{align}
    \hspace{-1pc} I^{(x,\theta,\varepsilon)} &= [x - \theta r^{(x, \theta, \varepsilon)}, x + (1-\theta) r^{(x, \theta, \varepsilon)}] \text{ where }
    \\ \hspace{-0.85pc} r^{(x, \theta, \varepsilon)} &= \inf \big(r : f(x + (1-\theta) r) - f(x - \theta r) \geq \varepsilon \big) \label{eq::defn-of-r}\,.
\end{align}
\end{definition}
The interpretation of this is that $I^{(x,\theta,\varepsilon)}$ is the minimal interval $x$ such that $|f(I^{(x,\theta,\varepsilon)})| \geq \varepsilon$ and such that $x$ occurs at $\theta$ within $I^{(x,\theta,\varepsilon)}$, i.e. a $\theta$ fraction of $I^{(x,\theta,\varepsilon)}$ falls below $x$ and $1 - \theta$ falls above. Its width is $r^{(x,\theta,\varepsilon)}$. This implies that bins are a special type of pseudo-bins. Specifically, for any $x$ and $N$ (and any compander $f$),
\begin{align}
    I^{(n_N(x))} = I^{(x,\theta,1/N)} \text{ for some } \theta \in [0,1] \,.
\end{align}
We now consider the size of pseudo-bins as $\varepsilon \to 0$:
\begin{lemma} \label{lem::asymptotic-interval-sizes}
If $f$ is differentiable at $x$, then
\begin{align}
    \lim_{\varepsilon \to 0} \varepsilon^{-1} r^{(x,\theta,\varepsilon)} = f'(x)^{-1} 
\end{align}
(including going to $\infty$ when $f'(x) = 0$). The limit converges uniformly over $\theta \in [0,1]$.
\end{lemma}

The proof is given in \Cref{sec::proof_asymptotic-interval-sizes}.
Note that applying this to bins means $\lim_{N \to \infty} N r^{(n_N(x))} = f'(x)^{-1}$, and hence when $f'(x) > 0$ we have $r^{(n_N(x))} = N^{-1} f'(x)^{-1} + o(N^{-1})$.

%\subsection{Intervals are Approximately Uniform}

For any interval $I$, we want to measure how close $p$ is to uniform over $I$ using the distance measure $d_{KS}(p,q)$ from \eqref{eq::infty-norm}. We will show that when $F'_p(x) = p(x)$ is well-defined and positive at $x$, $p$ is approximately uniform on any sufficiently small interval $I$ around $x$. Formally:
\begin{proposition} \label{prop::uniform-at-small-scales}
If $p(x) = F'_p(x) > 0$ is well-defined, then for every $\varepsilon > 0$ there is a $\delta > 0$ such that for all intervals $I$ such that $x \in I$ and $r_I \leq \delta$,
\begin{align}
d_{KS}(p|_I, \unif_I) \leq \varepsilon \, .
\end{align}
\end{proposition}

We give the proof in \Cref{sec::proof_uniform-at-small-scales}. This allows us to use the following:
\begin{proposition} \label{prop::divergence-of-close-distributions}
Let $p$ be a probability measure and $I$ be an interval containing $x$ such that $r_I \leq x/4$ and $d_{KS}(p|_I,\unif_I) \leq \varepsilon$ where $\varepsilon \leq 1/2$. Then
\begin{align}
| \ell_{p,I} - \ell_{\unif_I} | \leq 2 \varepsilon r_I^2 x^{-1} + O(r_I^3 x^{-2})\,.
\end{align}
\end{proposition}

Recall that $\ell_{p,I}$ is the interval loss of $I$ under distribution $p$ when all points in $I$ are quantized to $\widetilde{y}_{p,I}$, the centroid of the interval.
We give the proof of \Cref{prop::divergence-of-close-distributions} in \Cref{sec::proof_divergence-of-close-distributions}.

\begin{proposition} \label{prop::loss-under-uniform-dist}
For any $x > 0$ and any sequence of intervals $I_1, I_2, \dots \subseteq [0,1]$ all containing $x$ such that $r_{I_i} \to 0$ as $i \to \infty$,
\begin{align}
    \ell_{\unif_{I_i}} = \frac{1}{24} r_{I_i}^2 x^{-1} + O(r_{I_i}^3 x^{-2}) \,.
\end{align}
\end{proposition}

The proof is in \Cref{sec::proof_loss-under-uniform-dist}.

%\newpage

Note that the above lemmas are all about asymptotic behavior as intervals shrink to $0$ in width; to deal with the (edge) case where they do not, we need the following lemma:
\begin{lemma} \label{lem::increase-interval-loss-positive}
For any $I$ such that $\bbP_{X \sim p} [X \in I] > 0$, there is some $a_I > 0$ such that
\begin{align}
    \ell_{p,J} \geq a_I \text{ for any } J \supseteq I \,.
\end{align}
\end{lemma}
We give the proof in %\jennifer{switched order in appendix} 
\Cref{sec::additional_lemmas}.

\subsection{Proof of \Cref{prop::locloss-convergence}}

\label{sec::locloss-convergence-pf}

We now combine the above results to prove \Cref{prop::locloss-convergence}, i.e. that $\lim_{N \to \infty} g_N(X) = g(X)$ almost surely when $X \sim p$. Because $p \in \cP$ (i.e. it is a continuous probability distribution) we will treat the bins as closed sets, i.e. $I^{(n)} = [f(\frac{n-1}{N})^{-1}, f(\frac{n}{N})^{-1}]$; this does not affect anything since the resulting overlap is only a finite set of points.

\begin{proof}

Since $p \in \cP$ then when $X \sim p$ the following hold with probability $1$:
\begin{enumerate}
    \item $0 < X < 1$;
    \item $f'(X)$ is well-defined;
    \item $p(X) = F'_p(X)$ is well defined;
    \item $p(X) > 0$.
\end{enumerate}
This is because if $p \in \cP$, and $|S|$ denotes the Lebesgue measure of set $S$, then
\begin{align}
    |S| = 0 \implies \bbP_{X \sim p} [X \in S] = 0\,.
\end{align}
This implies (1) since $\{0,1\}$ is measure-$0$.

Additionally, by Lebesgue's differentiation theorem for monotone functions, any monotonic function on $[0,1]$ is differentiable almost everywhere on $[0,1]$ (i.e. excluding at most a measure-$0$ set), and compander $f$ and CDF $F_p$ are monotonic. This implies 2) and 3). Finally, 4) follows because the set of $X$ such that $p(X) = 0$ has probability $0$ under $p$ by definition.

%because for any $\varepsilon > 0$,
%\begin{align}
%    \{X : p(X) = 0\} &\subseteq \{X : p(X) \leq \varepsilon\}
%    \\ \implies \bbP_{X \sim p}[p(X) = 0] &\leq \bbP_{X \sim p}[p(X) \leq \varepsilon]
%    \\ &\leq \varepsilon |\{X : p(X) \leq \varepsilon\}|
%    \\ &\leq \varepsilon \,.
%\end{align}
%where $|\{X : p(X) \leq \varepsilon\}|$ is the Lebesgue measure of the set $\{X : p(X) \leq \varepsilon\}$ (and is at most $1$ since it is a subset of $[0,1]$).
Therefore, we can fix $X \sim p$ and assume it satisfies the above properties. 

We now consider the bin size $r_{(n_N(X))}$ as $N \to \infty$; there are two cases, (a) $\lim_{N \to \infty} r_{(n_N(X))} = 0$ and (b) $\limsup_{N \to \infty} r_{(n_N(X))} > 0$. For case (b), since the length of the interval does not go to zero, $g_N(X) = N^2 \ell_{p,(n_N(X))} \to \infty$; additionally, $g(X) = \infty$ by default since case (b) requires that $f'(X) = 0$, and so $g_N(X) \to g(X)$ as we want. 

\emph{Case (a):} In this case (which holds for all $X$ if $f \in \cF^\dagger$), any $\delta > 0$ there is some sufficiently large $N^*$ (which can depend on $X$) such that
\begin{align}
    N \geq N^*\implies r_{(n_N(X))} \leq \delta \, .
\end{align}
By \Cref{prop::uniform-at-small-scales}, for any $\varepsilon > 0$ there is some $\delta > 0$ such that for all intervals $I$ where $X \in I$ and $r_I \leq \delta$, we have $d_{KS}(p|_I, \unif_I) \leq \varepsilon$. Putting this together implies that for any $\varepsilon > 0$, there is some sufficiently large $N^*_\varepsilon$ such that for all $N \geq N^*_\varepsilon$,
\begin{align}
    d_{KS}(p|_{(n_N(X))}, \unif_{(n_N(X))}) \leq \varepsilon \, .
\end{align}
i.e. $p$ is $\varepsilon$ close to uniform on $I^{(n_N(X))}$. Furthermore, we can always choose $\varepsilon \leq 1/2$ and $N^*_\varepsilon$ sufficiently large that $r_{(n_N(X))} \leq X/4$ (since $\lim_{N \to \infty} r_{(n_N(X))} = 0$). Under these conditions, for $N > N^*_\varepsilon$ we can apply \Cref{prop::divergence-of-close-distributions} and get
\begin{align}
    |&\ell_{p,(n_N(X))}  - \ell_{\unif_{(n_N(X))}}| 
    \nonumber \\ &\leq 2 \varepsilon r_{(n_N(X))}^2 X^{-1} + O(r_{(n_N(X))}^3 X^{-2}) \, .
\end{align}
We can then turn this around: as $N \to \infty$, we have $\varepsilon \to 0$ and hence $\varepsilon = o(1)$ (as $N \to \infty$), so
\begin{align} \label{eq::p-close-to-unif}
    |\ell_{p,(n_N(X))} - \ell_{\unif_{(n_N(X))}}| = o(r_{(n_N(X))}^2 X^{-1}) \, .
\end{align}

We then apply \Cref{prop::loss-under-uniform-dist} (note that since $r_{(n_N(X))} \leq X/4$ and $X \leq 2\bar{y}_{(n_N(X))}$, we know automatically that $r_{(n_N(X))} \leq \bar{y}_{(n_N(X))}/2$) to get that
\begin{align}
    \ell_{\unif_{(n_N(X))}} = \frac{1}{24} r_{(n_N(X))}^2 \bar{y}_{(n_N(X))}^{-1} + O(r_{(n_N(X))}^3 X^{-2})\,.
\end{align}
However, since $X$ is fixed and $r_{(n_N(X))} \to 0$ as $N \to 0$ (and $|X - \bar{y}_{(n_N(X))}| \leq r_{(n_N(X))}$ since they are both in the bin $I^{(n_N(X))}$), we know that $\bar{y}_{(n_N(X))} = X (1 + o(1))$ where $o(1)$ is in terms of $N$ (as $N \to \infty$). Hence (noting that $(1 + o(1))^{-1}$ is still $1 + o(1)$ and $O(r_{(n_N(X))}^3 X^{-2})$ is $o(1) r_{(n_N(X))}^2 X^{-1}$) we can re-write the above and combine with \eqref{eq::p-close-to-unif} to get 
\begin{align}
    \ell_{\unif_{(n_N(X))}} &= \frac{1}{24} (1 + o(1)) r_{(n_N(X))}^2 X^{-1}
    \\ \implies \ell_{p, (n_N(X))} &= \frac{1}{24} (1 + o(1)) r_{(n_N(X))}^2 X^{-1} \, .
\end{align}
We now split things into two cases: (i) $f'(X) > 0$; (ii) $f'(X) = 0$.

%, depending on whether $p$ has a non-zero probability of returning a value of $X$ on which $f$ is `flat': (i) $\bbP_{X \sim p}[f'(X) = 0] > 0$, and (ii) $\bbP_{X \sim p}[f'(X) = 0] = 0$.

\emph{Case i ($f'(X) > 0$):} For all $N$ there is a $\theta \in [0,1]$ such that $I^{(n_N(X))} = I^{(X,\theta,1/N)}$ (bins are pseudo-bins, see \Cref{def::pseudobin}).
Thus, by \Cref{lem::asymptotic-interval-sizes} (which shows uniform convergence over $\theta$),
\begin{align}
    \lim_{N \to \infty} N r_{(n_N(X))} = f'(X)^{-1}\,.
\end{align}
Thus, we may re-write as a little-$o$ and plug into $g_N(X)$:
\begin{align}
    r_{(n_N(X))} &= N^{-1} f'(X)^{-1} + o(N^{-1}) 
    \\ &= N^{-1} f'(X)^{-1} (1 + o(1)) 
    \\ \implies g_N(X) &= N^2 \ell_{p, (n_N(X))} \label{eq::gn-to-g-01}
    \\ &= N^2 \frac{1}{24} (1 + o(1)) r_{(n_N(X))}^2 X^{-1} \label{eq::gn-to-g-02}
    \\ &= N^2 \frac{1}{24} (1 + o(1)) N^{-2} f'(X)^{-2} X^{-1} \label{eq::gn-to-g-03}
    \\ &= \frac{1}{24} (1 + o(1)) f'(X)^{-2} X^{-1} \label{eq::gn-to-g-04}
\end{align}
implying $\lim_{N \to \infty} g_N(X) = g(X)$ as we wanted.

\emph{Case ii ($f'(X) = 0$):} As before, for any $N$ there is some $\theta \in [0,1]$ such that $I^{(n_N(X))} = I^{(X,\theta,1/N)}$. Thus, by \Cref{lem::asymptotic-interval-sizes} and as $f'(X) = 0$, we have
\begin{align}
    \lim_{N \to \infty} N r_{(n_N(X))} = \infty \, .
\end{align}
since the convergence in \Cref{lem::asymptotic-interval-sizes} is uniform over $\theta$. We can then re-write this as a little-$\omega$:
\begin{align}
    r_{(n_N(X))} = \omega(N^{-1}) \, .
\end{align}
This implies that
\begin{align}
    g_N(X) &= N^2 \ell_{p, (n_N(X))} 
    \\ &= N^2 \frac{1}{24} (1 + o(1)) r_{(n_N(X))}^2 X^{-1}
    \\ &= N^2 \frac{1}{24} (1 + o(1)) \omega(N^{-2}) X^{-1}
    \\ &= \omega(1)
\end{align}
where $\omega(1)$ means $\lim_{N \to \infty} g_N(X) = \infty$. But since $f'(X) = 0$, by convention we have $g(X) = \frac{1}{24} f'(X)^{-2} X^{-1} = \infty$ and so $\lim_{N \to \infty} g_N(X) = g(X)$ as we wanted.

\emph{Case (b):} $\limsup_{N \to \infty} r_{(n_N(X))} > 0$. Note that this can only happen if $f'(X) = 0$, so $g(X) = \infty$; hence our goal is to show that $\lim_{N \to \infty} g_N(X) = \infty$.

Related to the above, this only happens if $f$ is not strictly monotonic at $X$, i.e. if there is some $a < X$ or some $b > X$ such that $f(X) = f(a)$ or $f(X) = b$ (or both). If both, $[a,b] \subseteq I^{(n_N(X))}$ for all $N$. Since $p(X)$ is well-defined and positive, any nonzero-width interval containing $X$ has positive probability mass under $p$. Thus, by \Cref{lem::increase-interval-loss-positive}, there exists some $\alpha > 0$ such that all $J \supseteq [a,b]$ satisfies $\ell_{p,J} \geq \alpha$. But then $g_N(X) \geq N^2 \alpha$ and goes to $\infty$.

If only $a$ exists, we divide the granularities $N$ into two classes: first, $N$ such that $I^{(n_N(X))}$ has lower boundary exactly at $X$ (which can happen if $f(X)$ is rational), and second, $N$ such that $I^{(n_N(X))}$ has lower boundary below $X$. Call the first class $N^{(1)}(1), N^{(1)}(2), \dots$ and the second $N^{(2)}(1), N^{(2)}(2), \dots$. Then as no $b$ exists, $\lim_{i \to \infty} r^{(n_{N^{(1)}(i)}(X))} = 0$, i.e. the bins corresponding to the first class shrink to $0$ and the asymptotic argument applies to them, showing $g_{N^{(1)}(i)}(X) \to \infty$. For the second class, for any $i$, we have $I^{(n_{N^{(2)}(i)}(X))} \supseteq [a,X]$ and so we have an $\alpha > 0$ lower bound of the interval loss, and multiplying by $N^2$ takes it to $\infty$. Thus since both subsequences of $N$ take $g_N(X)$ to $\infty$, we are done. An analogous argument holds if $b$ exists but not $a$.

As this holds for any $X$ under conditions 1-4, which happens almost surely, we are done.
\end{proof}

%%%%%%%%%%%%
%%%%% begin DCT

\subsection{Proof of \Cref{prop::locloss-dominating}} \label{sec::locloss-dominating-proof}

To finish our Dominated Convergence Theorem (DCT) argument, we to prove \Cref{prop::locloss-dominating}, which gives an integrable function $h$ dominating all the local loss functions $g_N$. As with \Cref{prop::locloss-convergence}, we do this in stages. We first define:

\begin{definition}
For any interval $I$, let
\begin{align}
    \ell^*_I = \sup_q \ell_{q,I}
\end{align}
where $q$ is a probability distribution over $[0,1]$. If $I = I^{(n)}$ we can denote this as $\ell^*_{(n)}$.
\end{definition}
Since $\ell_{q,I}$ is only affected by $q|_I$ (i.e. what $q$ does outside of $I$ is irrelevant), we can restrict $q$ to be a probability distribution over $I$ without affecting the value of $\ell^*_I$. The question is thus: what is the maximum single-interval loss which can be produced on interval $I$?

Then, we can use the upper bound
\begin{align} \label{eq::locloss-upper-bound-1}
    g_N(x) = N^2 \ell_{p, (n_N(x))} \leq N^2 \ell^*_{(n_N(x))} \,.
\end{align}
This has the benefit of simplifying the term by removing $p$. We now bound $\ell^*_I$:

\begin{lemma} \label{lem::interval-loss-ub}
For any interval $I$, $\ell^*_I \leq \frac{1}{2} r_I^2 \bar{y}_I^{-1}$.
\end{lemma}

We give the proof in \Cref{sec::proof_interval-loss-ub}.
We can then add the above result to \eqref{eq::locloss-upper-bound-1} in order to obtain
\begin{align} \label{eq::locloss-upper-bound-2}
    g_N(x) \leq N^2 \ell^*_{(n_N(x))} \leq N^2 \frac{1}{2} r_{(n_N(x))}^2 \bar{y}_{(n_N(x))}^{-1}\,.
\end{align}

However, it is hard to use this as the boundaries of $I^{(n_N(x))}$ in relation to $x$ are inconvenient. Instead, use an interval which is `centered' at $x$ in some way, with the help of the following:
\begin{lemma} \label{lem::ell-star}
If $I \subseteq I'$, then $\ell^*_I \leq \ell^*_{I'}$.
\end{lemma}

\begin{proof}
This follows as any $q$ over $I$ is also a distribution over $I'$ (giving $0$ probability to $I'\backslash I$).
\end{proof}

Thus, if we can find some interval $J$ such that $I^{(n_N(x))} \subseteq J$ (but of the right size) and which had more convenient boundaries, we can use that instead. We define:
\begin{definition}
For compander $f$ at scale $N$ and $x \in [0,1]$, define the interval
\begin{align}
    J^{f,N,x} = f^{-1} \Big(\Big[f(x) - \frac{1}{N}, f(x) + \frac{1}{N}\Big] \cap [0,1] \Big)\,.
\end{align}
\end{definition}
As mentioned, we want this because it contains $I^{(n_N(x))}$:
\begin{lemma} \label{lem::bin-ub}
For any strictly monotonic $f$ and integer $N$,
\begin{align}
    I^{(n_N(x))} \subseteq J^{f,N,x}\,.
\end{align}
\end{lemma}

\begin{proof}
Since $f$ is strictly monotonic, it has a well-defined inverse $f^{-1}$.

By definition the bin $I^{(n_N(x))}$, when passed through the compander $f$, returns $[\frac{n-1}{N}, \frac{n}{N}]$, i.e.
\begin{align}
    f(I^{(n_N(x))}) = \Big[\frac{n-1}{N}, \frac{n}{N}\Big] \, .
\end{align}
Note that this interval has width $1/N$ and includes $f(x)$ and (by definition) it is in $[0,1]$. Hence,
\begin{align}
    &f(I^{(n_N(x))}) \subseteq \Big[f(x) - \frac{1}{N}, f(x) + \frac{1}{N}\Big] \cap [0,1] 
    \\ &\implies f(I^{(n_N(x))}) \subseteq f(J^{f,N,x})
    \\ &\implies I^{(n_N(x))} \subseteq J^{f,N,x}
\end{align}
and we are done.
\end{proof}

Now we can consider the importance of $f \in \cF^\dagger$: by dominating a monomial $c x^\alpha$, we can `upper bound' the interval $J^{f,N,x}$ by the equivalent interval with the compander $f_*(x) = c x^\alpha$ (i.e. $J^{f,N,x} \subseteq J^{f_*,N,x}$), which is then much nicer to work with.\footnote{While $f_*(x)$ may not map to all of $[0,1]$, it's a valid compander (but sub-optimal as it only uses some of the $N$ labels).} This also guarantees that $f$ is strictly monotonic.

\begin{lemma} \label{lem::dominating-companders}
If $f_1, f_2 \in \cF$ are strictly monotonic increasing companders such that $f_2 - f_1$ is also monotonically increasing (not necessarily strictly) and $f_1(0) = 0$, then for any $x \in [0,1]$ and $N$,
\begin{align}
    J^{f_2, N, x} \subseteq J^{f_1, N, x}\,.
\end{align}
\end{lemma}
The proof is given in \Cref{sec::proof_dominating-companders}.
Finally, we need a quick lemma concerning the guarantee that if $f \in \cF^\dagger$, the function $g(x) = \frac{1}{24} f'(x)^{-2} x^{-1}$ is integrable under any distribution $p$:

\begin{lemma}\label{lem::h_finite}
Let $f \in \cF^\dagger$, and let $g(x) = \frac{1}{24} f'(x)^{-2} x^{-1}$. Then for any probability distribution $p$ over $[0,1]$,
\begin{align}
    \int_{[0,1]} g \, dp < \infty \,.
\end{align}
\end{lemma}

\begin{proof}
If $f \in \cF^\dagger$, then there is some $c > 0$ and $\alpha \in (0,1/2]$ such that $f(x) - cx^\alpha$ is monotonically increasing. Thus (whenever it is well-defined, which is almost everywhere by Lebesgue's differentiation theorem for monotone functions) we have $f'(x) \geq c \alpha x^{\alpha-1}$ and since $\alpha \in (0,1/2]$, we have $1 - 2\alpha \geq 0$. Thus, for all $x \in [0,1]$,
\begin{align}
    0 \leq g(x) \leq \frac{1}{24} c^{-2} \alpha^{-2} x^{1 - 2\alpha} \leq \frac{1}{24} c^{-2} \alpha^{-2} %\, .
\end{align}
%This means that
%\begin{align}
%    \int_{[0,1]} g \, dp &= \bbE_{X \sim p} [g(X)] 
%    \\ &\leq \bbE_{X \sim p}\left[\frac{1}{24} c^{-2} %\alpha^{-2}\right] = \frac{1}{24} c^{-2} \alpha^{-2}
%\end{align}
which of course implies that $\int_{[0,1]} g \, p < \infty$.
\end{proof}

We can now prove \Cref{prop::locloss-dominating}, which will complete the proof of \Cref{thm::asymptotic-normalized-expdiv}.

%\begin{proposition} \label{prop::locloss-dominating}
%Let $f$ be a compander and $c > 0$ and $\alpha \in (0,1]$ such that $f(x) - c x^\alpha$ is monotonically increasing. Letting $g_N$ be the local loss functions as in \eqref{eq::loc-loss-fn} and
%\begin{align}
%     h(x) = (2^{2/\alpha} + \alpha^2 2^{1/\alpha - 2}) (c \alpha)^{-2} x^{1-2\alpha} +  c^{-1/\alpha} 2^{1/\alpha - 2}
%\end{align}
%then $g_N(x) \leq h(x)$ for all $x, N$. Additionally, if $\alpha \leq 1/2$ then $\int_{[0,1]} h \, dp < \infty$.
%\end{proposition}
%\jennifer{last statement not true. Aviv says: It is true, isn't it?}
\begin{proof}[Proof of \Cref{prop::locloss-dominating}]
As before, let $f_*(x) = c x^\alpha$; thus $f_*(0) = 0$ so we can apply \Cref{lem::dominating-companders}.
 %$g_N(x) \leq h(x)$ for all $x, N$ such that $g_N(x), h(x)$ are well defined and $p(x) > 0$. 
We begin, as outlined in \eqref{eq::locloss-upper-bound-2}, with:
\begin{align}
    g_N(x) &= N^2 \ell_{p,(n_N(x))}
    \\  &\leq N^2 \ell^*_{(n_N(x))} \label{eq::1st-ineq}
    \\  &\leq N^2 \ell^*_{J^{f,N,x}} \label{eq::2nd-ineq}
    \\  &\leq N^2 \ell^*_{J^{f_*,N,x}} \label{eq::3rd-ineq}
\end{align}
where \eqref{eq::1st-ineq} follows from the definition of $\ell^*_I$; \eqref{eq::2nd-ineq} follows from \Cref{lem::ell-star,lem::bin-ub}; and \eqref{eq::3rd-ineq} follows from \Cref{lem::dominating-companders}. However, since $f_*(x) = c x^\alpha$, we have a specific formula we can work with. We have $f'_*(x) = \alpha c x^{\alpha-1}$ and $f_*^{-1}(\newvar) = (\newvar/c)^{1/\alpha} = c^{-1/\alpha} \newvar^{1/\alpha}$. Note that this means we can re-write
\begin{align}
    h(x) = (2^{2/\alpha} + \alpha^2 2^{1/\alpha - 2}) f'_*(x)^{-2} x^{-1} +  c^{-1/\alpha} 2^{1/\alpha - 2}
\end{align}
which sheds some light on the structure of $h(x)$. Using \Cref{lem::h_finite} proves that $\int_{[0,1]} h \, dp$ is finite if $\comp \in \compset$, which occurs when $\alpha \leq 1/2$. 

Fix a value of $x$. Let $r_N(x)$ be the width of $J^{f_*, N, x}$. We consider two cases: (i) $cx ^{\alpha} < 1/N$; and (ii) $cx ^{\alpha} \geq 1/N$. 

\emph{Case (i):} This implies $f(J^{f_*, N, x}) \subseteq [0, 2/N]$ so 
\begin{align}
    x &< c^{-1/\alpha} N ^{-1/\alpha}
    \\ \implies r_N(x) &\leq c^{-1/\alpha} (N/2) ^{-1/\alpha}\,.
\end{align}
Then, as $J^{f_*, N, x}$ has lower boundary $0$ in this case, $\bar{y}_{(n_N(x))} = r_N(x)/2$. Thus, using \eqref{eq::locloss-upper-bound-2},
\begin{align}
    g_N(x) &\leq N^2 \frac{1}{2} r_N(x)^2 \bar{y}_{(n_N(x))}^{-1} 
    \\ &\leq  c^{-1/\alpha} 2^{-1/\alpha} N ^{-1/\alpha + 2}
  \,.
\end{align}
If $\alpha \leq 1/2$, then $N ^{-1/\alpha + 2}$ is maximized at $N = 1$, and thus 
\begin{align}
    g_N(x) \leq c^{-1/\alpha} 2^{-1/\alpha}\,.
\end{align}
If $\alpha > 1/2$, the value $N ^{-1/\alpha + 2}$ is maximized for the largest possible $N$ still satisfying Case (i). Since $cx^{\alpha} < 1/N$, this implies that $N < c^{-1} x^{-\alpha}$. Then,
\begin{align}
    g_N(x) &\leq  c^{-1/\alpha} ( c^{-1} x^{-\alpha})^{-1/\alpha + 2} 2^{-1/\alpha}
    \\ & =  c ^{-2 } x ^{1-2\alpha } 2^{-1/\alpha}
    \\& =  \alpha^2  (c\alpha x^{\alpha- 1})^{-2} x^{-1} 2^{-1/\alpha}
    \\ & =   \alpha^2  f_*'(x)^{-2} x^{-1} 2^{-1/\alpha}
    \,.
\end{align}
Thus, for Case (i) we have that for any $a \in (0, 1]$,
\begin{align}
    g_N(x) &\leq  \alpha^2  f_*'(x)^{-2} x^{-1} 2^{-1/\alpha} +  c^{-1/\alpha}2^{-1/\alpha}\,.
\end{align}

\emph{Case (ii):} When $cx ^{\alpha} \geq 1/N$, 
since $x \in I \implies \bar{y}_I \geq x/2$ (the midpoint of an interval cannot be less than half the largest element of the interval), we can upper-bound $g_N(x)$ (using \eqref{eq::3rd-ineq} and \Cref{lem::interval-loss-ub}) by 
\begin{align} \label{eq::4th-ineq}
    g_N(x) \leq N^2 \frac{1}{2} r_N(x)^2 \bar{y}^{-1}_{J^{f_*,N,x}} \leq N^2 r_N(x)^2 x^{-1}\,.
\end{align}
We then bound $r_N(x)$ using the Fundamental Theorem of Calculus: since $f$ is monotonically increasing, for any $a \leq b$,
\begin{align}
    \int_a^b f'(t) \, dt \leq f(b) - f(a)
\end{align}
(any discontinuities can only make $f$ increase faster). Additionally $r_N(x) = b_1 - a_1$ where $f(b_1) = \max(f(x) + 1/N,1)$ and $f(a_1) = f(x) - 1/N$ (since it's Case (ii) we know $f(x)-1/N \geq 0$ and since $f \in \cF^\dagger$ is strictly monotonic $a_1, b_1$ are unique). Thus, if we define $a_2, b_2$ such that
\begin{align}
    \int_{a_2}^x f'(t) \, dt = 1/N \text{ and } \int_x^{b_2} f'(t) \, dt = 1/N
\end{align}
(or $a_2 = 0$ or $b_2 = 1$ if they exceed the $[0,1]$ bounds) we have $r_N(x) \leq b_2 - a_2$. Then, because $f - f_*$ is monotonically increasing, we can define $a_3, b_3$ where
\begin{align}
    \int_{a_3}^x f'_*(t) \, dt = 1/N \text{ and } \int_x^{b_3} f'_*(t) \, dt = 1/N
\end{align}
and get that $r_N(x) \leq b_3 - a_3$ (also allowing $b_3 \geq 1$ if necessary). This yields:
\begin{align}
    &r_N(x) \leq c^{-1/\alpha} \int_{\max(0, cx^\alpha-1/N)}^{\min(1, cx^\alpha+1/N)} (f^{-1}_*)'(\newvar) \, d\newvar
    \\ &= c^{-1/\alpha}\int_{\max(0, cx^\alpha-1/N)}^{\min(1, cx^\alpha+1/N)} \alpha^{-1} \newvar^{1/\alpha-1} \, d\newvar
    \\ &\leq c^{-1/\alpha}\int_{\max(0,cx^\alpha-1/N)}^{\min(1, cx^\alpha+1/N)} \alpha^{-1} (cx^\alpha+1/N)^{1/\alpha-1} \, d\newvar
    \\ &\leq c^{-1/\alpha}\int_{cx^\alpha-1/N}^{cx^\alpha+1/N} \alpha^{-1} (cx^\alpha+1/N)^{1/\alpha-1} \, d\newvar
    \\ &= (2/N) c^{-1/\alpha} \alpha^{-1} (cx^\alpha+1/N)^{1/\alpha-1}
\end{align}
\begin{align}
    \implies r_N(x) &\leq (2/N) c^{-1/\alpha} \alpha^{-1} (cx^\alpha+1/N)^{1/\alpha-1} 
    \\ &\leq 2 N^{-1}  c^{-1/\alpha} \alpha^{-1} (2 cx^\alpha)^{1/\alpha-1} 
    \\ &= N^{-1} c^{-1/\alpha} \alpha^{-1}  2^{1/\alpha} (c x^\alpha)^{1/\alpha - 1}
    \\ &= 2^{1/\alpha} N^{-1}  \big(c^{-1} \alpha^{-1} x^{1-\alpha}\big)
    \\ &= 2^{1/\alpha} N^{-1} f'_*(x)^{-1}\,.
\end{align}
Thus, we can incorporate this into our bound \eqref{eq::4th-ineq}
\begin{align}
    g_N(x) &\leq N^2 r_N(x)^2 x^{-1} 
    \\& \leq 2^{2/\alpha} f'_*(x)^{-2} x^{-1} 
    \,.
\end{align}
So, $h(x)$, as the sum of the two cases,
%\begin{align}
%    h(x) = (2^{2/\alpha} + \alpha^2 2^{1/\alpha - 2}) f'_*(x)^{-2} x^{-1} +  c^{-1/\alpha} 2^{1/\alpha - 2}
%\end{align}
upper bounds $g_N(x)$ no matter what. 

%We now want to show that $\int_{[0,1]} g \, dp < \infty \implies \int_{[0,1]} h \, dp < \infty$. Note that $h(x)$ can be re-expressed in terms of the derivative of $f_*(x) = c x^{\alpha}$; then because $f - f_*$ is monotonically increasing (implying that $f' \geq f'_*$ at all $x$), we can make the following inference:
%\begin{align}
%    h(x) &= (2^{2/\alpha} + \alpha^2 2^{1/\alpha - 2}) f'_*(x)^{-2} x^{-1} +  c^{-1/\alpha} 2^{1/\alpha - 2} 
%    \\ &\geq (2^{2/\alpha} + \alpha^2 2^{1/\alpha - 2}) f'(x)^{-2} x^{-1} +  c^{-1/\alpha} 2^{1/\alpha - 2} 
%    \\ &= 24(2^{2/\alpha} + \alpha^2 2^{1/\alpha - 2}) g(x) +  c^{-1/\alpha} 2^{1/\alpha - 2}
%\end{align}
%This implies that for any single-letter distribution $p$,
%\begin{align}
%    \int_{[0,1]} h \, dp \leq 24(2^{2/\alpha} + \alpha^2 2^{1/\alpha - 2}) \int_{[0,1]} g \, dp + c^{-1/\alpha} 2^{1/\alpha - 2}
%\end{align}
%This shows that $\int_{[0,1]} g \, dp < \infty \implies \int_{[0,1]} h \, dp < \infty$, and in particular that if $\alpha \leq 1/2$ then (by \Cref{lem::h_finite}) $\int_{[0,1]} h \, dp < \infty$ for all $p$. Thus, we are done.
%Thus we have our $h(x)$ which is integrable and dominates all $g_N(x)$, proving the result; writing it explicitly yields
%\begin{align}
%    h(x) &= (2^{2/\alpha} + \alpha^2 2^{1/\alpha - 2}) f'_*(x)^{-2} x^{-1} +  c^{-1/\alpha} 2^{1/\alpha - 2}
%    \\& = (2^{2/\alpha} + \alpha^2 2^{1/\alpha - 2}) c^{-2} \alpha^{-2} x^{1 - 2\alpha} +  c^{-1/\alpha} 2^{1/\alpha - 2}
%\end{align}
%and we are done.

We can also note that if $\alpha \leq 1/2$, then $x^{1-2 \alpha} \leq 1$ and hence we can upper-bound $h$ by a constant. Thus $\int_{[0,1]} h \, dp = \bbE_{X \sim p} [h(X)] < \infty$ trivially, for any $p$, and we are done.
\end{proof}

This completes the proof of \eqref{eq::norm_loss} in \Cref{thm::asymptotic-normalized-expdiv}.

\subsection{Proof of \Cref{lem::asymptotic-interval-sizes}}

\label{sec::proof_asymptotic-interval-sizes}

\begin{proof}
Note that for fixed $\theta$ and $x$, $r^{(x,\theta,\varepsilon)}$ is nonnegative and monotonically decreases as $\varepsilon$ decreases. Thus $\lim_{\varepsilon \to 0} r^{(x,\theta,\varepsilon)} \geq 0$ is well defined.

We first assume that $\lim_{\varepsilon \to 0}  r^{(x,\theta,\varepsilon)} = 0$ for all $\theta \in [0,1]$. Let $s_\theta(r)$ be defined as
\begin{align}
    s_\theta(r) := \frac{ f(x + (1-\theta)r) - f(x - \theta r)}{r}\,.
\end{align}
We want to show that $\lim_{r \to 0} s_\theta(r) = f'(x)$ for all $\theta \in [0,1]$, and that this limit is uniform over $\theta \in [0,1]$. For $\theta \in \{0, 1\}$ we get respectively the right and left derivatives and since $f$ is differentiable at $x$ we are done for those cases. For $\theta \in (0,1)$ we write:
\begin{align}
    s_\theta(r) &= \frac{ f(x + (1-\theta)r) - f(x - \theta r)}{r} 
        \\ &= \frac{f(x + (1-\theta)r) - f(x)}{r}   \eqlinebreakshort+ \frac{f(x) - f(x - \theta r)}{r} 
        \\ &= (1-\theta)\frac{f(x + (1-\theta)r) - f(x)}{(1-\theta)r}   \eqlinebreakshort+ \theta\frac{f(x - \theta r)-f(x)}{-\theta r} \,.
        %\\ &= -\theta \frac{f(x - \theta r) - f(x)}{-\theta r}\eqlinebreakshort + (1-\theta) \frac{f(x + (1-\theta)r) - f(x)}{(1-\theta)r}
\end{align}
This implies
\begin{align}
   \lim_{r \to 0} s_\theta(r) &= \lim_{r \to 0} \bigg( (1-\theta)\frac{f(x + (1-\theta)r) - f(x)}{(1-\theta)r}   \eqlinebreakshort+ \theta\frac{f(x - \theta r)-f(x)}{-\theta r}   \bigg) \\
        &= (1-\theta) f'(x) + \theta f'(x) = f'(x)\,.
\end{align}
Furthermore we note that the convergence is uniform over $\theta \in [0,1]$. This is because for any $\alpha > 0$, there is a $\delta > 0$ such that for $|r| \leq \delta$,
\begin{align}
    \bigg| \frac{f(x+r) - f(x)}{r} - f'(x) \bigg| \leq \alpha\,.
\end{align}
But $|r| \leq \delta \implies |-\theta r| \leq \delta$ and $|(1-\theta)r| \leq \delta$. Thus,
\begin{align}
    \eqstartnonumshort|s_\theta(r) - f'(x)| \eqbreakshort
    &= \bigg|(1-\theta)\frac{f(x + (1-\theta)r) - f(x)}{(1-\theta)r}   \eqlinebreakshort+ \theta\frac{f(x - \theta r)-f(x)}{-\theta r} - f'(x) \bigg| \\
    %&= \bigg|\bigg(\theta \frac{f(x + \theta h) - f(x)}{\theta h} - \theta f'(x)\bigg) \eqlinebreakshort- \bigg((\theta - 1) \frac{f(x + (\theta - 1)h) - f(x)}{(\theta - 1)h} - (\theta - 1) f'(x)\bigg) \bigg| \\
    &\leq \bigg| (1-\theta)\frac{f(x + (1-\theta)r) - f(x)}{(1-\theta)r}  - (1-\theta) f'(x) \bigg| \eqlinebreakshort+ \bigg| \theta\frac{f(x - \theta r)-f(x)}{-\theta r} - \theta f'(x) \bigg| \\
    &\leq (1-\theta) \alpha + \theta \alpha \\
    &= \alpha\,.
\end{align}
Thus we have uniform convergence of $s_\theta(r)$ to $f'(x)$ over all $\theta \in [0,1]$ as $r \to 0$. Since $r^{(x,\theta,\varepsilon)} \to 0$ as $\varepsilon \to 0$,
\begin{align}
    f'(x) &= \lim_{\varepsilon \to 0} s_\theta(r^{(x,\theta,\varepsilon)})
    \\ &= \lim_{\varepsilon \to 0} \frac{f(x + (1-\theta)r^{(x,\theta,\varepsilon)}) - f(x - \theta r^{(x,\theta,\varepsilon)})}{r^{(x,\theta,\varepsilon)}}
    \\ &= \lim_{\varepsilon \to 0} \frac{\varepsilon}{ r^{(x,\theta,\varepsilon)}}
    \\ &\implies \lim_{\varepsilon \to 0} \varepsilon^{-1}  \, r^{(x,\theta,\varepsilon)} = f'(x)^{-1}
\end{align}
as we wanted. The third equality comes from the definition of $r^{(x,\theta,\varepsilon)}$ \eqref{eq::defn-of-r} and the fact that $f'(x)$ is well-defined.

Now we need to consider what happens if $\lim_{\varepsilon \to 0} r^{(x,\theta,\varepsilon)} \neq 0$ for some values of $\theta$; this can either be because the limit is positive or because the limit does not exist, but in either case it is clearly only possible if $f$ is not strictly monotonic at $x$ and hence only if $f'(x) = 0$. Additionally, it can only happen if $f$ is flat at $x$, i.e. there is either some $a < x$ or some $a > x$ such that $f(a) = f(x)$ (or both). In this case, for any $0 < \theta < 1$, $I^{(x,\theta,\varepsilon)}$ contains the interval between $a$ and $x$ and hence $r^{(x, \theta, \varepsilon)} \geq |x-a|$. For $\theta = 0$ and $\theta = 1$, either $r^{(x, \theta, \varepsilon)}$ is bounded away from $0$, or it approaches $0$; in the first case, $\varepsilon^{-1} r^{(x,\theta,\varepsilon)} \to \infty$ by default, while in the second the proof for the $\lim_{\varepsilon \to 0} r^{(x,\theta,\varepsilon)} = 0$ case holds.

Thus, for all values of $\theta \in [0,1]$, we know that $\lim_{\varepsilon \to 0} \varepsilon^{-1} r^{(x,\theta,\varepsilon)} = \infty$ as we need; and this is uniform over $\theta$ because for any $\theta \in (0,1)$ we have $\varepsilon^{-1} r^{(x,\theta,\varepsilon)} \geq \varepsilon^{-1} |x - a|$, meaning that for any large $\alpha > 0$, we can choose $\varepsilon^*$ small enough so that for all $\varepsilon < \varepsilon^*$ all of the following hold: (i) $\varepsilon^{-1} |x - a| > \alpha$; (ii) $\varepsilon^{-1} r^{(x,0,\varepsilon)} > \alpha$; and (iii) $\varepsilon^{-1} r^{(x,0,\varepsilon)} > \alpha$. Thus, we have uniform convergence and we are done.
\end{proof}

\subsection{Proof of \Cref{prop::uniform-at-small-scales}}

\label{sec::proof_uniform-at-small-scales}

\begin{proof}
We can assume that $\varepsilon \leq 1/2$ (if not, just use the value of $\delta$ corresponding to $\varepsilon = 1/2$). Let $\delta > 0$ be such that for all $x'$ such that $|x' - x| \leq \delta$,
\begin{align}
\Big| \frac{F_p(x') - F_p(x)}{x' - x} - p(x) \Big| \leq p(x) \varepsilon/8\,.
\end{align}
Since the derivative $p(x) = F'_p(x)$ is well-defined, this $\delta$ must exist. Then for $x' \in I$,
\begin{align}
&\big| (F_p(x') - F_p(x)) - (x' - x) p(x) \big| 
\eqlinebreakshort\leq  |x' - x| p(x) \varepsilon/8 \leq  r_I p(x) \varepsilon/8\,.
\end{align}
Now let $x''$ also be such that $|x'' - x| \leq \delta$. Then
\begin{align}
\big| (F_p(x'') &- F_p(x')) - (x'' - x') p(x) \big| 
\\ &= \big| ((F_p(x'') - F_p(x)) - (x'' - x)p(x) ) \eqlinebreakshort - ~((F_p(x') - F_p(x)) - (x' - x)p(x)) \big| 
\\ & \leq r_I p(x)\varepsilon/4 \label{eq::unif-upper-numerator}\,.
\end{align}
Let $x'$ be the lower boundary of $I$, so $x' + r_I$ is the upper boundary of $I$ (for which the above of course applies). Then we get
\begin{align}
\big| (F_p(x'  + r_I) - F_p(x')) - r_I p(x) \big| &\leq r_I p(x) \varepsilon/4
\\ \implies \Big| \frac{F_p(x'  + r_I) - F_p(x')}{r_I p(x)} -  1 \Big| &\leq \varepsilon/4\,. \label{eq::unif-upper-denominator}
\end{align}
Then we know that for any $x'' \in I$,
\begin{align}
F_{p|_I}(x'') = \frac{F_p(x'') - F_p(x')}{F_p(x'  + r_I) - F_p(x')}\,.
\end{align}
By \eqref{eq::unif-upper-numerator} we know that
\begin{align}
(x'' - x')p(x) - &r_I p(x) \varepsilon/4 \leq F_p(x'') - F_p(x') \nonumber \\ &\leq (x'' - x')p(x) + r_I p(x) \varepsilon/4
\\ \implies r_I p(x) ((x'' - x')&/r_I - \varepsilon/4) \leq F_p(x'') - F_p(x') \nonumber \\ &\leq r_I p(x) ((x'' - x')/r_I + \varepsilon/4) 
\end{align}
and by \eqref{eq::unif-upper-denominator} we know that 
\begin{align}
r_I p(x) - r_I p(x) \varepsilon/4 &\leq F_p(x + r_I) - F_p(x') \nonumber \\ &\leq r_I p(x) + r_I p(x) \varepsilon/4
\\ \implies r_I p(x) (1 - \varepsilon/4) &\leq F_p(x + r_I) - F_p(x') \nonumber \\ &\leq r_I p(x) (1 + \varepsilon/4)\,.
\end{align}
Noting that $(x'' - x')/r_I = F_{\unif_I}(x'') \in [0,1]$ is the CDF of the uniform distribution on $I$, we get that
\begin{align}
F_{p|_I}(x'') &\geq \frac{r_I p(x) ((x'' - x')/r_I - \varepsilon/4) }{r_I p(x) (1 + \varepsilon/4)} \\ &= \frac{(x'' - x')/r_I - \varepsilon/4}{1 + \varepsilon/4}
\\&\geq F_{\unif_I}(x'') - \varepsilon
\end{align}
and similarly that
\begin{align}
F_{p|_I}(x'') &\leq \frac{r_I p(x) ((x'' + x')/r_I - \varepsilon/4) }{r_I p(x) (1 - \varepsilon/4)} \\ &= \frac{(x'' - x')/r_I + \varepsilon/4}{1 - \varepsilon/4}
\\&\leq F_{\unif_I}(x'') + \varepsilon
\end{align}
and hence for such a $\delta > 0$ we have for all $I$ containing $x$ and such that $r_I \leq \delta$ we have
\begin{align}
|F_{p|_I}(x'') - F_{\unif_I}(x'')| \leq \varepsilon
\end{align}
for all $x'' \in I$. For $x'' \not \in I = [x', x' + r_I]$ we then observe that
\begin{align}
F_{p|_I}(x'') = F_{\unif_I}(x'') = \begin{cases} 0 &\text{if } x'' < x' \\ 1 &\text{if } x'' > x' + r_I \end{cases}
\end{align} 
thus finishing the proof.
\end{proof}

\subsection{Proof of \Cref{prop::divergence-of-close-distributions}}

\label{sec::proof_divergence-of-close-distributions}

\begin{proof}

Let $\xi = \widetilde{y}_{p,I} - \bar{y}_I$. %\jennifer{need to define what bar y is}
Then:
\begin{align}
|\xi| &= \bigg| \int_{I} \big(\bbP_{X \sim p|_I}[X \geq x] - \bbP_{X \sim \unif_I}[X \geq x]\big) \, dx \bigg|
	\\ & \leq \int_{I} \big| \bbP_{X \sim p|_I}[X \geq x] - \bbP_{X \sim \unif_I}[X \geq x] \big| \, dx
	\\ & \leq r_I \varepsilon \,.
\end{align}
For any distribution $q$ and any fixed value $\newvar$, define the shift operator $\shift{q}{\newvar}$ to denote the distribution of $X - \newvar$ where $X \sim q$ (i.e. just shift it by $\newvar$). Note that $\shift{p|_I}{\widetilde{y}_{p,I}}$ and $\shift{\unif_I} { \bar{y}_I}$ are both constructed to have expectation $0$, and in particular $\shift{\unif_I} { \bar{y}_I}$ is the uniform distribution over an interval of width $r_I$ centered at $0$. Additionally,
\begin{align}
d_{KS}\eqstartshort(\shift{p|_I}{ \widetilde{y}_{p,I}}, \shift{\unif_I}{ \bar{y}_I }) \eqbreakshort
&\leq  d_{KS}(\shift{p|_I} { \widetilde{y}_{p,I}}, \shift{\unif_I}{\widetilde{y}_{p,I}}) 
\\ &~~~~+ d_{KS}(\shift{\unif_I}{ \widetilde{y}_{p,I}}, \shift{\unif_I}{ \bar{y}_I}) \\ &\leq 2\varepsilon
\end{align}
since $d_{KS}(\cdot, \cdot)$ is a metric, $d_{KS}(q_1, q_2) = d_{KS}(\shift{q_1}{\newvar}, \shift{q_2}{\newvar})$ for any $q_1, q_2$ and $\newvar$, and 
\begin{align}
d_{KS}(\shift{\unif_I}{z_1}, \shift{\unif_I} {z_2}) \leq |z_2 - z_1|/r_I\,.
\end{align}

For convenience, let $q_1 = \shift{p|_I} { \widetilde{y}_{p,I}}$ and $q_2 = \shift{\unif_I}{ \bar{y}_I}$, and let $\newVar_1 \sim q_1$ and $\newVar_2 \sim q_2$. We know the following: $\bbE[\newVar_1] = \bbE[\newVar_2] = 0$; $d_{KS}(q_1, q_2) \leq 2\varepsilon$; and $q_1, q_2$ have support on $[-r_I, r_I]$.

Let $\eta_i = \bbE[\newVar_1^i] - \bbE[\newVar_2^i]$. Then we can compute the following:
\begin{align}
|\eta_i| &= \bigg| \int_0^{r_I^i} (\bbP[\newVar_1^i \geq x] -  \bbP[\newVar_2^i \geq x]) \, dx  
\eqlinebreakshort- \int_0^{r_I^i} (\bbP[\newVar_1^i \leq -x] -  \bbP[\newVar_2^i \leq -x]) \, dx \bigg|\,.
\end{align}
If $i$ is odd, then we do a $u$-substitution with $u = x^{1/i}$ and get
\begin{align}
|\eta_i| &= \bigg| \int_0^{r_I^i} (\bbP[\newVar_1 \geq x^{1/i}] -  \bbP[\newVar_2 \geq x^{1/i}]) \, dx  \eqlinebreakshort- \int_{-r_I^i}^{0} (\bbP[\newVar_1 \leq - x^{1/i}] -  \bbP[\newVar_2^i \leq -x^{1/i}]) \, dx \bigg|
\\ &= i \, \bigg| \int_0^{r_I} u^{i-1} (\bbP[\newVar_1 \geq u] -  \bbP[\newVar_2 \geq u]) \, du  
\eqlinebreakshort- \int_{-r_I}^0 u^{i-1} (\bbP[\newVar_1 \leq u] -  \bbP[\newVar_2 \leq u]) \, du \bigg|
\\ & \leq 2 \int_0^{r_I} i u^{i-1} 2\varepsilon du = 4\varepsilon r_I^i\,.
\end{align}
Similarly if $i$ is even we get
\begin{align}
|\eta_i| &= \bigg| \int_0^{r_I^i} (\bbP[\newVar_1 \geq x^{1/i}] -  \bbP[\newVar_2 \geq x^{1/i}]) \, dx 
\eqlinebreakshort + \int_{-r_I^i}^{0} (\bbP[\newVar_1 \leq - x^{1/i}] -  \bbP[\newVar_2^i \leq -x^{1/i}]) \, dx \bigg|
\\ &= i \, \bigg| \int_0^{r_I} u^{i-1} (\bbP[\newVar_1 \geq u] -  \bbP[\newVar_2 \geq u]) \, du  
\eqlinebreakshort + \int_{-r_I}^0 u^{i-1} (\bbP[\newVar_1 \leq u] -  \bbP[\newVar_2 \leq u]) \, du \bigg|
\\ & \leq 2 \int_0^{r_I} i u^{i-1} 2\varepsilon du = 4\varepsilon r_I^i
\end{align}
and we can conclude that $|\eta_i| \leq 4\varepsilon r_I^i$ in general.

Then we can take the respective Taylor expansions: let $X_1 \sim p|_I$ and $X_2 \sim \unif_I$ (and $\newVar_1 \sim q_1, \newVar_2 \sim q_2$ as above). We get
\begin{align}
\ell_{p,I} &= \bbE[X_1 \log(X_1/\widetilde{y}_{p,I})] 
\\ &= \widetilde{y}_{p,I}\bbE[(\newVar_1/\widetilde{y}_{p,I} + 1) \log(\newVar_1/\widetilde{y}_{p,I} + 1)]
\\ & = \widetilde{y}_{p,I}  \bbE\left[\newVar_1/\widetilde{y}_{p,I} + \frac{ (\newVar_1/\widetilde{y}_{p,I})^2}{2} - \frac{(\newVar_1/\widetilde{y}_{p,I})^3}{6(1+\eta)^2} \right]\label{eq::taylor_lagrange_ell_pI}
\end{align}
where $\eta$ is a number between $0$ and $\newVar_1/\widetilde{y}_{p,I}$ (we get this using Lagrange's formula for the error).

Since $\newVar_1 + \widetilde{y}_{p,I} \in I$, we know that 
\begin{align}
    \widetilde{y}_{p,I} - r_I \leq \newvar + \widetilde{y}_{p,I} \leq \widetilde{y}_{p,I} + r_I\,.
\end{align}
Since $r_I < x/ 4$ and $\widetilde{y}_{p,I} \geq x - r_I$ (as $x, \widetilde{y}_{p,I}$ share the width-$r_I$ interval $I$), we get that $\widetilde{y}_{p,I} > 3 r_I$, and therefore
\begin{align}
    \frac{2}{3} \widetilde{y}_{p,I} &< \newVar_1 + \widetilde{y}_{p,I} < \frac{4}{3} \widetilde{y}_{p,I}
    \\ \implies \frac{-1}{3} &< \newVar_1/\widetilde{y}_{p,I} < \frac{1}{3}\,.
\end{align}
This gives that $|\eta| < 1/3$. 
Using this and the fact that $\bbE[\newVar_1] = 0$ by construction, we can write \eqref{eq::taylor_lagrange_ell_pI} as
\begin{align}
    \ell_{p,I} &\leq \frac{1}{2} \bbE[\newVar_1^2]/\widetilde{y}_{p,I} + \frac{|\bbE[\newVar_1^3]|}{8/3 } (\widetilde{y}_{p,I})^{-2}
    \\ &\leq \frac{1}{2} \bbE[\newVar_1^2]/\widetilde{y}_{p,I} + \frac{r_I^3}{8/3 (x - r_I)^{2} } \,.
\end{align}
Since $r_I < x/4$, we know that $x - r_I > (3/4)x$, and hence
\begin{align}
    \ell_{p,I} &\leq \frac{1}{2} \bbE[\newVar_1^2]/\widetilde{y}_{p,I} + (2/3) r_I^3 x^{-2}\,.
\end{align}

 Hence we get
\begin{align}
    \ell_{p,I} &= \frac{1}{2} \bbE[\newVar_1^2]/\widetilde{y}_{p,I} + O(r_I^3 x^{-2})\,.
\end{align}
Because $x - r_I \leq \bar{y}_I$ as well (and $\newVar_2$ has support on $[-r_I, r_I]$) we can repeat the above arguments to conclude similarly that
\begin{align} \label{eq::uniform-taylor}
    \ell_{\unif_I} = \frac{1}{2}\bbE[\newVar_2^2]/\bar{y}_I + O(r_I^3 x^{-2})\,.
\end{align}
Hence their difference is
\begin{align}
    |\ell_{p,I} - \ell_{\unif_I}| \leq \\ \frac{1}{2} \big| \bbE[\newVar_1^2]/\widetilde{y}_{p,I} - \bbE[\newVar_2^2]/\bar{y}_I \big| + O(r_I^3 x^{-2}) \label{eq::what-we-want}\,.
\end{align}
Taking the main term, we split it into three parts:
\begin{align}
    \big| \bbE &[\newVar_1^2]/\widetilde{y}_{p,I} - \bbE[\newVar_2^2]/\bar{y}_I \big| 
    \\&\leq \big| \bbE[\newVar_1^2]/\widetilde{y}_{p,I} - \bbE[\newVar_1^2]/x \big| \label{eq::first-part-bd}
    \\ &~~+\big| \bbE[\newVar_2^2]/\bar{y}_I - \bbE[\newVar_2^2]/x \big| \label{eq::second-part-bd}
    \\ &~~+\big| \bbE[\newVar_1^2]/x - \bbE[\newVar_2^2]/x \big| \label{eq::third-part-bd}\,.
\end{align}
The first part \eqref{eq::first-part-bd} can be bounded by
\begin{align}
    \big| \bbE[\newVar_1^2]/\widetilde{y}_{p,I} - \bbE[\newVar_1^2]/x \big| &\leq |\bbE[\newVar_1^2]| \, |1/\widetilde{y}_{p,I} - 1/x|
    \\ &\leq r_I^2 \frac{|x - \widetilde{y}_{p,I}|}{\widetilde{y}_{p,I} x}
    \\ &\leq (4/3) r_I^3 x^{-2}
    \\ &= O(r_I^3 x^{-2})\,.
\end{align}
An analogous argument bounds \eqref{eq::second-part-bd}, giving
\begin{align}
    \big| \bbE[\newVar_2^2]/\bar{y}_I - \bbE[\newVar_2^2]/x \big| = O(r_I^3 x^{-2})\,.
\end{align}
Finally, \eqref{eq::third-part-bd} follows from
\begin{align}
    \big| \bbE[\newVar_1^2]/x - \bbE[\newVar_2^2]/x \big| = |\eta_2| x^{-1} \leq 4 \varepsilon r_I^2 x^{-1}\,.
\end{align}
Thus, plugging it all into \eqref{eq::what-we-want} we get
\begin{align}
    |\ell_{p,I} - \ell_{\unif_I}| \leq 2 \varepsilon r_I^2 x^{-1} + O(r_I^3 x^{-2}) \,.
\end{align}

\end{proof}

\subsection{Proof of \Cref{prop::loss-under-uniform-dist}}

\label{sec::proof_loss-under-uniform-dist}

\begin{proof}
Let $i^*$ be such that $r_{I_{i^*}} \leq x/4$ for all $i \geq i^*$ (since $\lim_{i \to \infty} r_{I_i} = 0$ this exists) and WLOG consider the sequence of $i \geq i^*$. The result then follows from the Taylor series of $\ell_{\unif_{I_i}}$, as shown by \eqref{eq::uniform-taylor} (see proof of \Cref{prop::divergence-of-close-distributions} in \Cref{sec::proof_divergence-of-close-distributions}). Keeping the definition from the proof of \Cref{prop::divergence-of-close-distributions}, we let $\newVar_2 \sim \shift{\unif_{I_i}}{\bar{y}_{I_i}}$, i.e. uniform over a width-$r_{I_i}$ interval centered at $0$. Thus we have $\bbE[\newVar_2^2] = \frac{1}{12} r_{I_i}^2$ and hence \eqref{eq::uniform-taylor} yields
\begin{align}
    \ell_{\unif_{I_i}} &= \frac{1}{2}\bbE[\newVar_2^2]/\bar{y}_{I_i} + O(r_I^3 x^{-2})
    \\ &= \frac{1}{24} r_{I_i}^2 \bar{y}_{I_i}^{-1} + O(r_{I_i}^3 x^{-2}) \label{eq::uniform-loss}\,.
\end{align}
But $\bar{y}_{I_i}$ and $x$ share the interval $I_i$ and hence as $r_{I_i} \to 0$,
\begin{align}
    \bar{y}_{I_i} &= x + O(r_{I_i}) 
    \\ &= x (1 + O(r_{I_i} x^{-1}))
    \\ \implies \bar{y}_{I_i}^{-1} &= x^{-1} (1 + O(r_{I_i} x^{-1}))
\end{align}
since when $r_{I_i}$ is very small, $O(r_{I_i} x^{-1})$ is very small so $(1 + O(r_{I_i} x^{-1})^{-1} = 1 + O(r_{I_i} x^{-1})$ (the inverse of a value close to $1$ is also close to $1$). Thus, we can replace $\bar{y}_{I_i}^{-1}$ in \eqref{eq::uniform-loss} to get
\begin{align}
    \ell_{\unif_I} = \frac{1}{24} r_{I_i}^2 x^{-1} + O(r_{I_i}^3 x^{-2})
\end{align}
as we wanted.
\end{proof}

\subsection{Single-Interval Loss Function Properties and Proof of \Cref{lem::increase-interval-loss-positive} } \label{sec::additional_lemmas}

We prove \Cref{lem::increase-interval-loss-positive} here; to do so, we show a few lemmas concerning the single-interval loss function $\ell_{p,I}$. First, we show an alternative formula for $\ell_{p,I}$ which sheds some light on it:

\begin{lemma} \label{lem::single-interval-loss-form-2}
For any $p, I$,
\begin{align}
    \ell_{p,I} = \bbE_{X \sim p|_I}[X \log X] - \widetilde{y}_{p,I} \log(\widetilde{y}_{p,I})\,.
\end{align}
\end{lemma}

\begin{proof}
We compute $\ell_{p,I}$ as follows:
\begin{align}
    \ell_{p,I} &= \bbE_{X \sim p}[X \log (X/\widetilde{y}_{p,I}) \, | \, X \in I]
    \\ &= \bbE_{X \sim p|_I}[X \log (X/\widetilde{y}_{p,I})]
    \\ &= \bbE_{X \sim p|_I}[X \log(X) - X \log(\widetilde{y}_{p,I})]
    \\ &= \bbE_{X \sim p|_I}[X \log X] - \bbE_{X \sim p|_I}[X] \log(\widetilde{y}_{p,I})
    \\ &= \bbE_{X \sim p|_I}[X \log X] - \widetilde{y}_{p,I} \log(\widetilde{y}_{p,I})
\end{align}
since $\widetilde{y}_{p,I} = \bbE_{X \sim p|_I}[X]$.
\end{proof}

We now want to show that it really does represent something resembling a loss function: first, that it is nonnegative, and second that it achieves equality if and only if $X \sim p$ on $I$ is known for sure (so the decoded value can be guaranteed to equal $X$).

\begin{lemma}
For any $p$ and $I \subseteq [0,1]$ (even $p$ is not continuous),
\begin{align}
    \ell_{p,I} \geq 0
\end{align}
with equality if and only if there is some $\newvar \in I$ s.t.
\begin{align}
    \bbP_{X \sim p}[X = \newvar \, | \, X \in I] = 1 \,.
\end{align}
\end{lemma}

\begin{proof}
Using \Cref{lem::single-interval-loss-form-2}, if we define the function $h(t) = t \log t$ then since $h$ is strictly convex, by Jensen's Inequality (where all expectations are over $X \sim p|_I$)
\begin{align}
    \ell_{p,I} = \bbE[h(X)] - h(\bbE[X]) \geq 0
\end{align}
with equality if and only if $X \sim p|_I$ is fixed with probability $1$.
\end{proof}

This yields the following corollary:
\begin{corollary} \label{cor::distortion-nonzero}
If $p \in \cP$ and $I$ has nonzero width,
\begin{align}
    \ell_{p,I} > 0 \,.
\end{align}
\end{corollary}
This follows because $p \in \cP$ is continuous and so cannot have all its mass on a particular value in any nonzero-width $I$. If $I$ has zero probability mass under $p$, then $\ell_{p,I}$ defaults to the interval loss under a uniform distribution.

Finally, we can prove \Cref{lem::increase-interval-loss-positive}. Recall that it shows that if $I$ has nonzero probability mass under $p$, one cannot get the interval loss to approach $0$ by choosing $J \supseteq I$, i.e. if $p \in \cP$ and $I$ is such that $\bbP_{X \sim p}[X \in I] > 0$, then there is some $\alpha > 0$ (which can depend on $I$) such that
\begin{align}
    \ell_{p,J} \geq \alpha \text{ for all } J \supseteq I \,.
\end{align}

\begin{proof}[Proof of \Cref{lem::increase-interval-loss-positive}]
We can re-write $\ell_{p,J}$ as
\begin{align}
    \ell_{p,J} &= \bbE_{X \sim p}[X \log(X/\widetilde{y}_{p,J}) \, | \, X \in J]
    \\ &= \int_J \frac{p(x)}{\int_J dp} x \log(x/\widetilde{y}_{p,J}) \, dx
\end{align}
where $\int_J \, dp$ is just the integral representation of $\bbP_{X \sim p}[X \in J]$.

Therefore, since $p \in \cP$, $\ell_{p,J}$ is continuous at $J$ with respect to the boundaries of $J$ (the inverse probability mass $(\int_J dp)^{-1}$ is continuous since $\int_J dp \geq \int_I dp > 0$).

Thus, we can consider $\ell_{p,J}$ as a continuous function over the boundaries of $J$ on the domain where $I \subseteq J \subseteq [0,1]$; this domain can be represented as a closed subset of $[0,1]^2$ and hence is compact. Thus, by the Weierstrass extreme value theorem, $\ell_{J,p}$ achieves its minimum $\alpha$ on this domain, and by \Cref{cor::distortion-nonzero} it must be positive.

Hence, we have shown that there is an $\alpha > 0$ such that for any $J \supseteq I$, $\ell_{p,J} > \alpha$.
\end{proof}

\subsection{Proof of \Cref{lem::interval-loss-ub}}

\label{sec::proof_interval-loss-ub}

\begin{proof}
We WLOG restrict ourselves to $q$ which are probability distributions over $I$. Let $\cP_I$ denote the set of probability distributions over $I$ (not necessarily continuous) and $\cP'_I$ denote the set of probability distributions over $I$ which place all the probability mass on the boundaries $\bar{y}_I - r_I/2$ and $\bar{y}_I + r_I/2$, i.e. for all $q' \in \cP'_I$ we have
\begin{align}
    \bbP_{X \sim q'} [X \in \{\bar{y}_I - r_I/2, \bar{y}_I + r_I/2\}] = 1 \, .
\end{align}
We then make the following claim:

\emph{Claim 1:} For all $q \in \cP_I$, exists $q' \in \cP'_I$ such that $\ell_{q,I} \leq \ell_{q',I}$.

This follows from the convexity of the function $x \log(x)$ and the definition of $\ell_{q,I}$, i.e.
\begin{align}
    \ell_{q,I} = \bbE_{X \sim q} [X \log (X/\widetilde{y}_{q,I})]
\end{align}
(since $q$ in this case is a distribution over $I$, we removed the condition $X \in I$ as it is redundant). In particular, if $q'$ is the (unique) distribution in $\cP'_I$ such that $\bbE_{X \sim q'}[X] = \widetilde{y}_{q,I}$ (i.e. we move all the probability mass to the boundary but keep the expected value the same), then $\ell_{q',I}$ can be computed by considering the average over the linear function which connects the end points of $X \log (X/\widetilde{y}_{q,I})$ over $I$. Because of convexity, this linear function is always greater than or equal to $X \log (X/\widetilde{y}_{q,I})$ on $I$, and therefore $\ell_{q,I} \leq \ell_{q',I}$. Thus, Claim 1 holds and we can restrict our attention to $\cP'_I$.

For simplicity we introduce a linear mapping $\newvar$ from $[-1/2,1/2]$ to $I$: for $\theta \in [-1/2,1/2]$, let $\newvar(\theta) = \bar{y}_I + \theta r_I$ (so $\newvar(-1/2) = \bar{y}_I - r_I/2$ is the lower boundary of $I$, $\newvar(1/2) = \bar{y}_I + r_I/2$ is the upper boundary, and $\newvar(0) = \bar{y}_I$ is the midpoint). We also specially denote $a = \newvar(-1/2)$ to be the lower boundary and $b = \newvar(1/2)$ to be the upper boundary. Then, since any $q \in \cP'_I$ can only assign probabilities to $a$ and $b$, we can parametrize all $q \in \cP'_I$: let $q(\theta)$ denote the distribution assigning probability $1/2 + \theta$ to the upper boundary $b$ and $1/2 - \theta$ to the lower boundary $a$. Then this gives the nice formula:
\begin{align}
    \widetilde{y}_{q(\theta), I} = \bar{y}_I + \theta r_I = \newvar(\theta)
\end{align}
i.e. $q(\theta)$ is the unique distribution in $\cP'_I$ with expectation $\newvar(\theta)$. This brings us to our next claim:

\emph{Claim 2:} $\ell_{q(\theta), I} \leq 2 \ell_{q(0), I}$ for any $\theta \in [-1/2, 1/2]$.
Ignoring the redundant condition $X \in I$, we use
\begin{align}
    \ell_{q,I} &= %\bbE_{X \sim q} [X \log(X/\widetilde{y}_{q,I})]
    %\\ &= \bbE_{X \sim q} [X \log(X) - X \log(\widetilde{y}_{q,I})]
    %\\ &= \bbE_{X \sim q} [X \log(X)] - \bbE_{X \sim q} [X \log(\widetilde{y}_{q,I})]
    %\\ &= 
    \bbE_{X \sim q} [X \log(X)] - \widetilde{y}_{q,I} \log(\widetilde{y}_{q,I}) \label{eq::alt-loss-formula}
\end{align}
to re-write $\ell_{q(\theta),I}$ as follows:
\begin{align}
    \ell_{q(\theta),I} &= (1/2 - \theta) a \log(a) 
    + (1/2 + \theta) b \log(b)
    \eqlinebreakshort- \newvar(\theta) \log(\newvar(\theta))\,.
\end{align}
This implies that
\begin{align}
    \ell_{q(\theta), I} &\leq \ell_{q(\theta), I} + \ell_{q(-\theta), I}
    \\ &= \big(a\log(a) + b \log(b) \big)
    \eqlinebreakshort- \big(\newvar(\theta) \log(\newvar(\theta)) + \newvar(-\theta) \log(\newvar(-\theta)) \big)
    \\ &\leq \big(a\log(a) + b \log(b) \big) - 2 \bar{y}_I \log(\bar{y}_I)
    \\ &= 2 \ell_{q(0),I}
\end{align}
where the inequality follows because $x \log(x)$ is convex and the mean of $\newvar(\theta)$ and $\newvar(-\theta)$ is $\newvar(0) = \bar{y}_I$, showing Claim 2.

\emph{Claim 3:} $2 \ell_{q(0),I} \leq \frac{1}{2} r_I^2 \bar{y}_I^{-1}$.

This comes from rewriting according to \eqref{eq::alt-loss-formula} and then applying the Taylor series expansion of $(1+t) \log(1+t)$. Define $t = r_I/(2 \bar{y}_I) \leq 1$ (otherwise $I \not \in [0,1]$), we get:
\begin{align}
    2 &\ell_{q(0),I} 
    \\&= \big(a\log(a) + b \log(b) \big) - 2 \bar{y}_I \log(\bar{y}_I)
    \\ &= (\bar{y}_I - r_I/2) \log(\bar{y}_I - r_I/2)
    \eqlinebreakshort + (\bar{y}_I + r_I/2) \log(\bar{y}_I + r_I/2)
    - 2 \bar{y}_I \log(\bar{y}_I)
    \\ &= (\bar{y}_I - r_I/2) (\log(\bar{y}_I - r_I/2) - \log(\bar{y}_I))
    \eqlinebreakshort + (\bar{y}_I + r_I/2) (\log(\bar{y}_I + r_I/2) - \log(\bar{y}_I))
    \\ &= \bar{y}_I \big((1 - t) \log(1 - t) + (1 + t) \log(1 + t) \big)
    %\\ &= 2 \bar{y}_I \sum_{i=1}^\infty \frac{t^{2i}}{(2i-1)(2i)}
    %\\ &\leq 2 \bar{y}_I t^2 \sum_{i=1}^\infty \frac{1}{(2i-1)(2i)}
    %\\ &\leq 2 \bar{y}_I t^2
    %\\ &= \frac{1}{2} r_I^2 \bar{y}_I^{-1}
    \,.
\end{align}
We can use the inequality that $(1 - t) \log(1 - t) + (1 + t) \log(1 + t) \leq 2 t^2$ for $|t| \leq 1$, to get
\begin{align}
    2 \ell_{q(0),I} \leq 2 \bar{y}_I t^2
    = \frac{1}{2} r_I^2 \bar{y}_I^{-1}\,.
\end{align}
This resolves Claim 3.

The lemma then follows from Claims 1, 2, and 3.
\end{proof}

\subsection{Proof of \Cref{lem::dominating-companders}}

\label{sec::proof_dominating-companders}

\begin{proof}
First, note that the above conditions imply that $f_2(x) \geq f_1(x)$ and that $f'_2(x) \geq f'_1(x)$ for all $x$ where both are defined (almost everywhere).

Let $J^{f_i, N, x} = [a_i, b_i]$ for $i = 1, 2$. We will prove that $a_1 \leq a_2$ and $b_1 \geq b_2$. Note that by definition if $f_1(x) - 1/N \leq 0$ then $a_1 = 0$ and $a_1 \leq a_2$ happens by default; thus this is also the case if $f_2(x) - 1/N \leq 0$ since $f_2 \geq f_1$ means this implies $f_1(x) - 1/N \leq 0$. Meanwhile, if $f_2(x) + 1/N \geq 1$ we have
\begin{align}
    1/N \geq 1 - f_2(x) \geq  f_2(1) - f_2(x) &\geq f_1(1) - f_1(x) 
\end{align}
meaning that $b_1 = 1$ (and $b_2 = 1$) so $b_1 \geq b_2$; and similarly $f_1(x) + 1/N \geq 1$ simply implies $b_1 = 1 \geq b_2$.

Thus we do not need to worry about the boundaries hitting $0$ or $1$ (i.e. we can ignore the `$\cap [0,1]$' in the definition), as the needed result easily holds whenever it happens.

Then $a_1$ and $a_2$ are the values for which
\begin{align}
    \int_{a_2}^x f'_2(t) \, dt = \int_{a_1}^x f'_1(t) \, dt = 1/N\,.
\end{align}
But since $0 \leq f'_1(t) \leq f'_2(t)$, we know that
\begin{align}
    \int_{a_2}^x f'_2(t) \, dt = 1/N = \int_{a_1}^x f'_1(t) \, dt \leq \int_{a_1}^x f'_2(t) \, dt
\end{align}
which implies that $a_2 \geq a_1$. An analogous proof on the opposite side proves $b_1 \geq b_2$ and hence
\begin{align}
    J^{f_2, N, x} = [a_2, b_2] \subseteq [a_1, b_1] = J^{f_1, N, x}
\end{align}
as we needed.
\end{proof}

\section{Proof of \Cref{prop::approximate-compander}
}

\label{sec::approximate-compander-f-dagger}

\begin{proof}
First, note that $\comp_\delta - \delta x^{1/2} = (1 - \delta) \comp$ is monotonically increasing so $\comp \in \compset^\dagger$. Furthermore, where the derivative $f'$ exists (which is almost everywhere since it is monotonic and bounded),
\begin{align}
    f'_\delta (x) = (1-\delta)f'(x) + (\delta/2)x^{-1/2}\,.
\end{align}
Thus, pointwise, $\lim_{\delta \to 0} f'_\delta (x) = f'(x)$ for all $x$. Since for all $\delta > 0$ we have $\comp \in \compset^\dagger$, \Cref{thm::asymptotic-normalized-expdiv} applies to $\comp_\delta$. So, we have
\begin{align}
    \lim_{\delta \to 0} \widetilde{L}(p,\comp_\delta) &= \lim_{\delta \to 0} L^\dagger(p,\comp_\delta)
    \\ &= \lim_{\delta \to 0} \frac{1}{24} \int_0^1 p(x) f'_\delta(x)^{-2} x^{-1} \, dx 
\end{align}
and $\lim_{\delta \to 0} p(x) f'_\delta(x)^{-2} x^{-1} = p(x) f'(x)^{-2} x^{-1}$, i.e. pointwise convergence of the integrand. We now consider two possibilities: (i) $\int_0^1 p(x) f'(x)^{-2} x^{-1} < \infty$; (ii) $\int_0^1 p(x) f'(x)^{-2} x^{-1} = \infty$.

In case (i), WLOG assume that $\delta \leq 1/2$; then $f'_\delta(x) > \frac{1}{2} f'(x)$, which implies $f'_\delta(x)^{-2} < 4 f'(x)^{-2}$. Thus, we have an integrable dominating function ($4 p(x) f'(x)^{-2} x^{-1}$) and we can apply the Dominated Convergence Theorem, which shows what we want.

In case (ii), we need to show $\lim_{\delta \to 0} \int_0^1 p(x) f'_\delta(x)^{-2} x^{-1} \, dx = \infty$. Let $\cX_\delta^+ = \{x \in [0,1] : f'(x) \geq \delta x^{-1/2}\}$ and $\cX_\delta^- = [0,1] \backslash \cX_\delta^+$, with $\bone_{\cdot}(\cdot)$ denoting their respective indicator functions. Then
\begin{align}
    f'_\delta(x) &= (1-\delta) f'(x) + (\delta/2) x^{-1/2}
    \\ &\leq f'(x) + \delta x^{-1/2}
    \\ &\leq 2 f'(x) \, \bone_{\cX_\delta^+}(x) + 2 \delta x^{-1/2} \, \bone_{\cX_\delta^-}(x)
    \\ \implies f'_\delta(x)^{-2} &\geq \frac{1}{4} f'(x)^{-2} \, \bone_{\cX_\delta^+}(x) + \frac{1}{4} \delta^{-2} x \, \bone_{\cX_\delta^-}(x)\,.
\end{align}
This then shows that (switching to $\int \cdot dp$ notation)
\begin{align}
    \int f'_\delta(x)^{-2} x^{-1} \, dp \geq &\frac{1}{4} \int \bone_{\cX_\delta^+}(x) f'(x)^{-2} x^{-1} \, dp
    \\ &+ \frac{1}{4} \int \bone_{\cX_\delta^-}(x) \delta^{-2}  \, dp\,. 
\end{align}
Note that $\cX_\delta^+$ expands as $\delta \to 0$. We then have two sub-cases (a) $\lim_{\delta \to 0} \bbP_{X \sim p} [X \in \cX_\delta^+] = 1$; (b) $\lim_{\delta \to 0} \bbP_{X \sim p} [X \in \cX_\delta^+] < 1$, which implies that there is some $\beta > 0$ such that $\bbP_{X \sim p} [X \in \cX_\delta^-] > \beta$ for all $\delta$. Then in sub-case (a), we have
\begin{align}
    &\lim_{\delta \to 0} \frac{1}{4} \int \bone_{\cX_\delta^+}(x) f'(x)^{-2} x^{-1} \, dp 
    \\ = & \frac{1}{4} \lim_{\delta \to 0} \bbE_{X \sim p} [\bone_{\cX_\delta^+}(X) f'(X)^{-2} X^{-1}] = \infty\,.
\end{align}
This is infinite because $\cX_0^+ := \lim_{\delta \to \infty} \cX_\delta^+$ is probability measure-$1$ set, and by the definition of Lebesgue integration, integration over $\cX_0^+$ is equivalent to the limit of integration over $\cX_\delta^+$, and since it is probability measure $1$ integrating over it with respect to $p$ is equivalent to integrating over $[0,1]$. Meanwhile in sub-case (b) we have
\begin{align}
    \frac{1}{4} \int \bone_{\cX_\delta^-}(x) \delta^{-2}  \, dp = \frac{\delta^{-2}}{4} \bbP_{X \sim p}[X \in \cX_\delta^-] \geq \frac{\delta^{-2}}{4} \beta
\end{align}
which goes to $\infty$ as $\delta \to 0$, and we are done.
\end{proof}

%%%%%%%%%%%%%%%%%%
%%% APPENDIX FOR MINIMAX STUFF

\section{Beta and Power Companders}

\label{sec::appendix_other_companders}

In this appendix, we analyze \emph{beta companders}, which are optimal companders for symmetric Dirichlet priors and are based on the normalized incomplete beta function (\Cref{sec::beta_companding}) and \emph{power companders}, which have the form $f(x) = x^s$ and which have properties similar to the minimax compander when $s = 1/\log \az$ (\Cref{sec::power_compander_analysis}). 

We also add supplemental experimental results. First, we compare the beta compander with truncation (identity compander) and the EDI (Exponential Density Interval) compander we developed in \cite{adler_ratedistortion_2021} in the case of the uniform prior on $\triangle_{\az-1}$ (which is equivalent to a Dirichlet prior with all parameters set to $1$), on book word frequencies, and on DNA $k$-mer frequencies. EDI was, in a sense, developed to minimize the expected KL divergence loss for the uniform prior (specifically to remove dependence on $\az$) as a means of proving a result in \cite{adler_ratedistortion_2021}; the beta compander was then directly developed for all Dirichlet priors.

Second, we compare the theoretical prediction for the power compander against various data sets; this demonstrates a close match to the theoretical performance for synthetic (uniform on $\triangle_{\az-1}$) data and DNA $k$-mer frequencies, while the power compander performs better on book word frequencies. Note that this is not a contradiction, as the theoretical prediction is for its performance on the worst possible prior -- it instead indicates that book word frequencies are somehow more suited to power companders than the uniform distribution or DNA $k$-mer frequencies. 

Finally, we compare how quickly the beta and power companders converge to their theoretical limits (with uniform prior); specifically how quickly $N^2 \widetilde{L}(p,f,N)$ converges to $\widetilde{L}(p,f)$. The results show that for large $\az$ ($\approx 10^5$), both are already very close by $N = 2^8 = 256$; while for smaller values of $\az$, power companders still converge very quickly while beta companders may take even until $N = 2^{16} = 65536$ or beyond to be close.

\subsection{Beta Companders for Symmetric Dirichlet Priors}

\label{sec::beta_companding}

\begin{definition} 
When $\bX$ is drawn from a Dirichlet distribution with parameters $\balpha = \alpha_1,\dots,\alpha_\az$, we use the notation $\bX \sim \dir(\balpha)$. When $\alpha_1 = \dots = \alpha_\az = \alpha$, then $\bX$ is drawn from a symmetric Dirichlet with parameter $\alpha$ and we use the notation  $\bX \sim \dir_\az(\alpha)$.
\end{definition}

As a corollary to \Cref{thm::optimal_compander_loss}, we get that the optimal compander for the symmetric Dirichlet distribution is the following:
\begin{corollary} 
When $\bx \sim \dir_\az(\alpha)$, let $p(x)$ be the associated single-letter density (same for all elements due to symmetry). The optimal compander for $p$ satisfies
\begin{align}
    \compder(x) &= B \Big(\frac{\alpha + 1}{3}, \frac{(\az-1)\alpha + 2}{3} \Big)^{-1} \eqlinebreakshort x^{(\alpha - 2)/3} (1-x)^{((\az-1)\alpha - 1)/3} \label{eq::beta_compander_deriv}
\end{align}
where $B(a,b)$ is the Beta function. Therefore, $\comp(x)$ is the normalized incomplete Beta function $I_x((\alpha + 1)/3, ((\az-1)\alpha + 2)/3)$.

Then
\begin{align}
    &\singleloss(p, \comp) \nonumber
    \\&= \frac{1}{2} B\Big(\frac{\alpha+1}{3}, \frac{(\az-1)\alpha + 2}{3} \Big)^3 B(\alpha, (\az-1)\alpha)^{-1}\label{eq::dir_Dpf}\,.
\end{align}

\end{corollary}

This result uses the following fact:
\begin{fact}  \label{fact::dirichlet-marginals}
For $\bX \sim \dir(\alpha_1, \dots, \alpha_\az)$, the marginal distribution on $X_k$ is $X_k \sim \Btdis(\alpha_k, \beta_k)$, where $\beta_k = \sum_{j \neq k} \alpha_j$%(of course, the $X_k$ are not independent)
. When the prior is symmetric with parameter $\alpha$, we get that all $X_k$ are distributed according to $\Btdis(\alpha, (\az-1)\alpha)$.
\end{fact}

%Hence when $P = \dir(\balpha)$, we have that  $p_k$ is the density of beta distribution  $\Btdis\left(\alpha_k, \sum_{j\neq k} \alpha_j\right)$.
%
%For $\bX \sim \dir(\balpha)$, we can match each coordinate $k$ with a different compander $f_k$.

\begin{remark}
Since \eqref{eq::dir_Dpf} scales with $\az^{-1}$, this means that $\widehat\cL_\az(\dir_\az(\alpha), \comp)$ is constant with respect to $\az$. This is consistent with what we get with the EDI compander (see \cite{adler_ratedistortion_2021}).
\end{remark}

We will call the compander $f$ derived from integrating \eqref{eq::beta_compander_deriv} the \emph{beta compander}. (This is because integrating \eqref{eq::beta_compander_deriv} gives an incomplete beta function.) The beta compander naturally performs better than the EDI method since this compander is optimized to do so. We can see the comparison in \Cref{fig::beta_and_EDI} that on random uniform distributions, 
the beta compander is better than the EDI method by a constant amount for all $\az$.

\begin{figure}
    \centering
    \includegraphics[scale = .4]{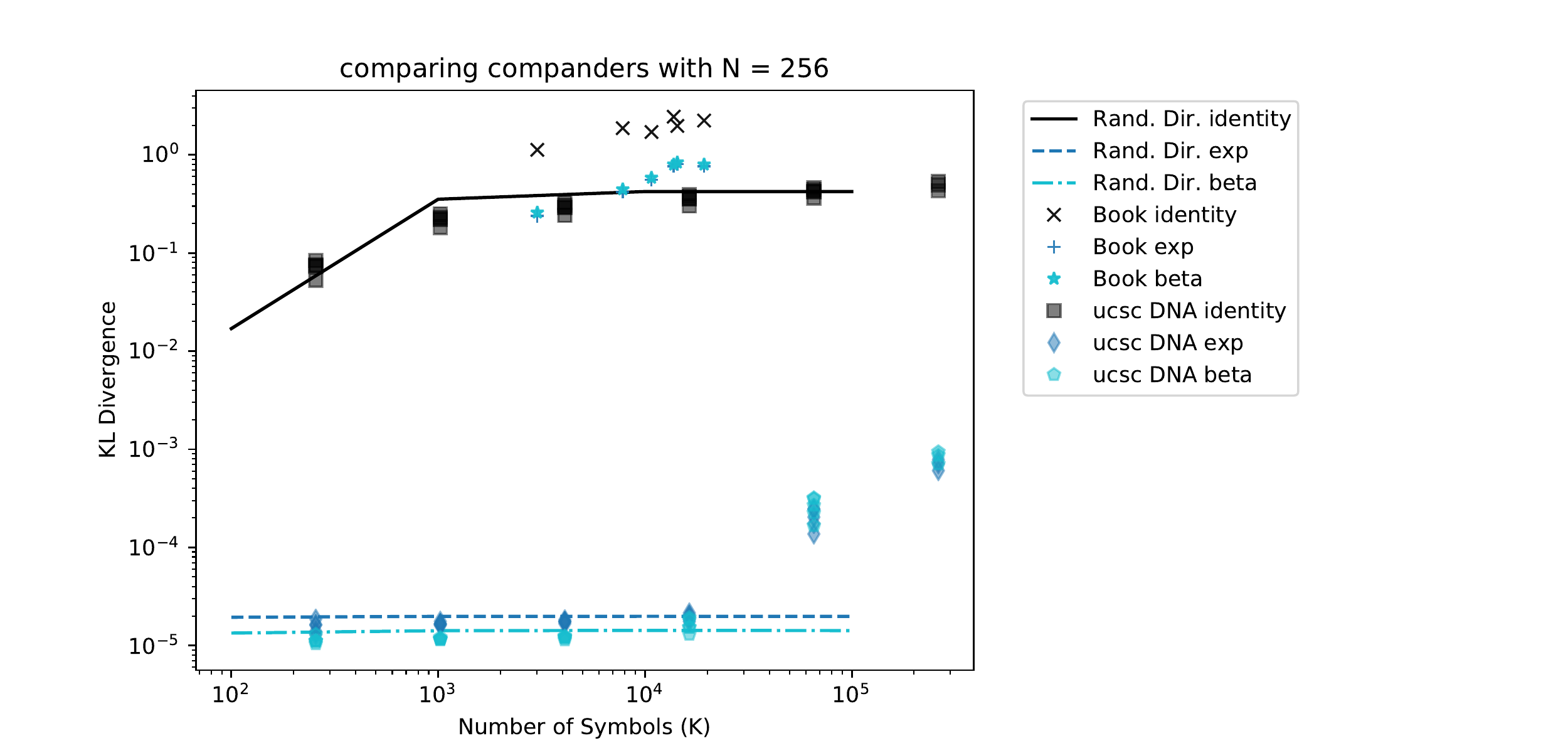}
    \includegraphics[scale = .4]{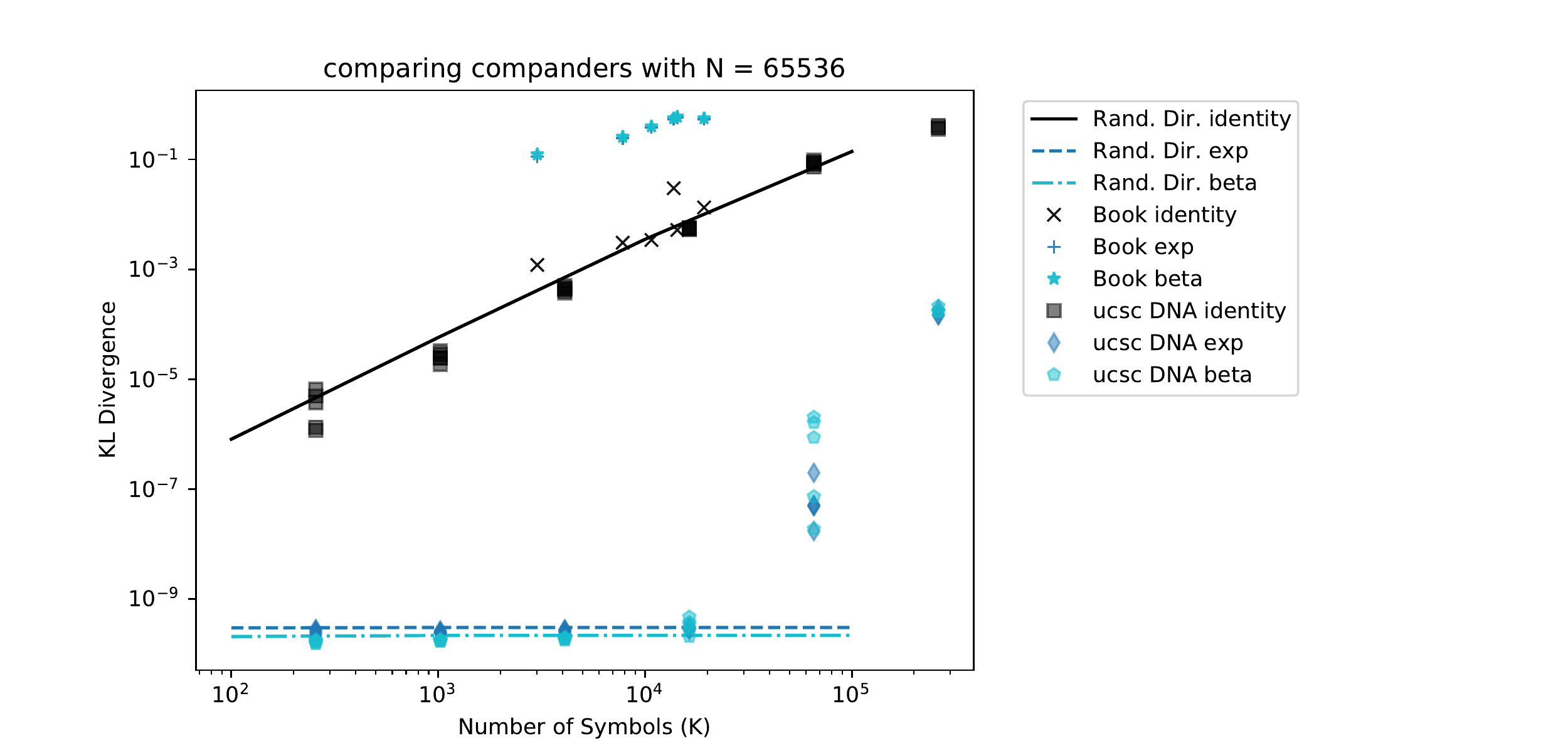}
    \caption{Comparing the beta compander and the EDI method. The random data is drawn with $\dir_\az(1)$ (i.e. uniform).}
    \label{fig::beta_and_EDI}
\end{figure}

The beta compander is not the easiest algorithm to implement however. It is necessary to compute an incomplete beta function in order to find the compander function $\comp$, which is not known to have a closed form expression. We reiterate \Cref{rmk::minimax-is-closed-form} that it is indeed interesting that the minimax compander, on the other hand, does have a closed form.

\subsection{Analysis of the Power Compander}

\label{sec::power_compander_analysis}

\begin{figure}
    \centering
    \includegraphics[scale = .4]{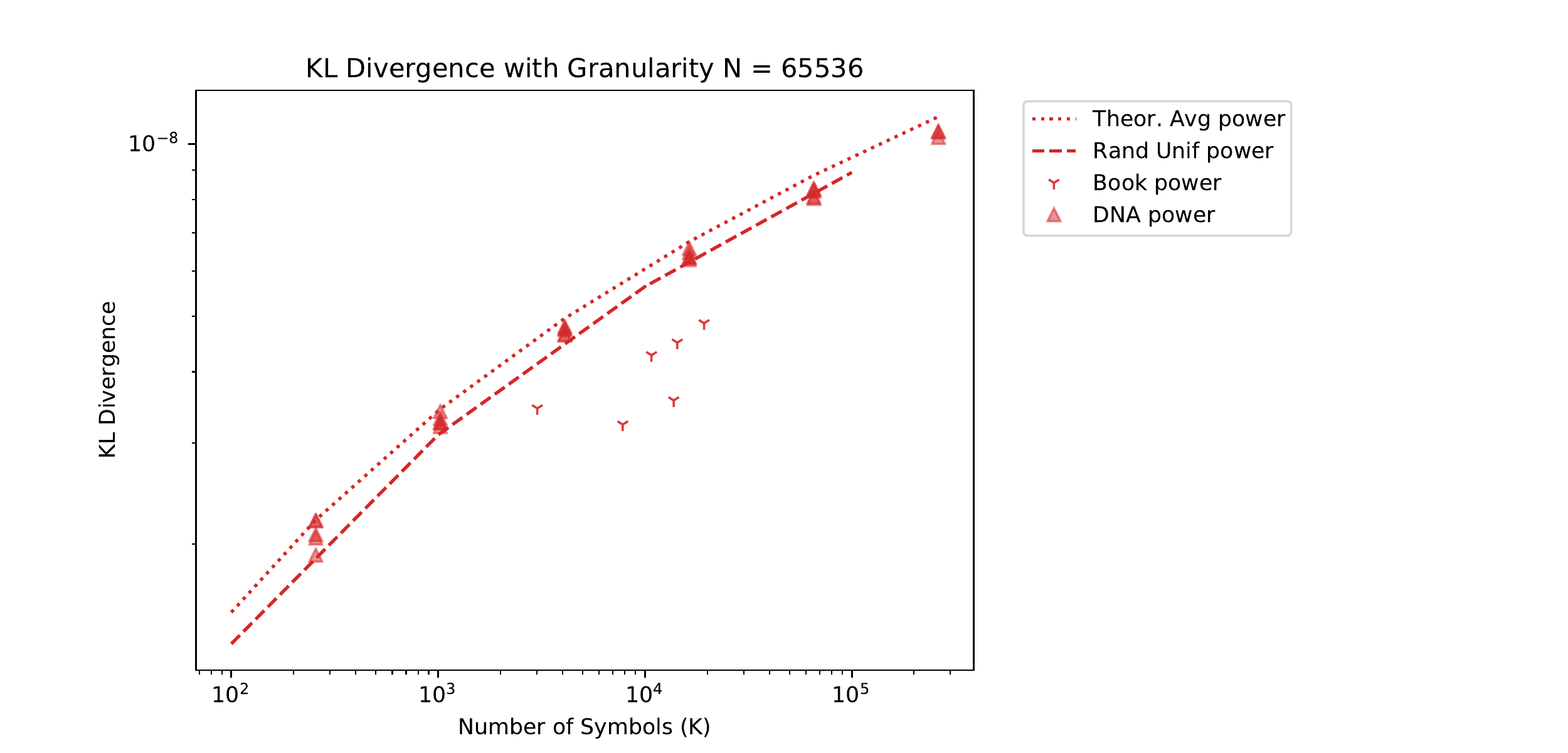}
    \caption{Comparing theoretical performance \eqref{eq::power_p_f} of the power compander to experimental results. }
    \label{fig::theoretical_power}
\end{figure}

Starting with \Cref{thm::asymptotic-normalized-expdiv}, we can use the asymptotic analysis to understand why the power compander works well for all distributions. The following proposition proves the first set of results in \Cref{thm::power_compander_results}.

\begin{proposition} \label{thm::power_r_loss}
Let single-letter density $p$ be the marginal probability of one letter on any symmetric probability distribution $P$ over $\az$ letters. 
For the power compander $\comp(x) = x^s$ where $s \leq \frac{1}{2}$,
    \begin{align}
        \singleloss(p, \comp) \leq \frac{1}{\az} \frac{1}{\dpfconst} s^{-2} \az^{2 s}
    \end{align}
and for any prior $P \in \cP^\triangle_\az$,
\begin{align}
    \rawloss_\az(P, x^s) \leq\frac{1}{\dpfconst} s^{-2} \az^{2 s}\,.
\end{align}
Optimizing over $s$ gives
\begin{align}
    \rawloss_\az(P, \comp) \leq \frac{e^2}{\dpfconst} \log^2 \az\label{eq::power_p_f}\,.
\end{align}
\end{proposition}

\begin{proof}
Since $\comp(x) = x^s$ we have that $\compder(x) = s x^{s-1}$. Using \Cref{thm::asymptotic-normalized-expdiv}, this gives
\begin{align}
    &\singleloss(p, \comp) \nonumber
    \\ &= \frac{1}{\dpfconst} s^{-2} \int_0^1 x^{1-2 s} p(x) dx = \frac{1}{\dpfconst} s^{-2}\bbE_{X\sim p}[X^{1-2 s}]\,.
\end{align}

The function $x^{1-2 s}$ is increasing and also a concave function. We want to find the maximin prior distribution $P \in \cP^\triangle_\az$ (with marginals $p$) with the constraint
\begin{align}
    \sum_{i} \bbE_{X_i \sim p}[X_i] &= 1
    %\\ \bbE_{ (X_1,...,X_\az)  \sim P}\left[\sum_{i} X_i\right] &= 1\,.
\end{align}
(another constraint is that values of $p$ are such that must sum to one, but we give a weaker constraint here).

We want to choose $P$ to maximize 
\begin{align}
    \sum_{i}  \bbE_{X_i \sim p}[X_i^{1-2 s}] & = \bbE_{ (X_1,...,X_\az)  \sim P}\left[\sum_{i} X_i^{1-2 s}\right]\,.
\end{align}
By concavity (even ignoring any constraint that $P$ is symmetric), the maximum solution is given when $X_1 = \dots = X_\az$. Therefore, the maximin $P$ is such that the marginal on one letter $p$ is
\begin{align}
    p(1/\az) = 1\,.
\end{align}
The probability mass function where $ 1/\az$ occurs with probability $1$ is a limit point of a sequence of continuous densities of the form
\begin{align}
    p(x) = \frac{1}{2\varepsilon} \text{ on } x \in \left[ \frac{1}{\az} - \varepsilon,  \frac{1}{\az} + \varepsilon\right]
\end{align}
as $\varepsilon \to 0$. We use this
since we are restricting to continuous probability distributions.

Evaluating with this gives
\begin{align}
    \singleloss(p, \comp) &= \frac{1}{\dpfconst} s^{-2}\bbE_{X\sim p}[X^{1-2 s}] \\
    &\leq \frac{1}{\dpfconst} s^{-2} \left(\frac{1}{\az} \right) ^{1-2 s} \\
    &= \frac{1}{\az} \frac{1}{\dpfconst} s^{-2} \az^{2 s}
\end{align}
which shows \eqref{eq::power_p_f}. Multiplying by $\az$ gives $\widetilde \cL_\az(P,\comp)$ for symmetric $P$. 

Note that for any non-symmetric $P$, we can always symmetrize $P$ to a symmetric prior $P_{sym}$ by averaging over all random permutations of the indices. Because the loss $\widetilde \cL_\az(P,\comp)$ is concave in $P$, the symmetrized prior $P_{sym}$ will give an higher value, that is $\widetilde \cL_\az(P,\comp) \leq \widetilde \cL_\az(P_{sym},\comp)$. Hence
$\widetilde \cL_\az(P,\comp) \leq \frac{1}{24}s^{-2}\az^{2s}$ holds for all priors.

Finding the $s$ which minimizes $\frac{1}{\dpfconst} s^{-2} \az^{2 s}$ is equivalent to finding $s$ which minimizes $s \log \az - \log s$.
\begin{align}
    0 &= \frac{d}{ds} s \log \az - \log s = \log \az - \frac{1}{s}\\
    &\implies s = \frac{1}{\log \az}\,.
\end{align}
We can plug this back into our equation, using the fact that $e^{\log \az} = \az$ implies that $\az^{\frac{1}{\log \az}} = e$.

Thus, using $\comp(x) = x^{\frac{1}{\log \az}}$ gives that 
\begin{align}
    \rawloss_\az(P, \comp) \leq \frac{e^2}{\dpfconst} \log^2 \az \text{ for any } P \in \cP^\triangle_\az\,.
\end{align}
To generate a prior $P \in \cP^\triangle_\az$ that matches this upper bound, we note that this means we want its marginal $p$ to maximize $\frac{1}{24} (\log^2 \az) \bbE_{X \sim p}[X^{1-2/\log \az}]$, and from before we know that fixing $X = 1/\az$ does this (since $\bbE_{X \sim p}[X] = 1/\az$ as $p$ is the marginal of $P$). While $p$ has to represent a probability density function, and therefore cannot be a point mass, we can restrict its support to an arbitrarily small neighborhood around $1/\az$ (and it is obvious that there are priors $P \in \cP^\triangle_\az$ with such a marginal), thus getting a match and showing that
\begin{align}
    \underset{P \in \cP^\triangle_\az}{\sup}  \rawloss_\az(P, \comp) = \frac{e^2}{\dpfconst} \log^2 \az \,.
\end{align}
\end{proof}

The power compander turns out to give guarantees bounds on the value on $\rawloss_\az(P, \comp)$ when $\comp$ is chosen so that $s = {1}/{\log \az}$. We show the comparison between this theoretical result on raw loss with the experimental results in \Cref{fig::theoretical_power}.

\subsection{Converging to Theoretical} \label{sec::beta_power_convergence}

For both the power compander and the beta  compander, we show in \Cref{fig::theoretical_compare} how quickly the experimental results converge to the theoretical results. Experimental results have a fixed granularity $N$ whereas the theoretical results assume that $N \to \infty$. The plots show that by $N = 2^{16}$ (each value gets $16$ bits), the experimental results for the power compander are very close to the theoretical results, and even for $N = 2^8$ they are not so far.  For the beta compander, the experimental results are close to the theoretical when $\az$ is large. When $\az  = 100$, the results for $N = 2^{16}$ is not that close to the theoretical result, which demonstrates the effect of using unnormalized (or raw) values. The difference between normalizing and not normalizing gets smaller as $\az$ increases.  

\begin{figure*}
    \centering
    \includegraphics[scale = .5]{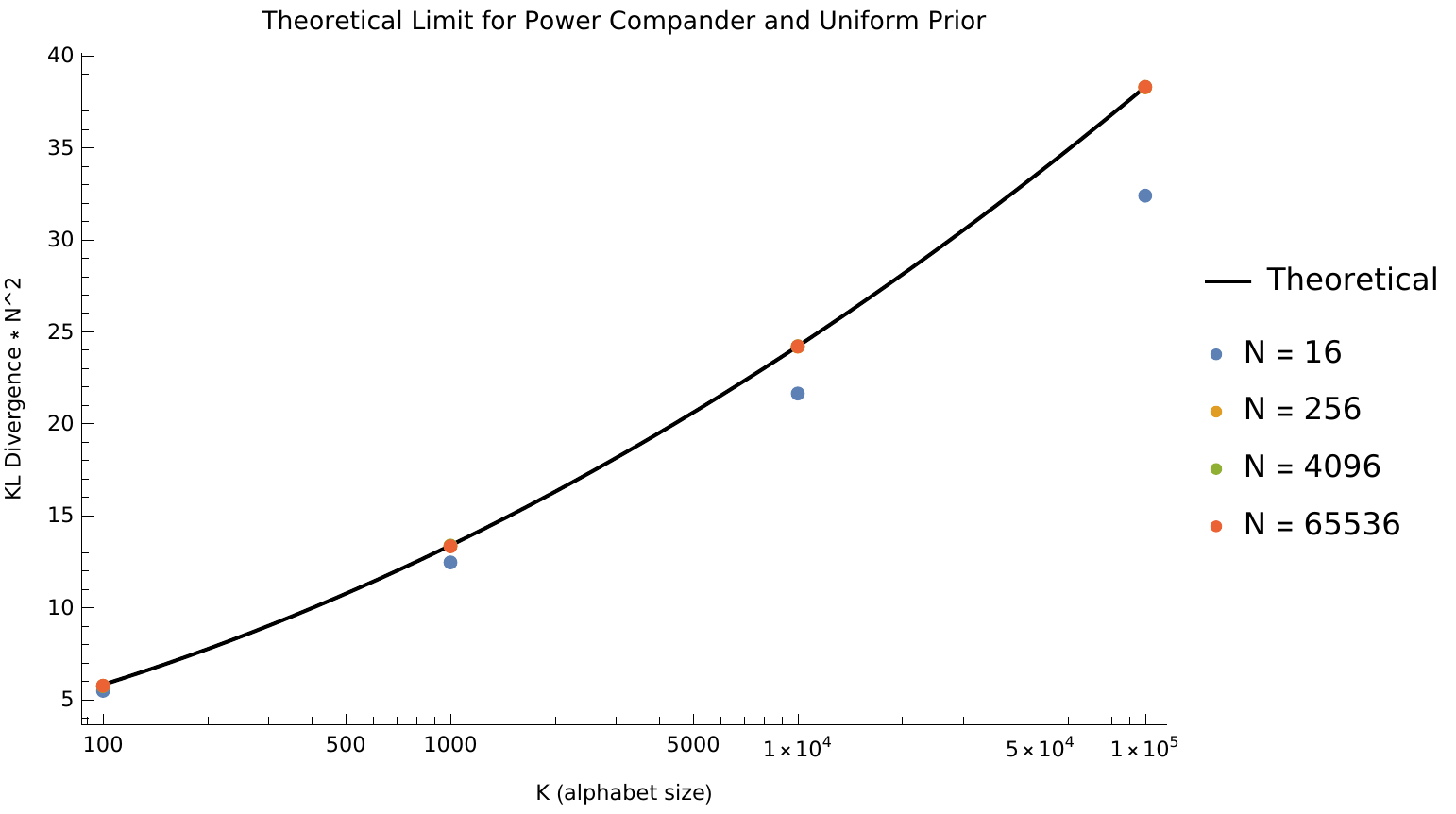}
    \includegraphics[scale = .5]{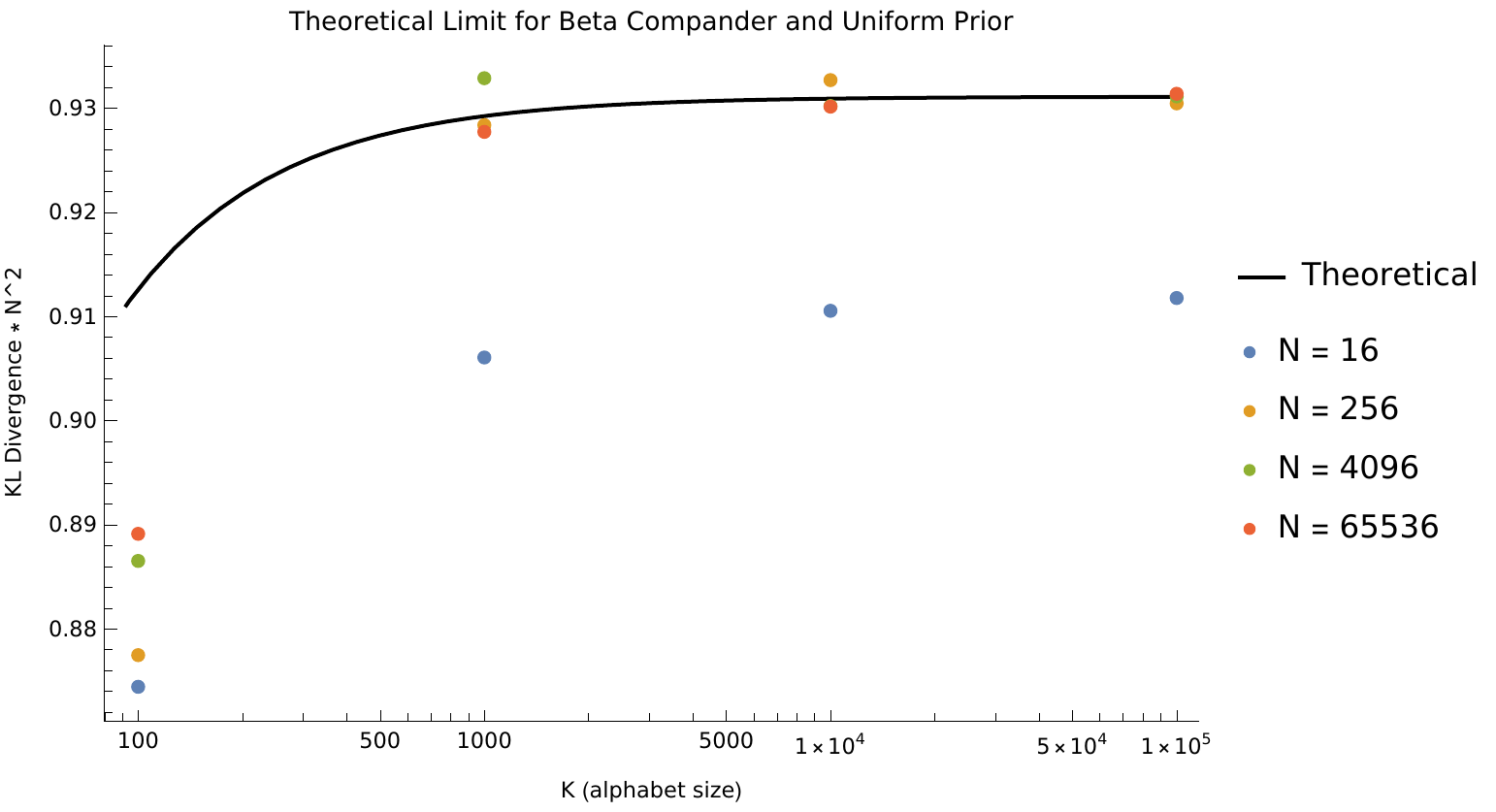}
    \caption{Comparing theoretical expression $\singleloss(p, f)$ with experimental result. The KL divergence value of the experimental results are multiplied to $N^2$ in order to be comparable to $\singleloss(p, f)$.}
    \label{fig::theoretical_compare}
\end{figure*}

\section{Minimax and Approximate Minimax Companders}

\label{sec::appendix_minimax}

In this appendix, we analyze the minimax compander and approximate minimax compander. Specifically, we analyze the constant $c_\az$, to show that it falls in $[1/4,3/4]$ (\Cref{sec::minimax_const_analysis}) and that $\lim_{\az \to \infty} c_\az = 1/2$ (\Cref{sec::lim_c_k}). We also show that when $c_\az$ is close to $1/2$, the approximate minimax compander (which is the same as the minimax compander except it replaces $c_\az$ with $1/2$) has performance close to the minimax compander against all priors $p \in \cP$ (\Cref{sec::proof_L_appx_minimax_compander}). 

\subsection{Analysis of Minimax Companding Constant}

\label{sec::minimax_const_analysis}

\begin{figure}
    \centering
    \includegraphics[scale = .4]{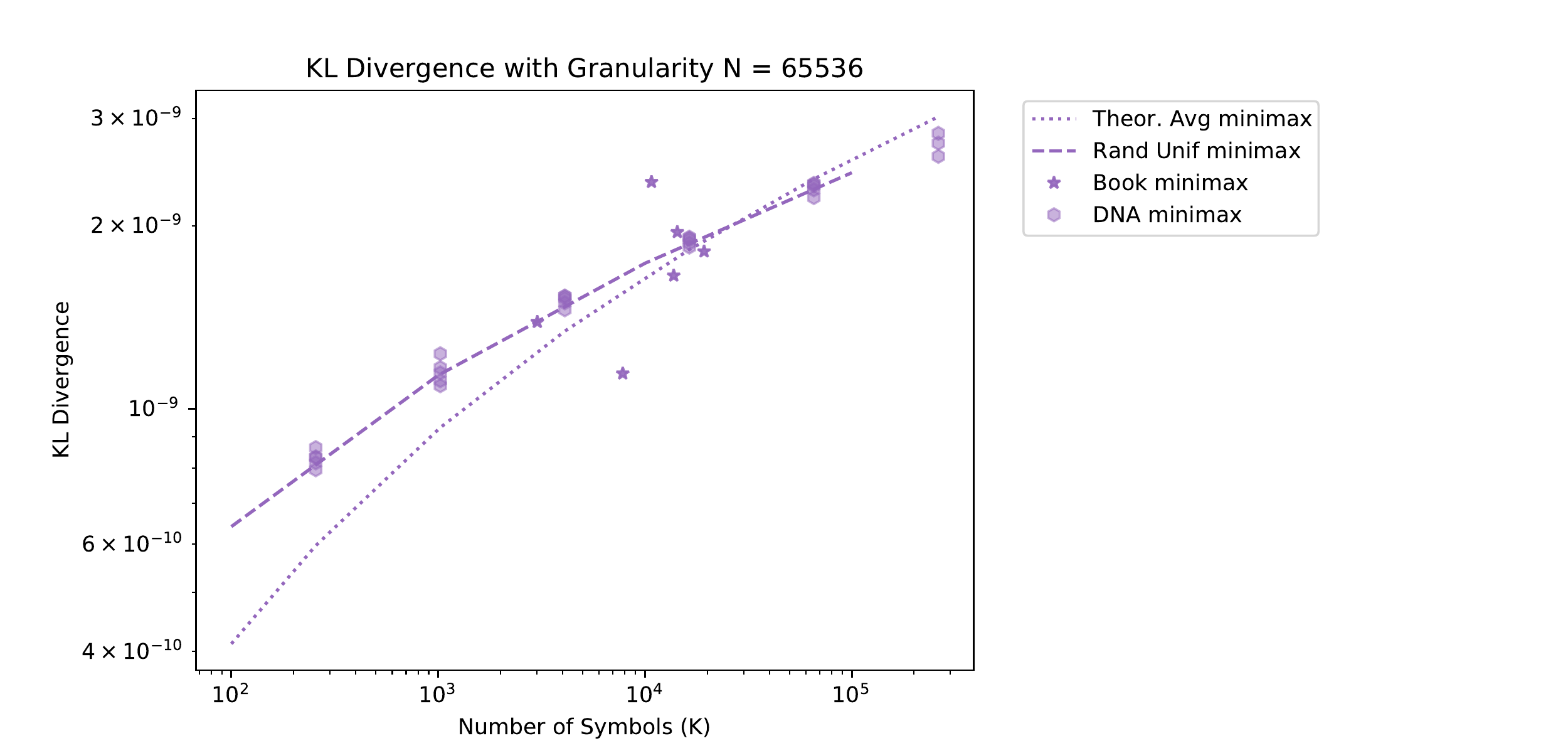}
    \caption{Comparing theoretical performance \eqref{eq::L_dagger_value_minimax} of the approximate minimax compander to experimental results.}
    \label{fig::theoretical_minimax}
\end{figure}

\subsubsection{Determining bounds on $c_\az$}

\label{sec::bounds_c_k}

If $a_\az, b_\az \geq 0$, then $p(x)$ is well-behaved (and bigger than $0$). 

%Satisfying the constraints that $p(x)$ is a density that sums to $1$, we get the relation that
%\begin{align}
%    b_\az &= 4 a_\az^{-2} - a_\az\,.
%\end{align}

%This gives $\int_0^1 p(x) \, dx = 1$. Furthermore, $4 a_\az^{-2} - a_\az \geq 0$ if and only if $a_\az \leq 2^{2/3}$, and $\bbE_{X \sim p}[X]$ increases with $a_\az$, ending at $1/3$ when $\theta = 2^{2/3}$ (this can be seen by noting that setting $a_\az = 2^{2/3}$ yields $p(x) = \frac{1}{2} x^{-1/2}$, which incidentally is an excellent sanity check as we know that this is the worst case prior when $\az = 3$ we get from \Cref{cor::worst_single_distribution}). Thus, setting the right $\theta$ satisfies all the constraints and the first-order optimality conditions and hence is globally optimal as well.

We need $a_\az$ and $b_\az$ to be such that $p(x)$ is a density that integrates to $1$ and also that $p(x)$ has expected value of $1/\az$. To do this, first we compute that
\begin{align}
    \bbE_{X\sim p}[X] %= \int_{0}^{1} p(x) x \, dx 
    &= \int_{0}^{1} x \left(a_\az x^{1/3} + b_\az x^{4/3}\right)^{-3/2} \, dx \\
    &= \frac{-2}{b_\az\sqrt{a_\az + b_\az}} + \frac{2 \arcsinh\left(\sqrt{\frac{b_\az}{a_\az}} \right)}{b_\az^{3/2}}\,.
\end{align}
The constraint that $\int_0^1 p(x) \, dx = 1$ requires that $a_\az \sqrt{a_\az+b_\az} = 2$. We can use this to get
\begin{align}
    \bbE_{X\sim p}[X] &= \frac{-a_\az}{b_\az} + \frac{a_\az \sqrt{\frac{a_\az}{b_\az} + 1} \arcsinh\left(\sqrt{\frac{b_\az}{a_\az}} \right)}{b_\az}\\
    & = \frac{-1}{r} + \frac{ \sqrt{\frac{1}{r} + 1} \arcsinh\left(\sqrt{r} \right)}{r}\\
    & = \frac{-1}{r} + \frac{ \sqrt{\frac{1}{r} + 1} \log\left( \sqrt{r} + \sqrt{r+1} \right)}{r}\label{eq::exact_expectation_r}
\end{align}
where we use $r = b_\az/ a_\az$.  We will find upper and lower bounds in order to approximate what $r$ should be. Using \eqref{eq::exact_expectation_r}, we can get
\begin{align}
    \bbE_{X\sim p}[X] & \leq \frac{1}{2} \frac{\log r}{r}
\end{align}
so long as $r > 3$. If we choose $r = c_1 \az \log \az$ and set $c_1 = .75$,
then
\begin{align}
    \bbE_{X\sim p}[X] & \leq \frac{1}{2} \frac{\log (c_1 \az \log \az)}{c_1 \az \log \az} \\
    &\leq \frac{1}{2 c_1 \az} + \frac{\log \log \az}{2 c_1 \az \log \az} + \frac{\log c_1}{2 c_1 \az \log \az} \leq \frac{1}{\az} 
\end{align}
so long as $\az> 4$.
Similarly, we have
\begin{align}
    \bbE_{X\sim p}[X] & \geq \frac{1}{3} \frac{\log r}{r}
\end{align}
for all $r$.  If we choose $r = c_2 \az \log \az$ and set $c_2 = .25$,
then
\begin{align}
    \bbE_{X\sim p}[X] & \geq \frac{1}{3} \frac{\log (c_2 \az \log \az)}{c_2 \az \log \az} \geq \frac{1}{\az} 
\end{align}
so long as $\az > \dpfconst$. 

Changing the value of $c$ changes the value of $\bbE_{X\sim p}[X]$ continuously. Hence, for each $\az > \dpfconst$, there exists a $c_\az$ so that if $r = c_\az \az \log \az$, then
\begin{align}
    \bbE_{X\sim p}[X]  = \frac{1}{\az}\,.
\end{align}
such that $.25 < c_\az < .75$.

This proves the result for $\az > 24$; numerical evaluation of $c_\az$ for $\az = 5, 6, \dots, 24$ then confirms that the result holds for all $\az > 4$.

\subsubsection{Limiting value of $c_\az$}

\label{sec::lim_c_k}

\begin{lemma}\label{lem::c_k_half}
In the limit, $c_\az \to 1/2$.
\end{lemma}

\begin{proof}
We start with $r = \frac{b_\az}{a_\az} = c_\az \az \log \az$,
and we need to meet the condition that
\begin{align}
\frac{-1}{r} + \frac{ \sqrt{\frac{1}{r} + 1} \log\left( \sqrt{r} + \sqrt{r+1} \right)}{r} = \frac{1}{\az}\,.
\end{align}
Substituting we get
\begin{align}
    \frac{1}{\az} &=  \frac{-1}{c_\az \az \log \az} + \sqrt{\frac{1}{c_\az \az \log \az} + 1}  \nonumber \\  & \quad \quad \frac{  \log\left( \sqrt{c_\az \az \log \az} + \sqrt{c_\az \az \log \az+1} \right)}{c_\az \az \log \az} \\
   \implies &c_\az   = \frac{-1}{ \log \az} + \sqrt{\frac{1}{c_\az \az \log \az} + 1} \nonumber \\ & \quad \quad \frac{  \log\left( \sqrt{c_\az \az \log \az} + \sqrt{c_\az \az \log \az+1} \right)}{  \log \az} \,.
\end{align}
Let $c = \lim_{\az \to \infty} c_\az$. Since $c_\az$ is bounded, we know that $\lim_{\az \to \infty} c_\az \az \log \az = \infty$ since $c_\az$ is bounded below by $1/4$; additionally $\log c_\az$ is bounded (above and below) since for $K > 4$ we have $c_\az \in [1/4, 3/4]$.

\begin{align}
c &= 
    \lim_{\az \to \infty}\frac{-1}{ \log \az} \eqlinebreakshort +  \sqrt{\frac{1}{c_\az \az \log \az} + 1} \eqlinebreakshort\frac{ \log\left( \sqrt{c_\az \az \log \az} + \sqrt{c_\az \az \log \az+1} \right)}{  \log \az} \\
    &= 0 + 1 \lim_{\az \to \infty} \frac{  \log\left( 2\sqrt{c_\az \az \log \az}  \right)}{  \log \az}\\
    &=  \lim_{\az \to \infty} \frac{\log 2 + \frac{1}{2} \log c_\az + \frac{1}{2} \log \az + \frac{1}{2} \log \log \az}{\log \az}\\
    & = \frac{1}{2}\,.
\end{align}

\end{proof}

\subsection{Approximate Minimax Compander vs. Minimax Compander}

\label{sec::proof_L_appx_minimax_compander}

%\jennifer{using a, b as notation seems bad. also using h}

For any $\az$, $c_\az$ can be approximated numerically. %While $p^*_\az$ must use the exact value (otherwise it's not in $\cP_{1/\az}$), 
To simplify the quantizer, recall we can use $c_\az \approx \frac{1}{2}$ for large $\az$ to get the approximate minimax compander \eqref{eq::appx-minimax-compander}.

%\begin{align} \label{eq::appx-minimax-compander}
%\comp^{**}_\az(x) := \frac{\arcsinh(\sqrt{(1/2) (\az %\log \az) \, x})}{\arcsinh(\sqrt{(1/2) \az \log \az})}
%\end{align}
This is close to optimal without needing to compute $c_\az$. Here we prove \Cref{thm::approximate-minimax-compander}.

%\begin{theorem} \label{thm::approximate-minimax-compander}
%If $c_\az \in [\frac{1}{2 (1 + \varepsilon)}, \frac{1 + \varepsilon}{2}]$, then for any $p \in \cP$,
%\begin{align}
%    \singleloss(p,\comp^{**}_\az) \leq (1+ %\varepsilon) \singleloss(p,\comp^*_\az)
%\end{align}
%\end{theorem}

\begin{proof}
Since $\comp^*_\az, \comp^{**}_\az \in \cF^\dagger$, we know that \begin{align}
    \singleloss(p,\comp^*_\az) = L^\dagger(p,\comp^*_\az) ~\text{ and }~ \singleloss(p,\comp^{**}_\az) = L^\dagger(p,\comp^{**}_\az)\,.
\end{align}
We define the corresponding asymptotic local loss functions
\begin{align}
    \locloss^*(x) &= \frac{1}{24} (\comp^*_\az)'(x)^{-2} x^{-1} 
    \\\locloss^{**}(x) &= \frac{1}{24} (\comp^{**}_\az)'(x)^{-2} x^{-1}
\end{align}
so that our goal is to prove
\begin{align}
    \int \locloss^{**} \, dp \leq (1+\varepsilon) \int \locloss^* \, dp\,.
\end{align}

Let $\gamma^* = c_\az (\az \log \az)$ and $\gamma^{**} = \frac{1}{2} (\az \log \az)$ (the constants in $\comp^*_\az$ and $\comp^{**}_\az$ respectively) and let $\phi^*(x) = \arcsinh(\sqrt{\gamma^* x})$ and $\phi^{**}(x) = \arcsinh(\sqrt{\gamma^{**} x})$. Then
\begin{align}
    (\phi^*)'(x) &= \frac{\sqrt{\gamma^*}}{2 \sqrt{x} \sqrt{\gamma^* x+1}} 
    \\ \text{and } ~ (\phi^{**})'(x) &= \frac{\sqrt{\gamma^{**}}}{2 \sqrt{x} \sqrt{\gamma^{**} x+1}} \,.
\end{align}
Note that $\comp^*_\az(x) = \phi^*(x)/\phi^*(1)$ and $\comp^{**}_\az(x) = \phi^{**}(x)/\phi^{**}(1)$. We now split into two cases: (i) $c_\az > 1/2$ and (ii) $c_\az < 1/2$. 

In case (i) (which implies $\gamma^* > \gamma^{**}$, and note that $\gamma^*/\gamma^{**} = 2c_\az \leq 1 + \varepsilon$), we get for all $x \in [0,1]$,
\begin{align}
    \frac{(\phi^*)'(x)}{(\phi^{**})'(x)} &= \sqrt{\frac{\gamma^*}{\gamma^{**}}} \sqrt{\frac{\gamma^{**} x + 1}{\gamma^* x + 1}} 
    \\&\quad\quad \in [1, \sqrt{\gamma^*/\gamma^{**}}] 
    \ \subseteq [1, \sqrt{1+\varepsilon}]
\end{align}
since $\sqrt{\frac{\gamma^{**} x+1}{\gamma^* x+1}} \in [\sqrt{\gamma^{**} / \gamma^*}, 1]$. %Thus,
\iffalse
\begin{align}
    \phi^*(1) &= \int_0^1 (\phi^*)'(t) \, dt \\ &\leq \sqrt{1 + \varepsilon} \int_0^1 (\phi^{**})'(t) \, dt \\ &\leq (\sqrt{1 + \varepsilon}) \phi^{**}(1) 
\end{align}
\fi
Because $\gamma^* \geq \gamma^{**}$ and $\arcsinh$ is an increasing function, we know that $\phi^*(1) \geq \phi^{**}(1)$.
%The other side of this implies $\phi^*(1) \geq \phi^{**}(1)$ (also the fact that $a \geq b$ and $\arcsinh$ is an increasing function). 
Thus, for any $x \in [0,1]$,
\begin{align}
    (\comp^{**}_\az)'(x) &= \frac{(\phi^{**})'(x)}{\phi^{**}(1)} \\ &\geq \frac{\frac{1}{\sqrt{1+\varepsilon}}(\phi^{*})'(x)}{\phi^{*}(1)} \\& = \frac{1}{\sqrt{1+\varepsilon}} (\comp^*_\az)'(x)
    \\ \implies (\comp^{**}_\az)'(x)^{-2} &\leq (1 + \varepsilon) (\comp^{*}_\az)'(x)^{-2}
    \\ \implies g^{**}(x) &\leq (1 + \varepsilon) g^*(x)
    \\ \implies \int g^{**} \, dp &\leq (1 + \varepsilon) \int g^* \, dp
\end{align}
which is what we wanted to prove.

Case (ii), where $c_\az < 1/2$ (implying $\gamma^{**} > \gamma^*$) can be proved analogously:
\begin{align}
    \frac{(\phi^{**})'(x)}{(\phi^*)'(x)} & = \sqrt{\frac{\gamma^{**}}{\gamma^*}} \sqrt{\frac{\gamma^* x + 1}{\gamma^{**} x + 1}} 
    \\ & \quad\quad \in [1, \sqrt{\gamma^{**}/\gamma^*}] \subseteq [1, \sqrt{1+\varepsilon}]
\end{align}
which then gives us $(\phi^{**})'(x) \geq (\phi^*)'(x)$ and
\begin{align}
    \phi^{**}(1) &= \int_0^1 (\phi^{**})'(t) \, dt \\ &\leq \sqrt{1 + \varepsilon} \int_0^1 (\phi^{*})'(t) \, dt \\ &\leq (\sqrt{1 + \varepsilon}) \phi^{*}(1) \,.
\end{align}
Thus, for any $x \in [0,1]$,
\begin{align}
    (\comp^{**}_\az)'(x) &= \frac{(\phi^{**})'(x)}{\phi^{**}(1)} \\ &\geq \frac{(\phi^{*})'(x)}{(\sqrt{1+\varepsilon})\phi^{*}(1)} \\& = \frac{1}{\sqrt{1+\varepsilon}} (\comp^*_\az)'(x)
    \\ \implies (\comp^{**}_\az)'(x)^{-2} &\leq (1 + \varepsilon) (\comp^{*}_\az)'(x)^{-2}
    \\ \implies g^{**}(x) &\leq (1 + \varepsilon) g^*(x)
    \\ \implies \int g^{**} \, dp &\leq (1 + \varepsilon) \int g^* \, dp
\end{align}
completing the proof for both cases.
\end{proof}

We show the comparison of the theoretical (asymptotic in $\az$ result) of the approximate minimax compander with the experimental results in \Cref{fig::theoretical_minimax}.

\iflong %%%%%%%%% worst-case results

%\newpage

\section{Worst-Case Analysis}

\label{sec::worst-case_analysis}

\newcommand{\len}{\text{len}}

In this section, we prove \Cref{thm::worstcase_power_minimax} which applies both to the minimax compander and the power compander. Since we are dealing with worst-case (i.e. not a random $\bx$) the centroid is not defined; therefore this theorem works with the \emph{midpoint decoder}. Thus, the (raw) decoded value of $x$ is $\bar{y}_{(n_N(x))}$.

Additionally, we are not using the raw reconstruction but the normalized reconstruction, and hence it does not suffice to deal with a single letter at a time. Thus, we will work with a full probability vector $\bx \in \triangle_{\az-1}$.

\begin{proof}[Proof of \Cref{thm::worstcase_power_minimax} and \eqref{eq::power_worst_case_bound} in \Cref{thm::power_compander_results}]

Let $\bx \in \triangle_{\az-1}$ be the vector we are quantizing, with $i$th element (out of $\az$, summing to $1$) $x_i$; since we are dealing with midpoint decoding, our (raw) decoded value of $x_i$ is $\bar{y}_{n_N(x_i)}$. For simplicity, let us denote it as $\bar{y}_i$, and the normalized value as $\normvar_i = \bar{y}_i / \big(\sum_j \bar{y}_j\big)$.

Let $\delta_i = \bar{y}_i - x_i$ be the difference between the raw decoded value $\bar{y}_i$ and the original value $x_i$. Then:
\begin{align}
    D_{\kl}&\left( \bx \| \bnormvar \right) = \sum_{i} x_i \log \frac{x_i}{\normvar_i} \\
    & = \sum_{i} x_i \log \frac{x_i}{\bar{y}_i} + \log\Big(\sum_i \bar{y}_i\Big) \\ 
    & = \sum_{i} (\bar{y}_i - \delta_i) \log \frac{\bar{y}_i - \delta_i}{\bar{y}_i} + \log\Big(1 + \sum_i \delta_i\Big)
    \,.
\end{align}
Next we use that $\log(1+\newvar) \leq \newvar$.
\begin{align}
     D_{\kl}\left( \bx \| \bnormvar \right)& \leq \sum_{i} (\bar{y}_i - \delta_i) \frac{-\delta_i}{\bar{y}_i} + \sum_i \delta_i \label{eq::use_log_ineq} \\ 
    &=\sum_i -\delta_i + \sum_i \frac{\delta_i^2}{\bar{y}_i} + \sum_i \delta_i\\
    & =\sum_i \frac{(\bar{y}_i - x_i)^2}{\bar{y}_i}
\end{align}
(note that in \eqref{eq::use_log_ineq} we used the inequality $\log(1+\newvar) \leq \newvar$ on \emph{both} appearances of the logarithm, as well as the fact that $\bar{y}_i - \delta_i = x_i \geq 0$).
%\nbyp{Remove all of the above and just use (7.31) from my book~\cite{ITbook}.} \textcolor{blue}{Jennifer-not the same thing. Bar y is not a probability, doesn't normalize to one. Other proofs assume it normalizes}

%Since there is no prior when considering the worst-case KL divergence, we use midpoint decoding. Let $\bar \normvar_i$ be the midpoint of the interval $x_i$ is mapped to. 

We now consider each bin $I^{(n)}$ induced by $\comp$. For simplicity let the dividing points between the bins be denoted by 
\begin{align}
    \beta_{(n)} = f^{-1}\Big(\frac{n}{N}\Big) = \bar{y}_{(n)} + r_{(n)}/2
\end{align}
(where $r_{(n)}$ is the width of the $n$th bin) so that $I^{(n)} = (\beta_{(n-1)}, \beta_{(n)}]$. %\nbyp{Remind what is $r_{(n)}$.} \textcolor{blue}{Done.}
Since all the
companders we are discussing are strictly monotonic, there is no ambiguity. Then, the Mean Value
Theorem (which we can use since the minimax compander, the approximate minimax compander, and the
power compander are all continuous and differentiable
%\nbyp{Did you want to say differentiable?} \textcolor{blue}{Yes, done.}
) says that, for each $I^{(n)}$ there is some value $\newvar_{(n)}$ such that
\begin{align}
    f'(\newvar_{(n)}) = \frac{f(\beta_{(n)}) - f(\beta_{(n-1)})}{\beta_{(n)} - \beta_{(n-1)}} = N^{-1} r_{(n)}^{-1}
\end{align}
(since $f(\beta_{(n)}) - f(\beta_{(n-1)}) = n/N - (n-1)/N = 1/N$ and $\beta_{(n)} - \beta_{(n-1)} = r_{(n)}$ by definition).

%For the $n$th interval $I^{(n)}$ (these intervals are determined by the compander function $f$), let $\acute{y}_{(n)}$ be the value associated with the mean value of the interval, defined as follows: For a continuous and smooth companding function $f(\newvar)$, and interval $n$ with boundaries $b_{n-1}$ and $b_{n}$, let the mean value $\acute y_{(n)}$ of interval $I^{(n)}$  be the $\newvar$ be such that
%\begin{align}
    %f'(\newvar) = \frac{f(b_{n}) - f(b_{n-1})}{b_{n} - b_{n-1}} = \frac{1}{N} \frac{1}{b_{n} - b_{n-1}}\,.
%\end{align}
%\aviv{Confused! I thought the decoder was the midpoint? Now it's some other thing whose existence is guaranteed by the Mean Value Theorem or something??}

Thus, we can re-write this as follows: 
\begin{align}
    r_{(n)} = N^{-1} f'(\newvar_{(n)})^{-1} \, .
\end{align}
We will also denote the following for simplicity: $I_i = I^{(n_N(x_i))}$; $r_i = r_{(n_N(x_i))}$; and $\newvar_i = \newvar_{(n_N(x_i))}$ (the bin, bin length, and bin mean value corresponding to $x_i$).

%For any $i$, let $\acute y_i = \acute y_{(n)}$ where $n$ is such that $x_i$ falls in interval $I^{(n)}$. Define the length $\len(x_i)$ as the length of the interval $I^{(n)}$ which $x_i$ falls in. If $x_i$ is in interval $I^{(n)}$, then
%\begin{align}
%    \len(x_i) &\eqdef  b_{n} - b_{n-1} = \frac{1}{N f'(\acute y_i)}\,.
%\end{align}

Trivially, since $\newvar_i \in I_i$, we know that $\frac{\newvar_i}{2} \leq \bar{y}_i$. 
Thus, we can derive (since $\bar{y}_i$ is the midpoint of $I_i$ and $x_i \in I_i$, we know that $|\bar{y}_i - x_i| \leq r_i/2$) that
\begin{align}
    D_{\kl} \left( \bx \| \bnormvar \right)& \leq \sum_i \frac{(\bar{y}_i - x_i)^2}{\bar{y}_i}\\
    & \leq \frac{1}{4}\sum_i \frac{r_i^2}{\bar{y}_i}\\
    & \leq \frac{1}{4}\sum_i \frac{1}{N^2 (\newvar_i / 2) (f'(\newvar_i))^2}\\
    & = \frac{1}{2} N^{-2} \sum_i \frac{1}{ \newvar_i (f'(\newvar_i))^2}\,. \label{eq::worst-case-general-mean-value-bound}
\end{align}
Note that while we are using midpoint decoding for our quantization, for the purposes of analysis, it is more convenient to express the all the terms in the KL divergence loss using the mean value. %To simplify notation going further, we let $y_i = \acute y_i$.

We now examine the worst case performance of the three companders: the power compander, the minimax compander, and the approximate minimax compander.

\emph{Power compander:} In this case, we have
\begin{align}
   f(x) = x^{s} \text{ and } f'(x) = s x^{s-1}
\end{align}
for $s = \frac{1}{\log \az}$ (which is optimal for minimizing raw distortion against worst-case priors). This yields
\begin{align}
    D_{\kl}\left( \bx \| \bnormvar \right)& \leq \frac{1}{2} N^{-2} s^{-2} \sum_i \frac{1}{\newvar_i \newvar_i^{2s - 2}}\\
    & = \frac{1}{2} N^{-2} s^{-2}\sum_i \newvar_i^{1 - 2s}
    \,.
\end{align}

So long as $s < 1/2$ (which occurs for $\az > 7$), the function $\newvar_i^{1 - 2s}$ is concave in $\newvar_i$. Thus, replacing all $\newvar_i$ by their average will increase the value. Furthermore, $\az^s = \az^{\frac{1}{\log \az}} = e$. Thus, we can derive: 
\begin{align}
    D_{\kl}\left( \bx \| \bnormvar \right)& \leq \frac{1}{2} N^{-2} s^{-2} \az \left(\frac{\sum_i \newvar_i}{\az}\right)^{1 - 2s} \\
    & = \frac{1}{2} N^{-2} (\log^2 \az) e^2 \Big(\sum_i \newvar_i\Big)^{1 - 2s} \\
     & \leq \frac{e^2}{2} N^{-2} (\log^2 \az)  \max\Big\{1, \sum_i \newvar_i\Big\}\label{eq::worst_case_power_bound_max}
     \,.
\end{align}

Next, we need to bound $\max\left\{1, \sum_i \newvar_i\right\}$. Assume that $\sum_i \newvar_i > 1$ (otherwise our bound is just $1$). Then, we note the following: $\sum_i x_i = 1$ by definition; $s^{-1} = \log \az$; and
\begin{align}
    r_i = N^{-1} f'(\newvar_i)^{-1} = N^{-1} s^{-1} \newvar_i^{1-s} \,.
\end{align}
This allows us to make the following derivation:
\begin{align}
    \sum_i | \newvar_i - x_i | &\leq \frac{1}{2} \sum_i r_i \\
    \implies \sum_i \newvar_i &\leq \sum_i x_i + \frac{1}{2} N^{-1} s^{-1} \sum_i \newvar_i^{1-s} \\
     &\leq 1 + \frac{1}{2} N^{-1} \log(\az) \az \left(\frac{\sum_i \newvar_i}{\az}\right)^{1-s} \label{eq::concavity-saves-the-day}\\ 
     &= 1 + \frac{e}{2} N^{-1} \log(\az)  \Big(\sum_i \newvar_i\Big)^{1-s} \label{eq::zounds} \\
     &\leq 1 + \frac{e}{2} N^{-1} \log(\az) \Big(\sum_i \newvar_i\Big) \,.
\end{align}
We get \eqref{eq::concavity-saves-the-day} by the same concavity trick: because $\newvar_i^{1-s}$ is concave in $\newvar_i$, replacing each individual $\newvar_i$ with their average can only increase the sum. We get \eqref{eq::zounds} because $\az^s = \az^{\frac{1}{\log \az}} = e$.

We can combine terms with $\sum_i \newvar_i$.
\begin{align}
    \left(1 - \frac{e}{2} N^{-1} \log \az \right)  \sum_i \newvar_i  \leq 1\,.
\end{align}
This implies that if $N > \frac{e}{2} \log \az$, then
\begin{align}
    \sum_i \newvar_i  &\leq \frac{1}{1 - \frac{e}{2} N^{-1} \log \az}\\
    &= \frac{N}{N -\frac{e}{2} \log \az} = 1 + \frac{e}{2} \frac{\log \az}{N - \frac{e}{2}\log \az} \label{eq::worst_case_power_bound_sumy}\,.
\end{align}
 Furthermore, if $N \geq e \log \az$, we get that $\sum_i \newvar_i \leq 2$. Combining \eqref{eq::worst_case_power_bound_max} with \eqref{eq::worst_case_power_bound_sumy}, we have
%\aviv{what about sum of y less than 1?}
\begin{align}
    \eqstartnonumshort D_{\kl}( \bx \| \bnormvar ) \eqbreakshort 
    &\leq \frac{e^2}{2} N^{-2} (\log^2 \az) \max\left\{1,  \left( 1 + \frac{e}{2} \frac{\log \az}{N - \frac{e}{2}\log \az}\right) \right\}\\
    & = \frac{e^2}{2} N^{-2} (\log^2 \az) \left( 1 + \frac{e}{2} \frac{\log \az}{N - \frac{e}{2}\log \az}\right)
\end{align}
for $N > \frac{e}{2} \log \az$. When $N \geq e \log \az$, this becomes the pleasing
\begin{align}
    D_{\kl}( \bx \| \bnormvar ) \leq e^2 N^{-2} \log^2 \az\,.
\end{align}

\emph{Minimax compander and approximate minimax compander:} Since they are very similar in form, it is convenient to do both at once. Let $c$ be a constant which is either $c_\az$ if we are considering the minimax compander, or $\frac{1}{2}$ if we're considering the approximate minimax compander; and let $\gamma = c \az \log \az$. Then our compander and its derivative will have the form
\begin{align}
    f(x) &= \frac{\arcsinh(\sqrt{\gamma x})}{\arcsinh(\sqrt{\gamma})} \\
    f'(x) &= \frac{1}{2 \arcsinh(\sqrt{\gamma})} \frac{\sqrt{\gamma}}{\sqrt{x}\sqrt{1 - \gamma x} } \\
    \implies f'(x)^{-1} &= 2 \arcsinh(\sqrt{\gamma}) \sqrt{\frac{x}{\gamma} + x^2}\,.
\end{align}
This then yields that
\begin{align}
    r_i &= N^{-1} f'(\newvar_i)^{-1} \\
    &= 2 N^{-1} \arcsinh(\sqrt{\gamma}) \sqrt{\frac{\newvar_i}{\gamma} + \newvar_i^2}\,.
\end{align}

%To start, to simplify notation, we will use $A = a \az \log \az = c_\az \az \log \az$ (or for the approximate minimax compander, $A = a \az \log \az = \frac{1}{2} \az \log \az)$).
%\begin{align}
%    f(x) &= \frac{\arcsinh(\sqrt{A x})}{\arcsinh(\sqrt{A })}\\
%    f'(x) &= \frac{1}{2 \arcsinh(\sqrt{A})} \frac{\sqrt{A}}{\sqrt{x}\sqrt{1 - Ax} } \\
%    \frac{1}{f'(x)} &= 2 \arcsinh(\sqrt{A}) \sqrt{\frac{x}{A} + x^2}
%\end{align}

Then we can derive from \eqref{eq::worst-case-general-mean-value-bound} that
\begin{align}
    &D_{\kl}( \bx \| \bnormvar) \leq \frac{1}{2} N^{-2} (2 \arcsinh(\sqrt{\gamma}))^2 \sum_i \frac{\frac{\newvar_i}{\gamma} + \newvar_i^2}{\newvar_i} \\
    &= 2 N^{-2} (\arcsinh(\sqrt{\gamma}))^2 \left(\frac{\az}{\gamma} + \sum_i \newvar_i \right)
    \\ &\leq 2 N^{-2} (\arcsinh(\sqrt{\gamma}))^2 \left(\frac{\az}{\gamma} + \max\left\{1,\sum_i  \newvar_i\right\} \right) \label{eq::worst-case-kl-bound-minimax}\,.
\end{align}
Assuming that $\sum_i \newvar_i > 1$ (otherwise the bound is just $1$),
\begin{align}
    \sum_i | \newvar_i - x_i | &\leq \sum_i \frac{r_i}{2} \\
    \implies \sum_i \newvar_i &\leq \sum_i x_i + N^{-1} \arcsinh(\sqrt{\gamma}) \sum_i  \sqrt{\frac{\newvar_i}{\gamma} + \newvar_i^2} \\
    &= 1 + N^{-1} \arcsinh(\sqrt{\gamma}) \sum_i  \sqrt{\frac{\newvar_i}{\gamma} + \newvar_i^2} \label{eq::minimax_sum_yi_bound}\,.
\end{align}
To bound the sum in \eqref{eq::minimax_sum_yi_bound}, using the fact that $\sqrt{\cdot}$ is concave (so averaging the inputs of a sum of square roots makes it bigger), we get
\begin{align}
\sum_i  \sqrt{\frac{\newvar_i}{\gamma} + \newvar_i^2}
    & \leq \sum_i  \sqrt{\frac{\newvar_i}{\gamma}} + \sqrt{\newvar_i^2} \\
    &\leq  \az \left(\frac{\sum_i \newvar_i}{\az (c \az \log \az)}\right)^{1/2} + \sum_i \newvar_i  \\
    &\leq  \left(\frac{\sum_i \newvar_i}{c  \log \az}\right)^{1/2} + \sum_i \newvar_i  \\
    &\leq  \frac{\sum_i \newvar_i}{(c \log \az)^{1/2}} + \sum_i \newvar_i  \\
    &= \Big( \sum_i \newvar_i \Big) \left(1 + \frac{1}{(c \log \az)^{1/2}}\right) \\
     &=  \eta\Big( \sum_i \newvar_i\Big) 
\end{align}
where $\eta = 1 +(c \log \az)^{-1/2}$.
Then \eqref{eq::minimax_sum_yi_bound} becomes
\begin{align}
    \sum_i \newvar_i &\leq 1 + \eta N^{-1} \arcsinh(\sqrt{\gamma}) \Big( \sum_i \newvar_i\Big)
    \,.
\end{align}
Since we have $\sum_i \newvar_i$ on both sides of the equation, we can combine these terms like before.
\begin{align}
    &(1 - \eta N^{-1} \arcsinh(\sqrt{\gamma})) \sum_i \newvar_i \leq 1 \\
    \implies &\sum_i \newvar_i \leq \frac{N}{N - \eta \arcsinh(\sqrt{\gamma})}
\end{align}
if $N > \eta \arcsinh (\sqrt{\gamma}) $.
Combining these and using the expression $\arcsinh(\sqrt{\newvar}) = \log (\sqrt{\newvar + 1} + \sqrt{\newvar}) \leq \log(2 \sqrt{\newvar} + 1)$ we get from \eqref{eq::worst-case-kl-bound-minimax} that
%\begin{align}
    %D_{\kl}(\bx \| \bnormvar) &\leq 2 N^{-2} (\arcsinh(\sqrt{\gamma}))^2 \left(\frac{\az}{\gamma} + \max\left\{1,\sum_i  \newvar_i\right\} \right) %\label{eq::worst-case-kl-bound-minimax}
%\end{align}

\begin{align}
    \eqstartnonumshort D_{\kl}( \bx\| \bnormvar) \eqbreakshort
    &\leq 2 N^{-2} (\arcsinh(\sqrt{\gamma}))^2 \eqlinebreakshort \left(\frac{\az}{\gamma} + \frac{N}{N - \eta \arcsinh(\sqrt{\gamma})}\right) \\
    & = 2 N^{-2} (\arcsinh(\sqrt{c \az \log \az}))^2 \eqlinebreakshort \left(\frac{\az}{c \az \log \az} + \frac{N}{N - \eta \arcsinh(\sqrt{c \az \log \az})}\right) \\
    & \leq 2N^{-2} (\log(2\sqrt{c \az \log \az} + 1))^2 \eqlinebreakshort\left(\frac{1}{c  \log \az} + \frac{N}{N - \eta \log(2\sqrt{c \az \log \az} + 1)}\right) \,.
\end{align}
This holds for all $N > \eta \log(2\sqrt{c \az \log \az} + 1)$; furthermore, if $N > 3 \eta\log(2\sqrt{c \az \log \az} + 1)$, the second term in the parentheses is at most $3/2$ (and if $N$ is larger, this term goes to $1$). 
Recall $c$ is between $1/4$ and $3/4$ (as it is either $c_\az$ or $1/2$) when $\az >4$. Then, we know that for all $\az >4$ that $\eta < 2.57 \dots$ and $1/(c \log \az) < 5/2$.
Thus, for 
\begin{align}
N &> 8\log(2\sqrt{c \az \log \az} + 1)
\\ &> 3 (2.6) \log(2\sqrt{c \az \log \az} + 1)
\\ &> 3 \eta\log(2\sqrt{c \az \log \az} + 1)
\end{align}
we can bound the entire parenthesis term by $4$. Then,
\begin{align}
    D_{\kl}&( \bx \| \bnormvar)
    \leq 8 N^{-2} (\log(2\sqrt{c \az \log \az} + 1))^2
    \\ &\leq 8 N^{-2} (\log(3\sqrt{c \az \log \az}))^2
    \\ &= 2 N^{-2} (\log (c \az \log \az) + 2\log 3)^2 \label{eq::numerical-bd-01}
    \\ & = 2 N^{-2} \Big(1 + O\Big(\frac{\log \log \az}{\log \az} \Big)\Big) \log^2 \az \,.
\end{align}
Note that whether $c$ is $c_\az$ or $1/2$, it is always between $1/4$ and $3/4$, and so it has no effect on the order of growth. We also note that the above (stated more crudely) is an order of growth within $O(N^{-2}\log^2 \az)$.

We can obtain a relatively clean upper bound on the error term $O\big(\frac{\log \log \az}{\log \az}\big)$ by setting $c = 3/4$ (which is larger than the whole range of possible values); in this case, numerically computing \eqref{eq::numerical-bd-01}, we get that the error term is at most $18 \frac{\log \log \az}{\log \az}$ for $\az > 4$. The quantity  $18 \frac{\log \log \az}{\log \az}$ has a maximum value of around $6.62183$.
%Recall that $a$ is $c_\az$ which approaches $1/2$ as shown in \Cref{lem::c_k_half} (or for the approximate minimax compander, simply $a = 1/2$).
\end{proof}
\fi %%%%%%%%% worst-case results

The statement above (which is used for \Cref{thm::worstcase_power_minimax}) computes constants for our bound which work for both the minimax compander and approximate minimax compander and only requires that $\az > 4$.

If we are only concerned with large alphabet sizes,
to improve the constants for the approximate minimax compander (where $c = 1/2$), we can instead use the following: For $\az \geq 55$ and $N > 6 \log(2\sqrt{c \az \log \az} + 1)$,
\begin{align}
     D_{\kl}&( \bx \| \bnormvar) \leq  N^{-2} \Big(1 + 6\frac{\log \log \az}{\log \az} \Big) \log^2 \az \,.
\end{align}

\section{Uniform Quantization}

\label{sec::uniform}

In this section, we examine of the performance of uniform quantization under KL divergence loss. This is the same as applying the truncate compander.

First, we will prove \eqref{eq::uniform_achieve} of \Cref{rmk::loss_with_uniform}.

\begin{proof}[Proof of \eqref{eq::uniform_achieve}]

Let $p$ be the single-letter distribution which is uniform over $\left[0, {2}/{\az}\right]$ for each symbol. Specifically, the probability density function is
\begin{align}
    p(x) = \frac{\az}{2} \text{ for } x \in \left[0, \frac{2}{\az}\right]
\end{align}
and since the expected value under $p$ is $1/\az$, we have that $p \in \cP_{1/\az}$.

We want to compute the single-letter loss for $p$, but notice that we cannot use \Cref{thm::asymptotic-normalized-expdiv} to do so, since the quantity $L^\dagger(p, f)$ is not finite here (this is not surprising since we are showing a case where the dependence of $\widetilde{L}(p,f,N)$ on $N$ is larger than $\Theta(N^{-2})$). Thus we need to compute the single-letter loss starting with 
\eqref{eq::raw-ssl}.

\begin{align}
     \singleloss(p, \comp, N) &=  \bbE_{X \sim p} \big[ X \log ( X/\widetilde{y}(X)) \big]
     \\& = \sum_{n = 1}^N \int_{I^{(n)}} p(x) x \log \frac{x}{\tilde y_n} dx
     \\& = \sum_{n = 1}^N \int_{I^{(n)}} \bbI\{x < 2/\az\}\frac{\az}{2} x \log \frac{x}{\tilde y_n} dx
     \\ &\geq \frac{\az}{2}\sum_{n = 1}^{\lfloor 2N /\az \rfloor} \int_{n/N}^{(n+1)/N} x \log \frac{x}{\tilde y_n} dx 
      \\ &= \frac{\az}{2}\sum_{n = 1}^{\lfloor 2N /\az \rfloor} \int_{\tilde y_n - \frac{r}{2}}^{\tilde y_n + \frac{r}{2}} x \log \frac{x}{\tilde y_n} dx 
\end{align}
where we let $r = 1/N$.

Using the Taylor expansion for $\log(1+x)$, we can get that
\begin{align}
    \int_{\tilde y_n - \frac{r}{2}}^{\tilde y_n + \frac{r}{2}} x \log \frac{x}{\tilde y_n} dx = \frac{r^3}{24 \tilde y_n} + O\left(\frac{r^5}{\tilde y_n^3} \right)\,.
\end{align}
%\jennifer{Can write more details about Taylor Expansion}

This gives that
\begin{align}
    \singleloss(p, \comp, N) &\geq \frac{\az}{2}\sum_{n = 1}^{\lfloor 2N /\az \rfloor} \frac{r^3}{24 \tilde y_n} - O\left(\frac{r^5}{\tilde y_n^3} \right)
    \\& = \frac{\az}{48} \frac{1}{ N^3} \sum_{n = 1}^{\lfloor 2N /\az \rfloor} \frac{1}{ \tilde y_n} - \sum_{n = 1}^{\lfloor 2N /\az \rfloor} O\left(\frac{1}{N^5 \tilde y_n^3} \right)\,.
\end{align}

Because the intervals are uniform, the centroid is the midpoint of each interval, which means that
\begin{align}
    \tilde y_n = \frac{n - 1/2}{N}\,.
\end{align}
This gives that
\begin{align}
     \sum_{n = 1}^{\lfloor 2N /\az \rfloor} \frac{1}{\tilde y_n} 
     &= \sum_{n = 1}^{\lfloor 2N /\az \rfloor} \frac{1}{ \frac{n - 1/2}{N}} \\&> N \sum_{n = 1}^{\lfloor 2N /\az \rfloor} \frac{1}{ n}
     \\ & > C_1 N \log (2N /\az )\,.
\end{align}
We also need to bound the smaller order terms to make sure they are not too big, 
\begin{align}
     \sum_{n = 1}^{\lfloor 2N /\az \rfloor} \frac{1}{\tilde y_n^3} 
     &< N^3\left(2^3 + \sum_{n = 2}^{\lfloor 2N /\az \rfloor} \frac{1}{ (n - 1)^3}\right)
     \\ & = N^3 C_3\,.
\end{align}
Combining these give
\begin{align}
     \singleloss(p, \comp, N) & \geq \frac{\az}{48 N^3} C_1 N \log (2N /\az ) -  O\left(\frac{1}{N^2} \right)
     \\& = \Omega\left(\frac{\az}{N^2} \log N \right)\,.
\end{align}
All the inequalities we used for the lower bound can easily be adjusted to make an upper bound. For instance, the floor function in the summation can be replaced with a ceiling function. The quantity $\tilde y_n$ can be rounded up or down and the inequalities approximating sums can have different multiplicative constants. This gives that for $p(x)$, we have
\begin{align}
    \singleloss(p, \comp, N) = \Theta\left(\frac{\az}{N^2} \log N \right)\,.
\end{align}
%\jennifer{do we need to show upper bound more precisely?}

Combining this single-letter density with the proof of  \Cref{prop::bound_worstcase_prior_exist} gives a prior $P$ over the simplex so that
\begin{align}
    \rawloss_\az(P, \comp, N) = \az \singleloss(p, \comp, N) =  \Theta\left(\frac{\az^2}{N^2} \log N\right)\,.\label{eq::raw_loss_showing_uniform_kover2}
\end{align}
when $f$ is the truncate compander.

We want to relate the raw loss in \eqref{eq::raw_loss_showing_uniform_kover2} to the expected loss $\cL_\az(P, \comp, N)$. This requires us to look at the normalization constant.

\begin{align}
    \bbE_{\bX \sim P}& \left[\log \left(\sum_{k = 1}^\az \tilde y_k \right) \right] 
    \\&= \bbE_{\bX \sim P} \left[\log \left(\sum_{k = 1}^\az \tilde y_k  - \sum_{k = 1}^\az  x_k + \sum_{k = 1}^\az x_k\right) \right]
    \\ & =  \bbE_{\bX \sim P} \left[\log \left(\sum_{k = 1}^\az \delta_k + 1\right) \right]
\end{align}
where $\delta_k = \tilde y_k - x_k$. We can bound
\begin{align}
    -\frac{1}{2N}\leq \delta_k \leq  \frac{1}{2N}
    \\ -\frac{\az}{2N}\leq \sum_{k = 1}^\az \delta_k \leq  \frac{\az}{2N}\,.
\end{align}
Additionally, we know that by construction,
\begin{align}
    \bbE_{\bX \sim p} \left[\sum_{k = 1}^\az \delta_k \right] = \sum_{k = 1}^\az (\tilde{y}_k - x_k) =  0
\end{align}
since $\tilde{y}_k$ is produced by the centroid decoder. Therefore, since $\log$ is concave, we have
\begin{align}
    \bbE_{\bX \sim P} &\left[\log \left(\sum_{k = 1}^\az \delta_k + 1\right) \right] 
    \\ &\geq \frac{1}{2} \left( \log \left(1 - \frac{\az}{2N}\right) + \log \left(1 + \frac{\az}{2N}\right) \right)
    \\ &\geq \frac{1}{2} \cdot 2 \cdot \frac{-(\az/(2N))^2}{2}
    \\ &= -\frac{1}{8} \az^2 N^{-2}
\end{align}
where the second inequality follows from the Taylor series of $\log(1+w)$. But this means that
\begin{align}
    -\bbE_{\bX \sim P} \left[\log \left(\sum_{k = 1}^\az \delta_k + 1\right) \right] = O\left(\frac{\az^2}{N^2}\right)
\end{align}
and hence by the proof of \Cref{lem::im-a-barby-girl}
\begin{align}
    \cL(P,f,&N) 
    \\ &= \widetilde{\cL}(P,f,N) + \bbE_{\bX \sim P} \left[\log \left(\sum_{k = 1}^\az \delta_k + 1\right) \right]
    \\ &=  \Theta\left(\frac{\az^2}{N^2} \log N\right) + O\left(\frac{\az^2}{N^2}\right)
    \\ &= \Theta\left(\frac{\az^2}{N^2} \log N\right)
\end{align}
since the extra $\log N$ factor causes the first term to dominate the second.
\end{proof}

The density $p(x)$ which produces \eqref{eq::raw_loss_showing_uniform_kover2} is not necessarily the worst possible density function in terms of the dependence of raw loss on the granularity $N$; however, it achieves simultaneously a worse-than-$\Theta(N^{-2})$ dependence on $N$ and a very large dependence on the alphabet size $\az$ (namely $\Theta(\az^2)$) with the uniform quantizer (i.e. truncation), and is therefore an ideal example of why the uniform quantizer is vulnerable to having poor performance.

For illustration, we will also sketch an analysis of the performance of the uniform prior against prior $p(x) = (1-\alpha) x^{-\alpha}$ where $\alpha = \frac{\az-2}{\az-1}$ (as mentioned in \Cref{rmk::loss_with_uniform}); this is constructed so that $\bbE_{X \sim p}[X] = 1/\az$ and hence $p \in \cP_{1/\az}$. The analysis shows that the loss is proportional to $N^{-(2-\alpha)}$.

Let $N$ be large; for this sketch we will treat $p$ as roughly uniform over any bin $I^{(n)} := ((n-1)/N, n/N]$. Note that this does not strictly hold for small $n$ (no matter how large $N$ gets, $p$ never becomes approximately uniform over e.g. $I^{(1)}$) but this inaccuracy is most pronounced on the first interval $I^{(1)} = (0,1/n]$. Additionally, $p$ on $(0,1/n]$ is a stretched and scaled version of $p$ on $(0,1]$; for $n = 2, 3, \dots, N$, the distribution $p$ over $I^{(n)}$ is closer to being uniform, and hence the distortion over any bin under $p$ can be bounded below (and above) by a constant multiple of the distortion under a uniform distribution (the constant can depend on $\az$ but not $N$). Thus for determining the dependence of the (raw) distortion on $N$, this simplification does not affect the result.

Then, the expected distortion given that $X \in I^{(n)}$ is proportional (roughly) to $N^{-2} (n/N)^{-1} = n^{-1} N^{-1}$ (since the interval has width $\propto \, N^{-1}$ and is centered at a point $\propto \, n/N$), and the probability of falling into $I^{(n)}$ is proportional to $(n/N)^{1-\alpha} - ((n-1)/N)^{1-\alpha} \approx n^{-\alpha} N^{-(1-\alpha)}$; therefore (up to a multiplicative factor which is constant in $N$) the expected distortion is roughly
\begin{align}
    \sum_{n=1}^N n^{-1} N^{-1} n^{-\alpha} N^{-(1-\alpha)} = N^{-(2-\alpha)} \sum_{n=1}^N n^{-(1+\alpha)} \,.
\end{align}
But, noting that $\sum n^{-(1+\alpha)}$ is a convergent series, we can apply an upper bound
\begin{align}
    \sum_{n=1}^N n^{-(1+\alpha)} < \sum_{n=1}^\infty n^{-(1+\alpha)}
\end{align}
which is a (finite) constant which depends only on $\az$ (through $\alpha$) but not $N$. Hence, we obtain our $\Theta(N^{-(2-\alpha)}) = \Theta(N^{-2} \cdot N^{\alpha})$ order for the distortion. We note that as discussed this is worse than $\Theta(N^{-2} \log N)$.

\section{Connection to Information Distillation Details}

%The asymptotic result in \Cref{thm::optimal_compander_loss} means that for large $M$, using $f_a$ for each $a$, the overall loss will be
%\begin{align}
%    \inf_h I(A;B)-I(A;\widetilde{B}) \lesssim M^{-\frac{2}{\az}} \sum_{a \in \cA} \widetilde{L}(p_a,f_a)
%    \\ = M^{-\frac{2}{\az}} \frac{1}{24} \Big(\int_0^1 (p_a(x)x^{-1})^{1/3} dx\Big)^3
%\end{align}
%(i.e. the right hand side is approximately the raw loss, and the normalized loss is smaller).

\label{sec::info_distillation_detail}

In this section, we go over the technical results connecting quantizing probabilities with KL divergence and information distillation (discussed in \Cref{sec::info_distillation_main}), in particular the proof of \Cref{prop:infodist}, which shows that information distillers and quantizers under KL divergence have a close connection.

In this section, we will use the notation $\widetilde{B}$ to denote $h(B)$. We denote by $P_A, P_B$ the marginals of $A$ and $B$ under the joint distribution $P_{A,B}$.

\subsection{Equivalent Instances of Information Distillation and Simplex Quantization}

%Recall that $h$ is information distiller for the set $\cB$.
%

We consider an information distillation instance, consisting of a joint probability distribution $P_{A,B}$ over $\cA \times \cB$ where $|\cA| = \az$ (and $\cB$ can be arbitrarily large or even uncountably infinite) and a number of labels $M$ to which we can distill; WLOG we will assume $\cA = [\az]$. The objective of information distillation is to find a distiller $h : \cB \to [M]$ which preserves as much mutual information with $A$ as possible, i.e. minimizes the loss
\begin{align}
    \idloss(P_{A,B}, h) := I(A;B) - I(A;\widetilde{B})
\end{align}
where $\widetilde{B} = h(B)$.\footnote{We do not include the parameter $M$ in the loss expression because it is already implicitly included as the range of the distiller $h$.} We denote an instance of the information distillation problem as $(P_{A,B}, M)_{\ID}$.

What is important about $b \in \cB$ for information distillation is what $B = b$ implies about $A$. We therefore denote by $\bx(b) \in \triangle_{\az-1}$ the conditional probability of $A$ given $B = b$, i.e.
\begin{align}
    x_a(b) = P_{A|B}(a|b) = \bbP[A = a \, | \, B = b] \,.
\end{align}

This then suggests a way to define the equivalent simplex quantization instance to a given information distillation instance. Recall that a simplex quantization instance (with average KL divergence loss) consists of a prior $P$ over $\triangle_{\az-1}$ and a number of quantization points $M$; the goal is to find a quantizer $\bnormvar : \triangle_{\az-1} \to \triangle_{\az-1}$ such that its range $\cZ$ has cardinality $M$ (or less) and which minimizes the expected KL divergence loss
\begin{align}
    \sqloss(P,\bnormvar) := \bbE_{\bX \sim P}[D_{\kl}(\bX \| \bz(\bX) )]\,.
\end{align}
We denote an instance of the simplex quantization problem (with average KL divergence loss) as $(P,M)_{\SQ}$.

\begin{definition}
    We call an information distillation instance $(P_{A,B},M)_{\ID}$ and a simplex quantization instance $(P,M)_{\SQ}$ \emph{equivalent} if they use the same value of $M$ and $P$ is the push-forward distribution induced by $\bx(\cdot)$ on $P_B$, i.e.
    \begin{align}
        B \sim P_B \implies \bX = \bx(B) \sim P\,.
    \end{align}
    We denote this $(P_{A,B},M)_{\ID} \equiv (P,M)_{\SQ}$.
\end{definition}

We show that any instance of one problem has at least one equivalent instance of the other.

\begin{lemma} \label{lem::equiv_instances_01}
    For any information distillation instance $(P_{A,B},M)_{\ID}$, there is some $(P,M)_{\SQ}$ such that $(P_{A,B},M)_{\ID} \equiv (P,M)_{\SQ}$ and vice versa.
\end{lemma}

\begin{proof}
    In either case, given the limit on the number of labels/quantization points $M$, we use it for the equivalent instance.

    Given an information distillation instance with joint distribution $P_{A,B}$, we have a well-defined function $\bx : \cB \to \triangle_{\az-1}$ and therefore the push-forward distribution $P$ of $P_B$ under $\bx(\cdot)$ is well-defined, giving us the equivalent instance $(P,M)_{\SQ}$.

    Given a simplex quantization instance with prior $P$, we let $\cB = \triangle_{\az-1}$ and let $P_{A,B} = P_{A|B} P_B$ given by $P_B = P$ (a probability distribution over $\triangle_{\az-1}$) and $P_{A|B}(a|b) = x_a(b)$, i.e. $A$ is distributed on $\cA = [\az]$ according to $B \in \triangle_{\az-1}$. Then $\bx(\cdot)$ is just the identity function and therefore $P = P_B$ is the push-forward distribution as we need.
\end{proof}
Note that each information distillation instance $(P_{A,B},M)_{\ID}$ has a unique equivalent simplex quantization instance (since $P$ is determined by being the push-forward distrbution of $P_B$), whereas each simplex quantization instance $(P,M)_{\SQ}$ may have many different equivalent information distillation instances, as $\cB$ can be arbitrarily large and elaborate.

The goal will be to show that if we have equivalent instances $(P_{A,B},M)_{\ID} \equiv (P,M)_{\SQ}$ then a distiller $h$ for $(P_{A,B},M)_{\ID}$ will have an `equivalent' quantizer $\bnormvar$ for $(P,M)_{\SQ}$ (achieving the same loss) and vice versa. This is generally achieved by the following scheme: we arbitrarily label the $M$ elements of $\cZ$ as $\bnormvar^{(j)}$ for $j \in [M]$, so
\begin{align}
    \cZ = \{\bnormvar^{(1)}, \dots, \bnormvar^{(M)}\} \, .
\end{align}
Then we will generally have equivalence between $h$ and $\bnormvar$ if the following relation holds:
\begin{align}
    \bnormvar(\bx(b)) = \bnormvar^{(h(b))} ~~\text{ for all } b \in \cB \,.
\end{align} 
Then we will derive
\begin{align}
    \idloss(P_{A,B}, h) = \sqloss(P, \bnormvar) \,.
\end{align}

However, as mentioned, this may be true (and/or possible) only if $h$ or $\bnormvar$ avoid certain trivial inefficiencies, hence the inequalities in \Cref{prop:infodist}. These will be formally defined and discussed in the following subsections.

\subsection{Separable Information Distillers}

We consider what happens when we have $b,b'$ such that $\bx(b) = \bx(b')$, i.e. $B = b$ and $B = b'$ induce the same conditional probability for $A$ over $\cA$. In this case, in the `equivalent' simplex quantization instance, the quantizer $\bz$ will quantize $\bx = \bx(b) = \bx(b')$ to a single value $\bnormvar^{(j)} \in \cZ$, while the distiller has the option of assigning $h(b) \neq h(b')$; if so, it is not clear what value the `equivalent' quantizer $\bz$ will assign to $\bx = \bx(b) = \bx(b')$. However, we will show that we can ignore such cases. We define:

\begin{definition} \label{def:separable-distillers}
    We call a quantizer $h$ \emph{separable} if for any $b,b' \in \cB$,
\begin{align}
    \bx(b) = \bx(b') \implies h(b) = h(b')
\end{align}
i.e. if $b$ and $b'$ induce the same conditional probability vector for $A$, they are assigned the same quantization label. 
\end{definition}
We call the set of information distillers $\cH$ and the set of separable information distillers $\cH_{\sf{sep}}$. %To find the best information distiller, we show that it is sufficient to consider separable information distillers:

Since the important attribute of any $b \in \cB$ (for information distillation) is how $B = b$ affects the distribution of $A$, there is no reason why $b, b' \in \cB$ should be assigned different labels by the distiller if $\bx(b) = \bx(b')$; thus, intuitively, it is clear that considering separable distillers is sufficient for discussing bounds the the performance of optimal distillers. We show this formally:

\begin{lemma}\label{lem::opt_H_sep}
For any $h \in \cH$ (inducing $\widetilde{B} = h(B)$), there is some $h^* \in \cH_{\sf{sep}}$ (inducing $\widetilde{B}^* = h^*(B)$) such that
\begin{align}
    I(A;\widetilde{B}) \leq I(A;\widetilde{B}^*)\,.
\end{align}
This then implies:
\begin{align}
    \sup_{h \in \cH} I(A;\widetilde{B}) = \sup_{h \in \cH_{\sf{sep}}} I(A;\widetilde{B})    \,.
\end{align}

\end{lemma}

\begin{proof}
This follows from the fact that it is optimal to only consider deterministic distiller (or quantization) functions, as shown in \cite{bhatt_distilling2021}. We may assume WLOG that $h \not \in \cH_{\sf{sep}}$.

%We then show that for any non-separable $h \in \cH \backslash \cH_{\sf{sep}}$ producing quantization $\widetilde{B}$, there is a $h^* \in \cH_{\sf{sep}}$ producing quantization $\widetilde{B}^*$ such that $I(A;\widetilde{B}) \leq I(A;\widetilde{B}^*)$.

First, note that $P_B$ induces a push-forward distribution $P$ over $\triangle_{\az-1}$ through $\bx(b)$. If $h \in \cH_{\sf{sep}}$, this means there is a deterministic $h_{\triangle} : \triangle_{\az-1} \to [M]$ satisfying
\begin{align}
    h(b) = h_\triangle(\bx(b)) \text{ for all } b \in \cB \,.
\end{align}
Then $I(A;h(B)) = I(A;h_{\triangle}(\bx(B)))$.

If $h \not \in \cH_{\sf{sep}}$, we still have a joint distribution $P_{\bx(B) \widetilde{B}}$; then we consider the conditional probability distribution $P_{\widetilde{B}|\bx(B)}(\widetilde{b}|\bx(b))$. This can be viewed as a \emph{non-deterministic} distiller $h_\triangle : \triangle_{\az-1} \to [M]$ (it returns a random output with distribution dependent on input $b$) under prior $P$, and similarly
\begin{align}
    I(A;h(B)) = I(A;h_{\triangle}(\bx(B))) 
\end{align}
since the joint distribution $P_{A\widetilde{B}}$ is the same either way. But by \cite{bhatt_distilling2021}, for $\bX \sim P$ over $\triangle_{\az-1}$ and any non-deterministic distiller $h_\triangle : \triangle_{\az-1} \to [M]$, there is a deterministic distiller $h^*_\triangle : \triangle_{\az-1} \to [M]$ such that
\begin{align}
    I(A;h_\triangle(\bX)) \leq I(A;h^*_{\triangle}(\bX)) \,.
\end{align}
Finally, any deterministic $h^*_{\triangle} : \triangle_{\az-1} \to [M]$ has an equivalent (separable) $h^* : \cB \to [M]$ such that $h^*(b) = h^*_\triangle(\bx(b))$ for all $b \in \cB$, simply by definition. Thus, for any non-separable $h \in \cH$, there is an equivalent non-deterministic distiller $h_\triangle$ for $\bX \sim P$; for every non-deterministic distiller $h_\triangle$ for $\bX \sim P$, there is a better deterministic distiller $h^*_\triangle$; and for every deterministic distiller $h^*_\triangle$ for $\bX \sim P$, there is an equivalent $h^* \in \cH_{\sf{sep}}$, i.e.
\begin{align}
    I(A;h(B)) &= I(A;h_\triangle(\bX)) 
    \\ \leq I(A;h^*_\triangle(\bX)) &= I(A;h^*(B))\,.
\end{align}
This then implies that
\begin{align}
    \sup_{h \in \cH} I(A;\widetilde{B}) \leq \sup_{h \in \cH_{\sf{sep}}} I(A;\widetilde{B})
\end{align}
while the fact that $\cH_{\sf{sep}} \subseteq \cH$ implies
\begin{align}
    \sup_{h \in \cH} I(A;\widetilde{B}) \geq \sup_{h \in \cH_{\sf{sep}}} I(A;\widetilde{B})
\end{align}
thus producing the equality we want
\end{proof}

This of course also implies that for any $h \in \cH$, there is some $h^* \in \cH_{\sf{sep}}$ such that
\begin{align}
    \idloss(P_{A,B},h) \geq \idloss(P_{A,B},h^*).
\end{align}
and furthermore that
\begin{align}
    \inf_{h \in \cH} \idloss(P_{A,B},h) = \inf_{h \in \cH_{\sf{sep}}} \idloss(P_{A,B},h).
\end{align}

%Since maximizing $I(A;\widetilde{B})$ is equivalent to minimizing $I(A;B)-I(A;\widetilde{B})$, we may consider only separable $h$.

\subsection{Decoding-Optimal Simplex Quantizers}

%We now consider the problem of average-loss simplex quantization with KL divergence loss, without restriction to companders or scalar quantization. 

We now consider simplex quantizers under average KL divergence loss. In particular, we note an obvious potential inefficiency: letting $\cZ = \{\bnormvar^{(1)}, \dots, \bnormvar^{(M)}\}$ be the range of quantizer $\bnormvar$, we define $\cX^{(j)} := \{\bx \in \triangle_{\az-1} : \bz(\bx) = \bz^{(j)}\}$ for all $j$; then, given $\cX^{(j)}$ there will be some optimal choice for the value of $\bnormvar^{(j)}$ which minimizes the expected KL divergence. If $\bz$ does not use the optimal value (which will turn out to be the conditional expectation e.g. centroid of $\cX^{(j)}$), for instance by using a value of $\bz^{(j)}$ which is completely unrelated to $\cX^{(j)}$, then there is an obvious and easily-fixed inefficiency.

%In this problem, a prior $P$ over $\triangle_{\az-1}$ (not necessarily symmetric or continuous) is given along with a number $M$ of quantization bins we can use. The goal is to find a quantizer $\bnormvar : \triangle_{\az-1} \to \triangle_{\az-1}$ which outputs at most $M$ distinct values to solve
%\begin{align}
    %\inf_{\bnormvar} \bbE_{\bX \sim P} [D_{\kl}(\bX \| \bnormVar)] \,.
%\end{align}

One way to frame this is by breaking the quantization process into two steps, an \emph{encoder} $g : \triangle_{\az-1} \to [M]$ and a \emph{decoder} $\dec : [M] \to \triangle_{\az-1}$ so that the quantization of $\bX$ is $\bnormVar = \bnormvar(\bX) = \dec(g(\bX))$; we WLOG label the elements of $\cZ$ such that $\bz^{(j)} = \dec(j)$. Then the encoder $g$ partitions $\triangle_{\az-1}$ into the $M$ `bins' (analogous to the compander bins) $\cX^{(1)}, \dots \cX^{(M)}$ (the same as defined above):
\begin{align}
    \cX^{(j)} = \{\bx \in \triangle_{\az-1} : g(\bx) = j\}\,.
\end{align}

\begin{lemma}\label{lem::centroid_optimal}
    Given encoder $g$ and prior $P$, the optimal decoder function (for $g$ on $P$) is
    \begin{align}
        \dec^*_g = \underset{\dec}{\argmin} \bbE_{\bX \sim P} [D_{\kl}(\bX \| \dec(g(\bX)))]
    \end{align}
    satisfies, for all $j \in [M]$,
    \begin{align}
        \dec^*_g(j) = \bbE_{\bX \sim P}[\bX \, | \, \bX \in \cX^{(j)}]\,.
    \end{align}
    We call any quantizer consisting of an encoder $g$ and the optimal decoder function $\dec^*_g$ \emph{decoding-optimal}. This implies that for any quantizer $\bnormvar$ on prior $P$, there is a decoding-optimal $\bnormvar^*$ such that
    \begin{align}
        \sqloss(P,\bnormvar^*) \leq \sqloss(P,\bnormvar) \,.
    \end{align}
\end{lemma}

\begin{proof}
This is proved by \cite[Corollary 4.2]{ITbook}.
\end{proof}

Note that the optimal $\dec^*_g(j)$ is the centroid (conditional expectation under $P$) of the bin $\cX^{(j)}$ induced by $g$. %Thus, to find the KL divergence-minimizing quantizer, we can consider only the problem of finding an optimal encoder $g$ and assume the centroid decoder. %We refer to any quantizer $\bz$ which uses an encoder $g$ and decoder $\dec^*_g$ as \emph{decoding-optimal}.

%\footnote{\nbyp{Example of a useless footnote.} In practice the centroid decoder $\dec^*_g$ may be too complex to use, but it minimizes the KL divergence loss.}

\subsection{Deriving the Connection}

We now prove \Cref{prop:infodist}. We first define what it means for a distiller and a quantizer to be equivalent:
\begin{definition}
    If we have equivalent information distillation and simplex quantization problems $(P_{A,B}, M)_{\ID} \equiv (P,M)_{\SQ}$, then the distiller $h$ and quantizer $\bnormvar$ are \emph{equivalent} for these problems if:
    \begin{itemize}
        \item $h$ is separable and $\bnormvar$ is decoding-optimal;
        \item there is a labeling $\bz^{(1)}, \dots, \bz^{(M)}$ of the elements of $\cZ$ such that $\bnormvar(\bx(b)) = \bnormvar^{(h(b))}$ for all $b \in \cB$.
    \end{itemize}
    We denote this as $h \equiv \bnormvar$.
\end{definition}

We then claim that all separable distillers and decoding-optimal quantizers have equivalent counterparts:
\begin{lemma} \label{lem::equiv_instances_02}
    For any $(P_{A,B}, M)_{\ID} \equiv (P,M)_{\SQ}$, any separable $h$ for $(P_{A,B}, M)_{\ID}$ has an equivalent (decoding-optimal) $\bnormvar$, and any decoding-optimal $\bnormvar$ for $(P,M)_{\SQ}$ has an equivalent (separable) $h$.
\end{lemma}

\begin{proof}
    We handle the two directions separately:

    \emph{Any $h$ has an equivalent $\bnormvar$:} Since $h$ is separable, we know that $\bx(b) = \bx(b') \implies h(b) = h(b')$. Thus, we can define $\cX^{(j)}$ as
    \begin{align}
        \cX^{(j)} := \{ \bx \in \triangle_{\az-1} : h(b) = j ~ \forall b \text{ s.t. } \bx(b) = \bx\}
    \end{align}
    for all $j \in [M]$.
    Then we define $\bnormvar$ as follows: $\bnormvar(\bx) = \bnormvar^{(j)}$ for all $\bx \in \cX^{(j)}$, where
    \begin{align}
        \bnormvar^{(j)} = \bbE_{\bX \sim P}[\bX \, | \, \bX \in \cX^{(j)}] \,.
    \end{align}
    Then by construction of $\bnormvar^{(j)}$ we have that $\bnormvar$ is decoding-optimal and for $\bx \in \cX^{(j)}$ we have $h(b) = j$ for all $b$ such that $\bx(b) = \bx$ and $\bnormvar(\bx) = \bnormvar^{(j)}$, hence $\bnormvar(\bx) = \bnormvar^{(h(b))}$, so they are equivalent.

    \emph{Any $\bnormvar$ has an equivalent $h$:} We label the elements of $\cZ$ arbitrarily as $\bz^{(1)}, \dots, \bz^{(M)}$; then we let $h(b) = j$ for all $b$ such that $\bnormvar(\bx(b)) = \bnormvar^{(j)}$, which implies $\bnormvar(\bx(b)) = \bnormvar^{(h(b))}$.
\end{proof}

Now we show that equivalent solutions have the same loss:
\begin{proposition}\label{prop::KL_and_distill_equal}
If $(P_{A,B},M)_{\ID} \equiv (P,M)_{\SQ}$ and $h \equiv \bnormvar$, then
\begin{align}
      \idloss(P_{A,B},h) = \sqloss(P,\bnormvar) \label{eq::kl_and_MI}\,.
\end{align}
\end{proposition}
%\nbyp{I think this proposition will be modified in view of Prop.~\ref{prop:infodist}. But in any case, I wonder why is $\sup$ taken over symmetric priors suddenly?} -- \textcolor{blue}{AVIV: We don't have notation for nonsymmetric priors. I just left it as $P$.}
\begin{proof}
    Let $(A,B) \sim P_{A,B}$ and let $\bX = \bx(B)$ and $\bnormVar = \bnormvar(\bX)$. Then we know since $(P_{A,B},M)_{\ID} \equiv (P,M)_{\SQ}$ that $\bX \sim P$. Furthermore, defining 
    \begin{align}
        \cX^{(j)} = \{\bx \in \triangle_{\az-1} : h(b) = j ~\forall b \text{ s.t. } \bx(b) = \bx\}
    \end{align}
    and $\bnormvar^{(j)} = \bbE[\bX \, | \, \bX \in \cX^{(j)}]$, we know that since $h \equiv \bnormvar$ we have $\bnormVar = \bnormvar^{(h(B))}$. We now let $\normVar_i$ refer to the $i$th element of vector $\bnormVar$, and let $\widetilde{B} = h(B)$ and $\widetilde{b} = h(b)$. We then derive:
    \begin{align}
        \idloss(P_{A,B},h) &= I(A;B) - I(A;\widetilde{B})
        \\& = \int   \sum_{a} P_{A,B}(a ,b)   \log \frac{P_{A|B}(a |b)}{P_A(a)} db \eqlinebreakshort - \sum_{a ,\widetilde b} P_{A, \widetilde B}(a ,\widetilde b)   \log \frac{P_{A|\widetilde B}(a |\widetilde b) }{P_A(a)} %\label{eq::sum_marginals}
        \\& = \int   \sum_{a} P_{A,B}(a ,b)   \log \frac{P_{A|B}(a |b)}{P_A(a)}  \eqlinebreakshort - P_{A,B}(a ,b)   \log \frac{P_{A|\widetilde B}(a |\widetilde b) }{P_A(a)} db  
        \\& = \int   \sum_{a} P_{A,B}(a ,b)   \log \frac{P_{A|B}(a |b)}{P_{A|\widetilde B}(a |\widetilde b) }  db 
        \\& = \int P_B(b)  \sum_{a} P_{A|B}(a|b)   \log \frac{P_{A|B}(a |b)}{P_{A|\widetilde B}(a |\widetilde b) }  db 
        \\& = \bbE_{B} \bigg[\sum_{a} P_{A|B}(a|b)   \log \frac{P_{A|B}(a |b)}{P_{A|\widetilde B}(a |\widetilde b) } \bigg]
        \\& = \bbE_{B} \big[D_{\kl}((A|B) \| (A|\widetilde{B})) \big]
        \\& = \bbE_{\bX} \big[D_{\kl}(\bX \| \bnormVar) \big] \label{eq:to-sqloss}
        \\& = \sqloss(P, \bnormvar)
    \end{align}
    where \eqref{eq:to-sqloss} holds as $B \sim P_B \implies \bX \sim P$ and 
    \begin{align}
        \bnormVar = \bnormvar(\bX) = \bnormvar^{(\widetilde{B})} = \bbE_{\bX \sim P}[\bX \, | \, \bX \in \cX^{(\widetilde{B})}]
    \end{align}
    and since $A \sim \bX = \bX(B)$, we know that $P_{A|\widetilde{B}}(a|\widetilde{b}) = \bbE_{\bX \sim P}[X_a \, | \, \bX \in \cX^{(\widetilde{b})}]$.
\end{proof}

\begin{proof}[Proof of \Cref{prop:infodist}]

We get the proof of \Cref{prop:infodist} as a corollary to \Cref{prop::KL_and_distill_equal} and \Cref{lem::opt_H_sep,lem::centroid_optimal,lem::equiv_instances_02} (which show, respectively, that non-separable distillers can be replaced by separable distillers, that non-decoding-optimal quantizers can be replaced by decoding-optimal quantizers, and that any separable distiller has an equivalent decoding-optimal quantizer and vice versa). 

    Note that \Cref{prop:infodist} ensures $(P_{A,B},M)_{\ID} \equiv (P,M)_{\SQ}$ through its definition of $\bX$.

    Then, given a distiller $h \in \cH$, by \Cref{lem::opt_H_sep} we can find a separable $h^* \in \cH_{\sf{sep}}$ such that
    \begin{align}
        \idloss(P_{A,B}, h^*) \leq \idloss(P_{A,B}, h) \,.
    \end{align}
    By \Cref{prop::KL_and_distill_equal}, there is a quantizer $\bnormvar$ such that
    \begin{align}
        \sqloss(P, \bnormvar) \leq \idloss(P_{A,B}, h^*) \leq \idloss(P_{A,B}, h)  \,.
    \end{align}
    completing the result in the first direction.

    Given a quantizer $\bnormvar$, by \Cref{lem::centroid_optimal} there exists a decoding-optimal $\bnormvar^*$ such that
    \begin{align}
        \sqloss(P, \bnormvar^*) \leq \idloss(P, \bnormvar) \,.
    \end{align}
    By \Cref{prop::KL_and_distill_equal}, there is a distiller $h$ such that
    \begin{align}
        \idloss(P_{A,B}, h) \leq \sqloss(P, \bnormvar^*) \leq \sqloss(P, \bnormvar)  \,.
    \end{align}
    completing the result in the second direction.
\end{proof}

Now that we have shown \Cref{prop:infodist}, we can use it to derive the connection between the performance of our companders and the Degrading Cost $\mathrm{DC}$:

\iffalse
\begin{proposition}\label{prop::infsup_for_dc}
If $P$ is the push-forward distribution of $P_B$ through $\bx(\cdot)$,
\begin{align}
    %\hspace{-0.7pc}  \inf_{h} \hspace{-0.1pc} I(A;B) \hspace{-0.2pc} - \hspace{-0.2pc} I(A; h(B)) \hspace{-0.2pc} = \hspace{-0.2pc} \inf_{g} \hspace{-0.4pc} \underset{\bX \in P}{\bbE} \left[D_{\kl}(\bX \| \bnormVar) \right]  
    \inf_{h} \idloss(P_{A,B}, h) = \inf_{\bnormvar}\sqloss(P, \bnormvar) 
    \label{eq::infg_kl_and_MI}
\end{align}
and thus, for any $M,\az$, 
\begin{align}
%    \sup_{\substack{{P_{A,B}} \\ |\cA| = \az}} & \inf_{h}  I(A;B) - I(A; h(B)) 
%    \\&=  \sup_{P}  \inf_{g} \bbE_{\bX \in P} \left[D_{\kl}(\bX \| \bnormVar) \right]
\sup_{\substack{{P_{A,B}} \\ |\cA| = \az}} & \inf_{h} \idloss(P_{A,B}, h) = \sup_{P}  \inf_{\bnormvar}\sqloss(P, \bnormvar) \label{eq::supP_kl_and_MI}  
\end{align}
\end{proposition}
\fi

\begin{proposition}\label{prop::infsup_for_dc}
For any $\az, M$: 
\begin{align}
%    \sup_{\substack{{P_{A,B}} \\ |\cA| = \az}} & \inf_{h}  I(A;B) - I(A; h(B)) 
%    \\&=  \sup_{P}  \inf_{g} \bbE_{\bX \in P} \left[D_{\kl}(\bX \| \bnormVar) \right]
\mathrm{DC}(\az,M) = \sup_{P \text{ over } \triangle_{\az-1}} \inf_{\substack{\bnormvar \\ |\cZ| = M}}  \sqloss(P,\bz) \label{eq::supP_kl_and_MI}  \,.
\end{align}
\end{proposition}

%\bbE_{\bX \sim P}[D_{\kl}(\bX \| \bz(\bX))]

\begin{proof}

We show inequalities in both directions to get the equality.

First, note that for any joint distribution $P_{A,B}$ on $\cA \times \cB$ where $|\cA| = \az$ (WLOG we can assume $\cA = [\az]$), we know there is some prior $P$ over $\triangle_{\az-1}$ such that
\begin{align}
    (P_{A,B}, M)_{\ID} \equiv (P,M)_{\SQ}
\end{align}
for all $M$, by \Cref{lem::equiv_instances_01}, and that for any distiller $h : \cB \to M$ there is some quantizer $\bnormvar$ with cardinality-$M$ range such that
\begin{align}
    \sqloss(P, \bnormvar) \leq \idloss(P_{A,B}, h)
\end{align}
by \Cref{lem::equiv_instances_02} and \Cref{prop::KL_and_distill_equal}. Thus for any $P_{A,B}$ and $M$, for the equivalent $P$,
\begin{align}
    \inf_{h:\cB \to M} \idloss(P_{A,B}, h) \geq \inf_{\substack{\bnormvar \\ |\cZ| = M}}  \sqloss(P,\bz)
\end{align}
and hence we have
\begin{align}
    \mathrm{DC}(\az,M) &= \sup_{\substack{{P_{A,B}} \\ |\cA| = \az}} \inf_{h:\cB \to M} \idloss(P_{A,B}, h)
    \\ &\geq \sup_{P \text{ over } \triangle_{\az-1}} \inf_{\substack{\bnormvar \\ |\cZ| = M}}  \sqloss(P,\bz)\,.
\end{align}

Then, for any $P$ over $\triangle_{\az-1}$, we have the same logic: by \Cref{lem::equiv_instances_01} there is an equivalent $P_{A,B}$, so for any $P,M$ we can find $P_{A,B}$ for which
\begin{align}
     \inf_{\substack{\bnormvar \\ |\cZ| = M}}  \sqloss(P,\bz) \geq \inf_{h:\cB \to M} \idloss(P_{A,B}, h)\,.
\end{align}
Then we get that
\begin{align}
    \mathrm{DC}(\az,M) &= \sup_{\substack{{P_{A,B}} \\ |\cA| = \az}} \inf_{h:\cB \to M} \idloss(P_{A,B}, h)
    \\ &\leq \sup_{P \text{ over } \triangle_{\az-1}} \inf_{\substack{\bnormvar \\ |\cZ| = M}}  \sqloss(P,\bz)
\end{align}
and hence the equality in \eqref{eq::supP_kl_and_MI} holds.
\end{proof}

\Cref{prop::infsup_for_dc} is used to show \eqref{eq::infg_supP_KL}. %Note that in \eqref{eq::infg_supP_KL_symPrior}, we can change the supremum over $P_{A,B}$ to a supremum over symmetric priors because for any prior, there always exists a symmetric priors that always perform the same or worse. 

%The theoretical bounds in \Cref{thm::optimal_compander_loss} imply that a quantizer $h$ for $\cB$ using compander $f = (f_1,\dots, f_{\az})$ has performance
%\begin{align}
    %\lim_{M \to \infty} M^{2/\az} (I(A; B) - I(A; \tilde B))  \leq \sum_{a\in \cA} \singleloss(p_a, f_a)\,.
%\end{align}

%Note that since we use minimax companders to achieve this result, it is achieved by a fixed quantizer
%\nbyp{Remove the thing below, or move to appendix. If you choose to keep, compare against~\eqref{eq::distilling_result2} which beats previous SOTA.}

\subsection{Comparison}
%\jennifer{moved section}
Compared to \eqref{eq::distilling_result}, our bound in \Cref{prop::compander_worst_case_for_distilling} which uses the approximate minimax compander has a worse dependence 
%\footnote{\nbyp{Example of footnote  that should be in the main text.} \textcolor{blue}{AVIV: Done.}} 
on $M$. Our dependence on $M$ is worse since our compander method performs scalar quantization on each entry, and the raw quantized values do not necessarily add up to $1$. Other quantization schemes can rely on the fact that the values add up to $1$ to avoid encoding one of the $\az$ values. 
%This is because quantizing the simplex $\triangle_{\az-1}$ is really a $(\az-1)$-dimensional problem, but scalar quantization treats it as $\az$-dimensional. Thus, scalar quantizers (like companders) cover an extra dimension. Intuitively, given $M$ bins, each dimension of $\triangle_{\az-1}$ can be covered to $M^{\frac{1}{\az-1}}$ different levels, while scalar quantizers only have $M^{\frac{1}{\az}}$ per dimension.\footnote{While the extra dimension is eliminated by normalization, the quantization function does not take it into account, and the efficiency of the reconstruction points may not benefit. Additionally, with companders, a large number of potential quantization combinations (those whose raw values sum far from $1$) may even be impossible.}
%Our results suggest that if each dimension has a density of quantization points of roughly $N$ (i.e. if $\bnormvar = \bnormvar(\bx)$, $|\normvar_i - x_i| \propto N^{-1}$) and they are well-chosen, the expected KL divergence is proportional to $N^{-2}$; applying this to the question of scalar versus vector quantization predicts the (multiplicative) difference in dependence on $M$ to be
%\begin{align}
    %M^{-\frac{2}{\az}} / M^{-\frac{2}{\az-1}} = M^{\frac{2}{\az(\az-1)}}
%\end{align}
%as shown by comparing \eqref{eq::distilling_result} and \eqref{eq::compander_worst_case_for_distilling} (letting $M = N^\az \implies N^{-2} = M^{-2/\az}$). Although this is not a rigorous argument, it gives the intuition for the difference between scalar quantization and vector quantization methods. 
Offsetting this are the improved dependence on $\az$ ($\log^2 \az$ versus $\az-1$, as stated) and constant ($\leq 19$ and decreasing to $1$ as $\az \to \infty$ versus $1268$); this yields a better bound when $M$ is not exceptionally large.
For example, when $\az = 10$, our bound is better than \eqref{eq::distilling_result} so long as the conditions on $M^{1/\az}$ in \Cref{prop::compander_worst_case_for_distilling}
are met (which requires $M > 16^{10}$) and if $M < 1.014 \times 10^{97}$. While these may both seem like very large numbers, the former corresponds with only $4$ bits to express each value in the probability vector, while the latter corresponds with more than $32$ bits per value. In general, the `crossing point' (at which both bounds give the same result) is at
\begin{align}
    M = \Bigg(1268 \bigg( \hspace{-0.2pc} 1 \hspace{-0.2pc} + \hspace{-0.2pc} 18 \frac{\log \log \az}{\log \az} \hspace{-0.1pc} \bigg)^{-1} \frac{\az-1}{\log^2 \az}\Bigg)^{\frac{\az(\az-1)}{2}}
\end{align}
or, to put it in terms of `bits per vector entry' $b$ (taking $\log_2$ of the above to get bits and dividing by $\az$),
\begin{align}
    b \approx \frac{\az-1}{2} \bigg(\log_2(\az) - 2 \log_2 \log \az + 10.3\bigg)
\end{align}
for large $\az$.
The disadvantage is that our bound does not apply to the case of $\az < 5$ or $M$ which is not large. Note that scalar quantization in general only works with very large $M$, since even $2$ different encoded values per symbol requires $M = 2^\az$ different quantization values.

%If $P_{A,B}$ is fixed and known, then it is possible to use the compander $f$ prescribed in \Cref{thm::optimal_compander_loss} as the quantization function $h$. To get the best bound, we consider a separate compander function $f_a$ for each $a \in \cA$. To apply companders, it is necessary to determine $P_{A|B}(a | b)$ for each $a$ and $b$.
%We let $p_a$ be defined as the distribution of $P_{A|B}(a | B)$ where random variable $B$ is drawn using $P_B$.\footnote{For each $B$, $P_{A|B}(a|B)$ is a probability value in $[0,1]$; hence we can view $P_{A|B}(a|B)$, for any fixed $a$, as a random variable.} 
%Each $f_a$ is computed by using $p_a$ in \eqref{eq::best_f_raw}. Each $b \in \cB$ will be quantized to a $\az$ dimensional vector, where we apply $f_a$ on $P_{A|B}(a | b)$ for each $a$. This $\az$ dimensional vector can then be translated to a value in $[M]$. If $P_{A,B}$ is not known, we use the minimax (or approximate minimax) companders instead.

\end{document}